\documentclass[a4paper, smallextended]{svjour3}

\usepackage{bbding}
\usepackage[T1]{fontenc}
\usepackage{aurical}
\usepackage{huncial}
\usepackage{linearb}
\usepackage{textcomp}
\usepackage{amsmath}
\usepackage{amssymb}
\usepackage{bbm}
\usepackage{amsfonts}
\usepackage{amssymb,epic,eepic}
\usepackage{stmaryrd}
\usepackage{capt-of}
\usepackage[usenames,dvipsnames]{pstricks}
\usepackage{epsfig}
\usepackage{pst-grad}
 \usepackage{pst-plot}
\usepackage{graphicx}
\usepackage{here} 
\usepackage{hieroglf}
\usepackage{phaistos}
\usepackage{staves}
\smartqed

\begin{document}
\title{Classical-Quantum Arbitrarily Varying Wiretap Channel:
 Common Randomness Assisted Code and
Continuity }
\author{Holger Boche \and Minglai Cai \and Christian Deppe\and Janis N\"otzel}
\institute{ Holger Boche \and Minglai Cai \and Christian Deppe \at Lehrstuhl f\"ur Theoretische
Informationstechnik,\\ Technische Universit\"at M\"unchen,\\
Munich, Germany\\
\email{\{boche, minglai.cai, christian.deppe\}@tum.de}
 \and Janis N\"otzel \at
 Universitat Aut\`{o}noma de Barcelona,\\
 Barcelona, Spain\\
\email{Janis.Notzel@uab.cat}}
\titlerunning{Classical-Quantum Arbitrarily Varying Wiretap Channel}

\maketitle
\begin{abstract}
We determine the secrecy capacities under common
randomness assisted coding of arbitrarily varying classical-quantum
wiretap  channels. 
Furthermore, we determine   the secrecy capacity of a mixed channel model which is compound 
from the sender to the legitimate receiver and varies arbitrarily from the sender to the
eavesdropper. 
We 
examine when the secrecy
capacity is a continuous function of the system parameters as an application
and show that
resources, e.g.,
having access to a perfect copy of the outcome of a random experiment,  can guarantee continuity of 
the capacity function of arbitrarily varying classical-quantum wiretap channels.
\end{abstract}

\tableofcontents
\section{Introduction}

In the last few years,
the developments in modern communication systems have produced 
many results in a short amount of time.
 Quantum communication systems, especially,
have developed into a very active
field, setting new properties and limits.
 Our goal is to deliver a  general theory considering both channel robustness against
 jamming
and security against  eavesdropping
in quantum information theory, since
many modern communication systems are
often not perfect, but are vulnerable to jamming and eavesdropping.
The transmitters have to 
solve two main problems. First, the message (a secret key
or a secure message)
has to be encoded robustly, i.e.,
 despite
channel uncertainty, it
 can be decoded correctly by the legitimate receiver.
Second, the message has to be encoded in such a way that
the wiretapper's knowledge of the 
transmitted classical message can be 
kept arbitrarily small. This work is an extension of our
previous paper 
\cite{Bo/Ca/De}.

In our earlier work \cite{Bo/Ca/De}, we  investigated the transmission of messages from a sending  party to a receiving party.
 The messages were kept secret from an eavesdropper. Communication took place over a quantum channel which was, 
in addition to noise from the environment, subjected to the action of a jammer which actively manipulated the states.
The Ahlswede Dichotomy for arbitrarily varying classical-quantum
wiretap  channels has been established, i.e. either the deterministic 
capacity of an arbitrarily varying channel was zero or equal to its shared randomness assisted capacity.
 We also analyzed the secrecy capacity of  arbitrarily
varying classical-quantum wiretap  channels when the sender and the
receiver used various resources and  studied the helpfulness of
certain resources for robust and secure information transmission. 
We found out that even using the weakest non-secure resource (the correlation),
one could achieve the same security capacity using a strong resource
as the common randomness. But, nonetheless, 
a capacity formula was not  given in \cite{Bo/Ca/De}.

In this paper, we carry on our investigation of arbitrarily varying classical-quantum
wiretap  channels and shared randomness. We  
deliver a capacity formula for secure information
transmission through an arbitrarily varying classical-quantum
wiretap    channel using correlation as a resource. Together
with the result of \cite{Bo/Ca/De}, it yields a formula for
deterministic secrecy capacity of the arbitrarily varying classical-quantum
wiretap    channel. Using this  formula, we   analyze the stability of 
secrecy capacity, i.e., we ask under which condition, it is discontinuous as a
function of channel parameters, in other words, when 
small variations in the underlying model dramatically  change the effect of the jammer's actions.

To  determine our capacity formula, we follow the idea of 
 \cite{Bj/Bo/Ja/No} and \cite{Wi/No/Bo} in the classical cases:
At first, we consider a mixed channel model that is called the
arbitrarily varying classical-quantum wiretap channel. Then, we
apply Ahlswede's robustification technique to establish
 the common randomness assisted secrecy capacity of an arbitrarily varying classical-quantum
wiretap    channel.
\vspace{0.2cm}

Quantum mechanics differs significantly from classical mechanics;
it has its own laws.
A quantum
channel is  a communication  channel which can transmit
quantum information. In this paper, we consider  the
classical-quantum channels, i.e., the sender's inputs are classical data
and the receiver's outputs are quantum systems.
The capacity of
  classical-quantum channels has been determined  in
\cite{Ho} and \cite{Sch/Wes}.

In  the model of an \emph{arbitrarily varying channel}, we consider channel
uncertainty, i.e. 
transmission over a channel  which is not stationary, but can change  with every use
of the channel. We  interpret it as a channel with a
jammer who may change his input with every channel use and is not restricted to using a repetitive probabilistic strategy. 
It is understood that the sender and the receiver have to select their
coding scheme first. After that, the jammer makes his choice of the channel state to sabotage the message  transmission.
However, due to the physical properties, we consider
that  the jammer's changes only take  place in a known set. 
The arbitrarily varying channel was first introduced
 in \cite{Bl/Br/Th2}.

As was already mentioned in our earlier work \cite{Bo/Ca/De}, we are interested in the role that 
 shared randomness plays for the arbitrarily 
varying classical-quantum wiretap channel.
This is used in  \cite{Ahl1}, \cite{Ahl2}, and \cite{Ahl3} for the determination of 
the random capacity. 
 \cite{Ahl1}  showed a surprising result
which is now known as the Ahlswede Dichotomy:
Either the 
capacity of an arbitrarily varying channel is zero, or it equals its shared randomness assisted capacity. 
After this discovery, it has remained an open question as to exactly when the deterministic capacity is positive. 
In \cite{Rei},  a sufficient condition for this has been given, and in \cite{Cs/Na} it is
 proved that this condition is also necessary.
In \cite{Ahl0} it has also been shown
  that the capacity of certain arbitrarily varying channels can be equated to the zero-error 
capacity of related discrete memoryless channels.
The  Ahlswede Dichotomy demonstrates the importance of shared randomness for communication in a very clear
 form.

 A
classical-quantum  channel with a jammer is called an arbitrarily varying
classical-quantum  channel. The arbitrarily varying classical-quantum
 channel was introduced in \cite{Ahl/Bli}.  A lower bound for its capacity
has been given. An alternative proof 
and a proof of the
strong converse are given in \cite{Bj/Bo/Ja/No}.  In \cite{Ahl/Bj/Bo/No},
the Ahlswede Dichotomy for the  arbitrarily varying classical-quantum
 channels is established, and a  sufficient and  necessary  condition
for  the  zero deterministic capacity is given.
 In \cite{Bo/No}, a simplification of this  condition
 for the  arbitrarily varying classical-quantum
 channels is given.

In  the
model of a wiretap channel,
 we consider communication
with security. This was first introduced   in \cite{Wyn}
(in this paper, we will use a stronger security criterion than  \cite{Wyn}'s
 security criterion, cf. Remark \ref{remsc}). We
 interpret the  wiretap channel as a channel with an eavesdropper.
The relation of the different security criteria is discussed, for example, in 
\cite{Bl/La} with some generality and in \cite{Wi/No/Bo} with respect to arbitrarily varying channels.

A classical-quantum  channel with an eavesdropper is called a
classical-quantum wiretap channel, its secrecy capacity has been
determined in \cite{De} and \cite{Ca/Wi/Ye}.

In  the model of an 
arbitrarily varying wiretap channel,  we consider  transmission with both a
jammer and an eavesdropper.   Its secrecy capacity has been analyzed
in \cite{Bj/Bo/So2}. A lower bound of the  randomness  assisted secrecy capacity
has been given.

 A
classical-quantum  channel with both a jammer and an eavesdropper is
called an arbitrarily varying classical-quantum  wiretap  channel. It is defined as a family of pairs of
 indexed channels $\{(W_t,V_t) :t=1,\cdots,T\}$  with a common input alphabet
and possible different output alphabets
and connects a sender with
two receivers, a legitimate one and a wiretapper, where $t$ is called a
channel state of the channel pair. The  legitimate receiver accesses
the output of the first part of the pair, i.e.,  the first channel $W_t$  in the pair, and
the wiretapper observes the output of the second part, i.e.,  
the second channel $V_t$, respectively. A channel state $t$, which
varies from symbol to symbol in an  arbitrary manner, governs both
the legitimate receiver's channel and the wiretap  channel. A code for
the channel conveys information to the legitimate receiver such that the
wiretapper knows nothing about the transmitted information in the sense of 
the stronger security criterion (cf. Remark \ref{remsc}). This is
a generalization of compound classical-quantum wiretap
channels in \cite{Bo/Ca/Ca/De},
 when the channel states
are not stationary, but can change over   time.

The secrecy
 capacity   of the  arbitrarily varying classical-quantum  wiretap  channels
has been analyzed in \cite{Bl/Ca}. A lower bound of the  randomness  assisted capacity
has been given, and it has been shown that this bound is  either a lower  bound for
the deterministic capacity, or else the
deterministic capacity is equal to zero.

References \cite{Be/Br} and \cite{Be} are two well-known examples
for
secure
 quantum
information transmission  using quantum key distributions.
Good one-shot results for
quantum channels  with a
wiretapper who is limited in his actions  have been obtained.
But our goal is   to have a more general theory for
channel   security in  quantum information
theory, i.e.,  message transmission should be secure against every
possible kind of  eavesdropping. Furthermore, we are
interested in  asymptotic behavior when we  deliver a large
volume of messages by many channel uses.
Therefore, we consider a new paradigm for the design of quantum
channel systems, which is called \it embedded security\rm. Instead of 
the standard approach in secret communication, i.e. first ensuring 
a successful transmission of messages and then implementing
a cryptographic protocol, here we
 embed protocols with a guaranteed security right from the start
into  the
 physical layer, which is the bottom layer
of the model of communications systems. The concept  covers 
both secure message transmission
and secure key generation.

In \cite{Bo/No}, a classification of various resources is given. A distinction is made between  two extremal cases:
randomness and correlation. Randomness is the strongest resource, and it
requires a perfect copy of the
outcome of a random experiment, and thus, we should assume
an additional perfect channel.
On the other hand, correlation is the weakest resource.  
The work \cite{Bo/No} also puts emphasis on the quantification of the differences between 
correlation and common randomness and used the arbitrarily varying classical-quantum channel as a method of proof.
It can be shown that common randomness is a stronger resource than  correlation in the following sense:
An example is given where not even a finite amount of common randomness can be extracted from a given  correlation.
On the contrary,  a sufficiently large amount of common randomness  allows the sender and receiver to 
asymptotically simulate the statistics of any correlation.



In view of the aforementioned importance of shared randomness for robustness,
it is clear that the shared randomness is not allowed to be known by
the jammer (In stark contrast to this, we assume the eavesdropper has access 
to the outcomes of the shared random experiment). Therefore, backward communication from the eavesdropper to the jammer 
 would render the shared randomness completely useless. Thus we concentrate our
analysis on the case without feedback, i.e.  the eavesdropper cannot
send messages toward the jammer. The communication
from the jammer to the eavesdropper is explicitly possible, i.e.  the
eavesdropper could know the jammer's strategy. It is a challenging task for
future studies when the resource is secure against eavesdropping and
two-way communication between the jammer and the eavesdropper is allowed. In
this case, we have to build a code in such a way that the transmission of both
the message and the randomization is secure.

As an application of our results, we turn to the
question: when
the secrecy capacity is a continuous function of the system parameters?
The analysis of the continuity of  capacities of  quantum channels is
raised from the question 
whether small changes in the channel system are able to cause dramatic losses in the performance.
The continuity of the message and entanglement transmission capacity of a stationary memoryless quantum channel has been
listed as an open problem in \cite{WynSite} and was solved in \cite{Le/Sm}. 
 Considering   channels with active jamming faces an especially  new difficulty.
 The reason is that
 the  capacity in this case is, in
general, not specified by entropy quantities.
In \cite{Bo/No2} it has been shown when
the message transmission capacity of an
 arbitrarily varying quantum channels is continuous.
The condition for continuity of message transmission capacity of a classical 
 arbitrarily varying  wiretap channel has been given 
in \cite{Wi/No/Bo}. 

As a direct consequence of our capacity formula,
we show in this paper that a sharing resource is very helpful for the channel
stability in the sense that it  provides continuity of secrecy capacities.
 \vspace{0.15cm}

 This
paper is organized as follows:\vspace{0.15cm}

The main definitions  are given in
Section \ref{secprem}.

In Section  \ref{CAVWCQC} we determine a capacity
formula for a mixed channel model, i.e.
the enhanced secrecy capacity of  compound-arbitrarily varying wiretap classical-quantum channels.
This formula
will be used for
our result in Section  \ref{SCoavcqwc}.

In Section  \ref{SCoavcqwc} our main result
is presented. In this section we determine the secrecy capacities under common
randomness assisted coding of arbitrarily varying classical-quantum
wiretap  channels.

As an application of our main result, in Section  \ref{iosccits}
we discuss when the secrecy capacity of an arbitrarily varying
classical-quantum wiretap channel is a  continuous quantity of the
system parameters.

\section{Preliminaries}\label{secprem}
\subsection{Basic Notations}       
For a finite set $\mathsf{A}$, we denote the
set of probability distributions on $\mathsf{A}$ by $\mathsf{P}(\mathsf{A})$.
Let $\rho_1$ and  $\rho_2$ be  Hermitian   operators on a  finite-dimensional
complex Hilbert  space $G$.
We say $\rho_1\geq\rho_2$ and $\rho_2\leq\rho_1$ if $\rho_1-\rho_2$
is positive semidefinite.
 For a finite-dimensional
complex Hilbert space  $G$, we denote
the
set of  density operators on $G$ by
\[\mathcal{S}(G):= \{\rho \in \mathcal{L}(G) :\rho  \text{ is Hermitian, } \rho \geq 0_{G} \text{ , }  \mathrm{tr}(\rho) = 1 \}\text{ ,}\]
where $\mathcal{L}(G)$ is the set  of linear  operators on $G$, and $0_{G}$ is the null
matrix on $G$. Note that any operator in $\mathcal{S}(G)$ is bounded.\vspace{0.15cm}

 For  finite-dimensional
complex Hilbert spaces  $G$ and  $G'$,  a quantum channel $N$:
$\mathcal{S}(G) \rightarrow \mathcal{S}(G')$, $\mathcal{S}(G)  \ni
\rho \rightarrow N(\rho) \in \mathcal{S}(G')$ is represented by a
completely positive trace-preserving map
 which accepts input quantum states in $\mathcal{S}(G)$ and produces output quantum
states in  $\mathcal{S}(G')$.

If the sender wants to transmit a classical message of a finite set $A$ to
the receiver using a quantum channel $N$, his encoding procedure will
include a classical-to-quantum encoder 
to prepare a quantum message state $\rho \in
\mathcal{S}(G)$ suitable as an input for the channel. If the sender's
encoding is restricted to transmitting an  indexed finite set of
 quantum states $\{\rho_{x}: x\in \mathsf{A}\}\subset
\mathcal{S}(G)$, then we can consider the choice of the signal
quantum states $\rho_{x}$ as a component of the channel. Thus, we
obtain a channel $\sigma_x := N(\rho_{x})$ with classical inputs $x\in \mathsf{A}$ and quantum outputs,
 which we call a classical-quantum
channel. This is a map $\mathbf{N}$: $\mathsf{A} \rightarrow
\mathcal{S}(G')$, $\mathsf{A} \ni x \rightarrow \mathbf{N}(x) \in
\mathcal{S}(G')$ which is represented by the set of $|\mathsf{A}|$ possible
output quantum states $\left\{\sigma_x = \mathbf{N}(x) :=
N(\rho_{x}): x\in \mathsf{A}\right\}\subset \mathcal{S}(G')$, meaning that
each classical input of $x\in \mathsf{A}$ leads to a distinct quantum output
$\sigma_x \in \mathcal{S}(G')$. In view of this, we have the following
definition.\vspace{0.15cm}

Let $\mathsf{A}$ be a  finite set and $H$ be a finite-dimensional
complex Hilbert space. A classical-quantum channel is
a linear map $W: \mathsf{P}(\mathsf{A})\rightarrow\mathcal{S}(H)$,
$\mathsf{P}(\mathsf{A})  \ni P
\rightarrow W(P) \in \mathcal{S}(H)$.
Let $a\in \mathsf{A}$. For a $P_a\in \mathsf{P}(\mathsf{A})$, 
defined by $P_a(a')= \begin{cases} 1 &\mbox{if } a'=a\\
0 &\mbox{if } a'\not=a \end{cases}$, we 
write $W(a)$ instead of $W(P_a)$.

\begin{remark} In much literature, a classical-quantum channel is
defined as a map $ \mathsf{A}\rightarrow\mathcal{S}(H)$,
$\mathsf{A}  \ni a
\rightarrow W(a) \in \mathcal{S}(H)$. This is a 
special case when the input is limited on the set $\{P_a :
a\in \mathsf{A}\}$. \end{remark}\vspace{0.15cm}

For a probability distribution $P$ on a finite set $\mathsf{A}$  and a positive constant $\delta$,
we denote the set of typical sequences by 
\[\mathsf{T}^n_{P,\delta} :=\left\{ a^n \in \mathsf{A}^n: |\frac{1}{n} N(a' \mid a^n) - P(a')|  \leq \frac{\delta}{n}\forall a'\in \mathsf{A}\right\}\text{ ,}\]
where $N(a' \mid a^n)$ is the number of occurrences of the symbol $a'$ in the sequence $a^n$.\vspace{0.15cm}


Let  $n\in\mathbb{N}$.
we define $\mathsf{A}^n:= \{(a_1,\cdots,a_n): a_i \in \mathsf{A}
\text{ } \forall i \in \{1,\cdots,n\}\}$. 
 The space which the vectors
$\{v_1\otimes \cdots \otimes v_n: v_i \in H
\text{ } \forall i \in \{1,\cdots,n\}\}$ span is denoted 
by $H^{\otimes n}$. We also write $a^n$ 
for the elements of
$\mathsf{A}^n$.

Associated to $W$ is the channel map on the n-block $W^{\otimes n}$: $P({\mathsf{A}}^n)
\rightarrow \mathcal{S}({H}^{\otimes n})$, such that 
 $W^{\otimes n}(P^n) = W(P_1)
\otimes \cdots \otimes W(P_n)$
if $P^n\in \mathsf{P}(\mathsf{A}^{ n})$ can be given by $P^n(a^n)=\prod_j P_j(a_j)$
for every
$a^n = (a_1,\cdots,a_n)$ $\in\mathsf{A}^{ n}$. 
Let $\theta$ $:=$ $\{1,\cdots,T\}$ be a finite set.
Let $\Bigl\{W_t:t\in\theta\Bigr\}$ be a set of classical-quantum channels.
For $t^n=(t_1,\cdots,t_n)$, $t_i\in\theta$ we define the n-block $W_{t^n}$ 
such that  for
 $W_{t^n}(P^n) = W_{t_1}(P_1)
\otimes \cdots \otimes W_{t_n}(P_n)$
if $P^n\in \mathsf{P}(\mathsf{A}^{ n})$ can be given by $P^n(a^n)=\prod_j P_j(a_j)$
for every
$a^n $ $\in\mathsf{A}^{ n}$. \vspace{0.15cm}

For a quantum state $\rho\in \mathcal{S}(G)$ we denote the von Neumann
entropy of $\rho$ by \[S(\rho)=- \mathrm{tr}(\rho\log\rho)\text{
.}\]

Let $\mathfrak{P}$ and $\mathfrak{Q}$ be quantum systems. We 
denote the Hilbert space of $\mathfrak{P}$ and $\mathfrak{Q}$ by 
$G^\mathfrak{P}$ and $G^\mathfrak{Q}$, respectively. Let $\phi^\mathfrak{PQ}$ be a bipartite
quantum state in $\mathcal{S}(G^\mathfrak{PQ})$. 
We denote the partial
trace over $G^\mathfrak{P}$ by
\[\mathrm{tr}_{\mathfrak{P}}(\phi^\mathfrak{PQ}):= 
\sum_{l} \langle l|_{\mathfrak{P}} \phi^\mathfrak{PQ} |  l \rangle_{\mathfrak{P}}\text{ ,}\]
where $\{ |  l \rangle_{\mathfrak{P}}: l\}$ is an orthonormal basis
of $G^\mathfrak{P}$.
We denote the conditional entropy by
\[S(\mathfrak{P}\mid\mathfrak{Q})_{\phi}:=
S(\phi^\mathfrak{PQ})-S(\phi^\mathfrak{Q})\text{ .}\]
The quantum mutual information is denoted by
\[I(\mathfrak{P};\mathfrak{Q})_{\phi}=S(\phi^\mathfrak{P})+ S(\phi^\mathfrak{Q}) -S(\phi^\mathfrak{PQ})\text{ .}\]
Here $\phi^\mathfrak{Q}=\mathrm{tr}_{\mathfrak{P}}(\phi^\mathfrak{PQ})$ and $\phi^\mathfrak{P}=\mathrm{tr}_{\mathfrak{Q}}(\phi^\mathfrak{PQ})$.
Let
$\mathtt{V}$: $\mathsf{A} \rightarrow
\mathcal{S}(G)$ be a classical-quantum
channel. Following \cite{Ahl/Win}, for $P\in P(\mathsf{A})$
the conditional entropy of the channel for $\mathtt{V}$ with input distribution $P$
is denoted by
 \[S(\mathtt{V}|P) := \sum_{x\in \mathsf{A}} P(x)S(\mathtt{V}(x))\text{
.}\]

Let 
$\Phi := \{\rho_a : a\in \mathsf{A}\}$  be a set of quantum  states labeled by
elements of $\mathsf{A}$. For a probability distribution $P$ on $\mathsf{A}$,
the    Holevo $\chi$ quantity is
defined as
\[\chi(P;\Phi):= S\left(\sum_{a\in \mathsf{A}} P(a)\rho_a\right)-
\sum_{a\in \mathsf{A}} P(a)S\left(\rho_a\right)\text{ .}\] 

For a set $\mathbf{A}$ and a  Hilbert space $G$, let
$\mathbf{V}$: $\mathsf{A} \rightarrow
\mathcal{S}(G)$ be a classical-quantum
channel. For a probability distribution $P$ on $\mathsf{A}$,
the    Holevo $\chi$ quantity  of the channel for $\mathbf{V}$ with input distribution $P$
is
defined as
\[\chi(\mathbf{V};\Phi):= S\left(\mathbf{V} (P)\right)-
S\left(\mathbf{V}|P\right)\text{ .}\] \vspace{0.15cm}

Let
$G$ be a finite-dimensional
complex Hilbert space.
Let $n \in \mathbb{N}$ and $\alpha > 0$.
We suppose $\rho \in \mathcal{S}(G)$ has
the spectral decomposition
$\rho = \sum_{x} P(x)  |x\rangle\langle x|$. Notice that by definition of
the spectral decomposition, the eigenvectors $\{|x\rangle : x\}$ form an orthonormal system
 (sometimes also called the ``computational basis'').
Its
$\alpha$-typical subspace is the subspace spanned
by $\left\{|x^n\rangle: x^n \in {\mathsf{T}}^n_{P, \alpha}\right\}$,
where  $|x^n\rangle:=\otimes_{i=1}^n|x_i\rangle$. The orthogonal subspace projector onto the
typical subspace is
\[ \Pi_{\rho ,\alpha}=\sum_{x^n \in {\mathsf{T}}^n_{P, \alpha}}|x^n\rangle\langle x^n|\text{ .}\]

Similarly let $\mathsf{A}$ be a finite set, and  $G$ be a finite-dimensional
complex Hilbert space.
Let
$\mathtt{V}$: $\mathsf{A} \rightarrow
\mathcal{S}(G)$ be a classical-quantum
channel. For $a\in\mathsf{A}$,  suppose
$\mathtt{V}(a)$ has
the spectral decomposition $\mathtt{V}(a)$
$ =$ $\sum_{j}
V(j|a) |j\rangle_a\langle j|_a$
for a stochastic matrix
$V(\cdot|\cdot)$.
In an effort to enhance readability, we will typically suppress the subscript 
$a$ in the above decomposition and typically write $\sum_j V(j|a)|j\rangle\langle j|$ 
whenever this causes no ambiguity. The same reasoning applies to the next definition.
 The $\alpha$-conditional typical
subspace of $\mathtt{V}$ for a typical sequence   $a^n$ is the
subspace spanned by
 $\left\{\bigotimes_{a\in\mathsf{A}}|j^{\mathtt{I}}\rangle_a : j^{\mathtt{I}_a} \in \mathsf{T}^{\mathtt{I}_a}_{V(\cdot|a),\delta}\right\}$.
Here $\mathtt{I}_a$ $:=$ $\{i\in\{1,\cdots,n\}: a_i = a\}$ is an indicator set that selects the indices $i$ in the sequence $a^n$
$=$ $(a_1,\cdots,a_n)$ for which the $i$th
symbol $a_i$ is equal to $a\in\mathsf{A}$.
The subspace is often referred to as the $\alpha$-conditional typical
subspace of the state  $\mathtt{V}^{\otimes n}(a^n)$.
  The orthogonal subspace projector onto it is defined as
	\[\Pi_{\mathtt{V}, \alpha}(a^n)=\bigotimes_{a\in\mathsf{A}}
\sum_{j^{\mathtt{I}_a} \in {\mathsf{T}}^{\mathtt{I}_a}_{\mathtt{V}(\cdot \mid a^n),\alpha}}|j^{\mathtt{I}_a} \rangle \langle j^{\mathtt{I}_a}|\text{ .}
\]
The typical subspace 
has following  properties:

For $\sigma \in \mathcal{S}(G^{\otimes n})$ and $\alpha > 0$,
there are positive constants $\beta(\alpha)$, $\gamma(\alpha)$, 
and $\delta(\alpha)$, depending on $\alpha$, such that
\begin{equation} \label{te1}
\mathrm{tr}\left({\sigma} \Pi_{\sigma ,\alpha}\right) > 1-2^{-n\beta(\alpha)}
\text{ ,}\end{equation}

\begin{equation} \label{te2}
2^{n(S(\sigma)-\delta(\alpha))}\le \mathrm{tr} \left(\Pi_{\sigma ,\alpha}\right)
\le 2^{n(S(\sigma)+\delta(\alpha))}
\text{ ,}\end{equation}

\begin{equation} \label{te3}
2^{-n(S(\sigma)+\gamma(\alpha))} \Pi_{\sigma ,\alpha} \le \Pi_{\sigma ,\alpha}
{\sigma} \Pi_{\sigma ,\alpha}
\le 2^{-n(S(\sigma)-\gamma(\alpha))} \Pi_{\sigma ,\alpha}
\text{ .}\end{equation}

For 
$a^n \in {\mathsf{T}}^n_{P, \alpha}$, 
there are positive constants $\beta(\alpha)'$, $\gamma(\alpha)'$, 
and $\delta(\alpha)'$, depending on $\alpha$ such that

\begin{equation} \label{te4}
\mathrm{tr}\left(\mathtt{V}^{\otimes n}(a^n) \Pi_{\mathtt{V}, \alpha}(a^n)\right)
> 1-2^{-n\beta(\alpha)'}
\text{ ,}\end{equation}

\begin{align} \label{te5}
&2^{-n(S(\mathtt{V}|P)+\gamma(\alpha)')} \Pi_{\mathtt{V}, \alpha}(a^n)
 \le \Pi_{\mathtt{V}, \alpha}(a^n)\mathtt{V}^{\otimes n}(a^n) \Pi_{\mathtt{V}, \alpha}(a^n)\notag\\
 &\le 2^{-n(S(\mathtt{V}|P)-\gamma(\alpha)')} \Pi_{\mathtt{V}, \alpha}(a^n)
\text{ ,}\end{align}
\begin{equation} \label{te6}
2^{n(S(\mathtt{V}|P)-\delta(\alpha)')}\le \mathrm{tr}\left(
\Pi_{\mathtt{V}, \alpha}(a^n) \right)\le 2^{n(S(\mathtt{V}|P)+\delta(\alpha)')}
\text{ .}\end{equation}

For the classical-quantum channel
$\mathtt{V}: \mathsf{P}(\mathsf{A}) \rightarrow \mathcal{S}(G)$ and a probability
distribution $P$ on $\mathsf{A}$,  we define
 a quantum state $P\mathtt{V}$ $:=$ $\mathtt{V}(P)$ on $\mathcal{S}(G)$.
 For $\alpha > 0$, we define an
orthogonal subspace projector $\Pi_{P\mathtt{V}, \alpha}$
 fulfilling (\ref{te1}), (\ref{te2}), and (\ref{te3}).
Let $x^n\in{\mathsf{T}}^n_{P, \alpha}$.
For $\Pi_{P\mathtt{V}, \alpha}$, there is a positive constant $\beta(\alpha)''$ such that
following inequality holds:
\begin{equation} \label{te7}  \mathrm{tr} \left(  \mathtt{V}^{\otimes n}(x^n) \cdot \Pi_{P\mathtt{V}, \alpha } \right)
 \geq 1-2^{-n\beta(\alpha)''} \text{ .}\end{equation}

We give here a sketch of the proof. For a detailed proof, please see \cite{Wil}.

\begin{proof}
(\ref{te1}) holds because 
$\mathrm{tr}\left({\sigma} \Pi_{\sigma ,\alpha}\right)$
$=$ $\mathrm{tr}\left(\Pi_{\sigma ,\alpha}{\sigma} \Pi_{\sigma ,\alpha}\right) $
$=$ $P({\mathsf{T}}^n_{P, \alpha})$.
(\ref{te2}) holds because
$\mathrm{tr} \left(\Pi_{\sigma ,\alpha}\right)$ $=$
$\left\vert {\mathsf{T}}^n_{P, \alpha} \right\vert$.
(\ref{te3}) holds because
$2^{-n(S(\sigma)+\gamma(\alpha))}$ $\le$
$P^n(x^n)$ $\le$ $2^{-n(S(\sigma)-\gamma(\alpha))}$
for $x\in {\mathsf{T}}^n_{P, \alpha}$ and a positive $\gamma(\alpha)$.
(\ref{te4}), (\ref{te5}), and (\ref{te6})
can be obtained in similar way.
(\ref{te7}) follows from the permutation-invariance of $\Pi_{P\mathtt{V}, \alpha}$.

 \qed \end{proof}

\subsection{Communication
Scenarios and Code Concepts}

\begin{definition}
Let $\mathsf{A}$ be a finite set, let
 $H$  be a finite-dimensional
complex Hilbert space, and
  $\theta$ $:=$ $\{1,\cdots,T\}$ be an index set.
    For every $t \in \theta$,  let $W_{t}$   be a classical-quantum channel
$\mathsf{P}(\mathsf{A}) \rightarrow \mathcal{S}(H)$.
We call the set of the  classical-quantum
channels  $\{W_t : t \in \theta\}$ an \bf arbitrarily varying
classical-quantum  channel \rm
when the channel state $t$ varies from
symbol to symbol in an  arbitrary manner.\end{definition}

When the sender inputs a sequence $a^n \in \mathsf{A}^n$ into the channel, the receiver
receives the output $W_t^{ n}(a^n)
 \in \mathcal{S}(H^{\otimes n})$, where $t^n = (t_1, t_2, \cdots ,
t_n)\in\theta^n$ is the channel state of $W_t^{ n}$.

\begin{definition}\label{symmet}
We say that the arbitrarily varying classical-quantum  channel
$\{W_t : t \in \theta\}$ is \bf symmetrizable \rm if
 there exists a
parametrized set of distributions $\{\tau(\cdot\mid a):
 a\in \mathsf{A}\}$ on $\theta$ such that for all $a$, ${a'}\in \mathsf{A}$,
\[\sum_{t\in\theta}\tau(t\mid a)W_{t}({a'})=\sum_{t\in\theta}\tau(t\mid {a'})W_{t}(a)\text{ .}\]
\end{definition}

When the sender inputs  a sequence $a^n \in \mathsf{A}^n$  into the channel, the receiver
receives the output $W_t^{\otimes n}(a^n)
 \in \mathcal{S}(H^{\otimes n})$, where $t^n = (t_1, t_2, \cdots ,
t_n)\in\theta^n$ is the channel state, while the wiretapper  receives an output quantum  state ${V}_t^{\otimes n}(a^n)
 \in \mathcal{S}({H}'^{\otimes n})$.\vspace{0.2cm}

\begin{definition}
Let $\mathsf{A}$ be a finite set. Let
 $H$ and $H'$ be finite-dimensional
complex Hilbert spaces. Let
  $\theta$ $:=$ $\{1,2, \cdots\}$ be an index set. For every $t \in \theta$,  let $W_{t}$   be a classical-quantum  channel
$\mathsf{P}(\mathsf{A}) \rightarrow \mathcal{S}(H)$ and ${V}_t$ be a classical-quantum  channel $\mathsf{P}(\mathsf{A})
\rightarrow \mathcal{S}(H')$. We call the set of the  classical-quantum 
channel pairs  $\{(W_t,{V}_t): t \in \theta\}$ an \bf arbitrarily
varying classical-quantum wiretap channel \rm when the
state $t$ varies from symbol to symbol in an  arbitrary
manner, while  the legitimate
receiver accesses the output of the first channel, i.e.,  $W_t$  in the pair
$(W_t,{V}_t)$, and the wiretapper observes the output of the second
 channel, i.e.,   ${V}_t$ in the pair $(W_t,{V}_t)$, respectively.\end{definition}

\begin{definition}
Let $\mathsf{A}$ be a finite set. Let
 $H$ and $H'$ be finite-dimensional
complex Hilbert spaces. Let  $\overline{\theta}$ $:=$ $\{1, 2, \cdots\}$ and
  $\theta$ $:=$ $\{1,2, \cdots\}$ be index sets. For every $s \in \overline{\theta}$  let $\overline{W}_{s}$   be a classical-quantum channel
$\mathsf{P}(\mathsf{A}) \rightarrow \mathcal{S}(H)$. For every $t\in\theta$ let ${V}_t$ be a classical-quantum channel $\mathsf{P}(\mathsf{A}) 
\rightarrow \mathcal{S}(H')$. We call the set of the  classical-quantum
channel pairs  $\{(\overline{W}_s,{V}_t) :s \in \overline{\theta}, t \in \theta\}$ a \bf compound-arbitrarily varying wiretap classical-quantum channel\rm,
when the channel state $s$ remains
constant over time, but the legitimate users can not  control  which  $s$  in the set  $\overline{\theta}$ will be  used and the
state $t$ varies from symbol to symbol in an  arbitrary
manner, while the legitimate
receiver accesses the output of the first channel, i.e.,  $\overline{W}_s$  in the pair
$(\overline{W}_s,{V}_t)$ and the wiretapper observes the output of the second
 channel, i.e.,   ${V}_t$ in the pair $(\overline{W}_s,{V}_t)$, respectively.\end{definition}

\begin{definition}
 An  \bf $(n, J_n)$   (deterministic)  code \rm $\mathcal{C}$ for a
 classical-quantum channel 
consists of a stochastic encoder $E$ : $\{
1,\cdots ,J_n\} \rightarrow \mathsf{P}(\mathsf{A}^n)$,  specified by
a matrix of conditional probabilities $E(\cdot|\cdot)$ and
 a collection of positive-semidefinite operators $\left\{D_j: j\in \{ 1,\cdots ,J_n\}\right\}
 \subset \mathcal{S}({H}^{\otimes n})$,
which is a partition of the identity, i.e., $\sum_{j=1}^{J_n} D_j =
\mathrm{id}_{{H}^{\otimes n}}$. We call these  operators the decoder operators.
\end{definition}
A code is created by the sender and the legitimate receiver before the
message transmission starts. The sender uses the encoder to encode the
message that he wants to send, while the legitimate receiver uses the
decoder operators on the channel output to decode the message.

 \begin{definition}
An \bf $(n, J_n)$  randomness  assisted quantum code \rm for the
arbitrarily varying classical-quantum wiretap channel
$\{(W_t,{V}_t): t \in \theta\}$ is a distribution $G$ on
$\left(\Lambda,\sigma\right)$, where we denote  the set of $(n, J_n)$ deterministic  codes by $\Lambda$
and $\sigma$ is a sigma-algebra so
chosen such that the functions $\gamma \rightarrow P_e(\mathcal{C}^{\gamma},
t^n)$ and  $\gamma \rightarrow \chi
\left(R_{uni};Z_{\mathcal{C}^{\gamma},t^n}\right)$ are both $G$-measurable
with respect to $\sigma$ for every $t^n\in\theta^n$, here
$Z_{\mathcal{C}^{\gamma},t^n} :=\{{V}_{t^{n}}(E^{\gamma}(\cdot\mid 1)),
{V}_{t^{n}}(E^{\gamma}(\cdot\mid 2)),\cdots,
{V}_{t^{n}}(E^{\gamma}(\cdot\mid J_n))\}$.
\label{randef}\end{definition}

\begin{definition}
A non-negative number $R$ is an achievable  (deterministic) \bf secrecy
rate \rm for the arbitrarily varying classical-quantum wiretap channel
$\{(W_t,{V}_t): t \in \theta\}$ if for every $\epsilon>0$, $\delta>0$,
$\zeta>0$ and sufficiently large $n$ there exists an  $(n, J_n)$
code $\mathcal{C} = \bigl(E, \{D_j^n : j = 1,\cdots J_n\}\bigr)$  such that $\frac{\log
J_n}{n}
> R-\delta$, and
\begin{equation} \max_{t^n \in \theta^n} P_e(\mathcal{C}, t^n) < \epsilon\text{ ,}\label{annian1}\end{equation}
\begin{equation}\max_{t^n\in\theta^n}
\chi\left(R_{uni};Z_{t^n}\right) < \zeta\text{
,}\label{b40}\end{equation} where $R_{uni}$ is the uniform
distribution on $\{1,\cdots J_n\}$. Here
$P_e(\mathcal{C}, t^n)$ (the average probability of the decoding error of a
deterministic code $\mathcal{C}$, when the channel state  of the
arbitrarily varying classical-quantum wiretap channel $\{(W_t,{V}_t): t \in \theta\}$  is $t^n =
(t_1, t_2, \cdots , t_n)$), is defined as \[ P_e(\mathcal{C}, t^n) := 1- \frac{1}{J_n} \sum_{j=1}^{J_n}
\mathrm{tr}(W_{t^n}(E(~|j))D_j)\text{ ,}\] and
$Z_{t^n}=\Bigl\{{V}_{t^n}(E(~|i)):$ 
$i\in\{1,\cdots,J_n\}\Bigr\}$ 
 is the set of the resulting
quantum state at the output of the wiretap channel when the channel
state of $\{(W_t,{V}_t): t \in \theta\}$ is $t^n$.
\end{definition}

\begin{remark} A weaker and widely used
 security criterion is obtained if we replace
(\ref{b40})  with $\max_{t \in \theta}\frac{1}{n}
\chi\left(R_{uni};Z_{t^n}\right)$ $<$ $\zeta$. In this paper,  we
will follow
  \cite{Bj/Bo/So} and use  (\ref{b40}).
\label{remsc}\end{remark}

\begin{definition}
 An \bf $(n, J_n)$  common randomness assisted
quantum code \rm for  the
arbitrarily varying classical-quantum wiretap
channel $\{(W_t,{V}_t): t \in \theta\}$ is a 
 finite subset $\Bigl\{C^{\gamma}=\{(E^{\gamma},D_j^{\gamma}):j=1,\cdots,J_n\}:\gamma\in\Gamma\Bigr\}$
of the set of $(n, J_n)$ deterministic  codes, labeled by a  finite set $\Gamma$.

\end{definition}

\begin{definition}
A non-negative number $R$ is an achievable \bf  enhanced secrecy
rate \rm for the compound-arbitrarily varying wiretap classical-quantum channel
$\{(\overline{W}_s,{V}_t) :s \in \overline{\theta}, t \in \theta\}$ if for every $\epsilon>0$, $\delta>0$,
$\zeta>0$ and sufficiently large $n$ there exists an  $(n, J_n)$
code $\mathcal{C} = \bigl(E^n, \{D_j^n : j = 1,\cdots J_n\}\bigr)$  such that $\frac{\log
J_n}{n}
> R-\delta$, and
\begin{equation} \max_{s \in \overline{\theta}} P_e(\mathcal{C},s,n) < \epsilon\text{ ,}\end{equation}
\begin{equation}\max_{t^n\in\theta^n}\max_{\pi\in\Pi_n}
\chi\left(R_{uni};Z_{t^n,\pi}\right) < \zeta\text{
,}\label{b40a}\end{equation} where $R_{uni}$ is the uniform
distribution on $\{1,\cdots J_n\}$. Here
$P_e(\mathcal{C},s,n)$  is defined as follows
\[ P_e(\mathcal{C}, s,n) := 1- \frac{1}{J_n} \sum_{j=1}^{J_n}  
\mathrm{tr}(\overline{W}_{s}^{\otimes n}(E^n(~|j))D_j^n)\text{ ,}\] and
$Z_{t^n,\pi}=\Bigl\{\sum_{a^n\in
\mathsf{A}^n}E^n(\pi(a^n)|1){V}^{t^n}(\pi(a^n)),\sum_{a^n\in
\mathsf{A}^n}E^n(\pi(a^n)|2){V}^{t^n}(\pi(a^n)),$ $\cdots,$ 
$\sum_{a^n\in
\mathsf{A}^n}E^n(\pi(a^n)|J_n){V}^{t^n}(\pi(a^n))\Bigr\}$.
\end{definition}

\begin{definition}
A non-negative number $R$ is an  achievable \bf secrecy
rate \rm for
$\{(\overline{W}_s,{V}_t) :s \in \overline{\theta}, t \in \theta\}$ if for every $\epsilon>0$, $\delta>0$,
$\zeta>0$ and sufficiently large $n$ there exists an  $(n, J_n)$
code $\mathcal{C} = \bigl(E^n, \{D_j^n : j = 1,\cdots J_n\}\bigr)$  such that $\frac{\log
J_n}{n}
> R-\delta$, and
\[ \max_{s \in \overline{\theta}} P_e(\mathcal{C},s,n) < \epsilon\text{ ,}\]
\[\max_{t^n\in\theta^n}\max_{\pi\in\Pi_n}
\chi\left(R_{uni};Z_{t^n}\right) < \zeta\text{
.}\] 
\end{definition}

\begin{definition}A non-negative number $R$ is an achievable \bf secrecy rate \rm for the
arbitrarily varying classical-quantum wiretap channel
$\{(W_t,{V}_t): t \in \theta\}$ \bf under common randomness assisted  quantum coding \rm if  for
every
 $\delta>0$, $\zeta>0$, and  $\epsilon>0$, if $n$ is sufficiently large
there is an $(n, J_n)$  common randomness assisted
quantum code $(\{\mathcal{C}^{\gamma}:\gamma\in \Gamma\})$  such that
$\frac{\log J_n}{n} > R-\delta$, and
\[ \max_{t^n\in\theta^n} \frac{1}{\left|\Gamma\right|} \sum_{\gamma=1}^{\left|\Gamma\right|}P_{e}(\mathcal{C}^{\gamma},t^n) < \epsilon\text{ ,}\]
\[\max_{t^n\in\theta^n} \frac{1}{\left|\Gamma\right|} \sum_{\gamma=1}^{\left|\Gamma\right|}
\chi\left(R_{uni},Z_{\mathcal{C}^{\gamma},t^n}\right) < \zeta\text{ .}\]This
means that we do not require  common randomness  to be secure
against eavesdropping.\end{definition}

\begin{definition}Let $\mathsf{X}$ and $\mathsf{Y}$ be finite sets.
 We denote  the
sets of joint probability distributions on $\mathsf{X}$ and
$\mathsf{Y}$ by $P(\mathsf{X},\mathsf{Y})$. Let  $(X,Y)$  be a
random variable distributed to a joint probability distribution
$p\in P(\mathsf{X},\mathsf{Y})$.
An \bf $(X,Y)$-correlation assisted
  $(n, J_n)$  code \rm $C(X,Y)$ for the
arbitrarily varying classical-quantum wiretap channel $(W_t,{V}_t)_{t \in \theta}$
consists of a set of stochastic encoders $\left\{E_{\mathsf{x}^n}:
\{ 1,\cdots ,J_n\} \rightarrow \mathsf{P}(\mathsf{A}^n):
\mathsf{x}^n\in\mathsf{X}^n\right\}$, and
 a set of collections of positive semidefinite operators $\left\{\{D_j^{(\mathsf{y}^n)}:
  j = 1,\cdots ,J_n\}  :\mathsf{y}^n\in\mathsf{Y}^n\right\}
$ on $\mathcal{S}({H}^{\otimes n})$ which fulfills $\sum_{j=1}^{J_n} D_j^{(\mathsf{y}^n)} =
\mathrm{id}_{{H}^{\otimes
n}}$ for every $\mathsf{y}^n\in\mathsf{Y}^n$.

$R$ is an achievable \bf $(X,Y)$
   secrecy rate \rm for  $(W_t,{V}_t)_{t \in \theta}$ if for every positive
$\epsilon$, $\delta$, $\zeta$ and sufficiently large $n$ there
exist an $(X,Y)$-correlation assisted   $(n, J_n)$  code
$C(X,Y) =
\biggl\{\left(E_{\mathsf{x}^n},D_j^{(\mathsf{y}^n)}\right): j\in\{
1,\cdots ,J_n\},\text{ }  \mathsf{x}^n\in\mathsf{X}^n,\text{ }
\mathsf{y}^n\in\mathsf{Y}^n\biggr\}$  such that $\frac{\log J_n}{n}
> R-\delta$, and
\[\max_{t^n \in \theta^n} \sum_{\mathsf{x}^n\in\mathsf{X}^n}
 \sum_{\mathsf{y}^n\in\mathsf{Y}^n}p(\mathsf{x}^n,\mathsf{y}^n) P_e(C(\mathsf{x}^n,\mathsf{y}^n), t^n) < \epsilon\text{ ,}\]
\[\max_{t^n\in\theta^n} \sum_{x^n\in\mathbf X^n}{p_\mathsf{X}}^{\otimes n}(x^n)
\chi\left(R_{uni};Z_{t^n,\mathsf{x}^n}\right) < \zeta\text{ ,}\]
where $P_e(C(\mathsf{x}^n,\mathsf{y}^n), t^n)$  
$ :=$ $1-$ $\frac{1}{J_n} \sum_{j=1}^{J_n}  \sum_{a^n \in \mathsf{A}^n}$
$E_{\mathsf{x}^n}(a^n|j)\mathrm{tr}(W_{t^n}(a^n)D_j^{(\mathsf{y}^n)})$.
\end{definition} 

\begin{definition}
The supremum of all achievable  (deterministic) secrecy rates of
$\{(W_t,{V}_t): t \in \theta\}$ is called the (deterministic)  secrecy
capacity of $\{(W_t,{V}_t): t \in \theta\}$, denoted by
$C_s(\{(W_t,{V}_t): t \in \theta\})$.

The supremum of  all  achievable secrecy rates
under common randomness assisted  quantum coding  of
$\{(W_t,{V}_t): t \in \theta\}$ is called the common randomness
assisted   secrecy capacity of $\{(W_t,{V}_t): t \in \theta\}$,
denoted by $C_s(\{(W_t,{V}_t): t \in \theta\};cr)$.

The supremum of  all achievable enhanced secrecy rates of
$\{(\overline{W}_s,{V}_t) :s \in \overline{\theta}, t \in \theta\}$  is called the 
enhanced   secrecy
capacity of $\{(\overline{W}_s,{V}_t) :s \in \overline{\theta}, t \in \theta\}$, denoted by
$\hat{C}_s(\{(\overline{W}_s,{V}_t) :s \in \overline{\theta}, t \in \theta\})$.

The  supremum of  all  achievable $(X,Y)$
   secrecy rate  of $\{(W_t,{V}_t): t \in \theta\}$
 is called the
 $(X,Y)$
      secrecy capacity of $\{(W_t,{V}_t): t \in \theta\}$.
\end{definition}

\section{Compound-Arbitrarily Varying Wiretap Classical-Quantum Channel}\label{CAVWCQC}

Let $\mathsf{A}$,  $H$, $H'$, $\theta$, and 
 $(W_t,V_t)_{t \in \theta}$ be defined as in Section \ref{secprem}.

Following the idea of \cite{Wi/No/Bo}, we first prove the following Theorem.

\begin{theorem}
Let  $\overline{\theta}$ $:=$ $\{1,\cdots,\overline{T}\}$ and
  $\theta$ $:=$ $\{1,\cdots,T\}$ be finite index sets.
	Let $\{(\overline{W}_s,{V}_t) :s \in \overline{\theta}, t \in \theta\}$ 
	be a compound-arbitrarily varying wiretap classical-quantum channel.
	We have
\begin{align}&
	\hat{C}_s(\{(\overline{W}_s,{V}_t) :s \in \overline{\theta}, t \in \theta\})\notag\\
	&= \lim_{n\rightarrow \infty} \frac{1}{n}\max_{\Lambda_n}
	\Bigl(\min_{s \in \overline{\theta}}\chi(p_U;B_s^{\otimes n})- \max_{t^n\in \theta^n}\chi(p_U;Z_{t^n})\Bigr)
	\text{ ,}\end{align}
	where $B_s$ are the resulting  quantum states at the output of the
legitimate receiver's channels. $Z_{t^n}$ are the resulting  quantum states  at
the output of wiretap channels. By $\max_{\Lambda_n}$, we mean that the maximum is taken 
over all ensembles that arise from taking an arbitrary finite set $\mathsf{U}$
and defining ensembles 
$\{p_U(u),\sum_{a^n\in\mathsf{A}^n}
p_{A^n|U}(a^n|u)W_s^{\otimes n}(a^n)\}_{u\in\mathsf{U}}$ 
and $\{p_U(u),\sum_{a^n\in\mathsf{A}^n}
p_{A^n|\mathsf{U}}(a^n|u)V_{t^n}(a^n)\}_{u\in\mathsf{U}}$
 for every $s \in \overline{\theta}$, and
$ t^n \in \theta^n$ to calculate the respective Holevo quantities.  $A^n$ is here a random
variable taking values on $\mathsf{A}^n$, $U$ a random
variable taking values on  $\mathsf{U}$
with probability  distribution $p_U$, and 
$p_{A^n\mid U} \in P(\mathsf{A}^n)$ 
the   conditional distribution of
$A^n$ given $U$.
	\label{loatbfis}
\end{theorem}

\begin{proof}
We fix a probability distribution
$p\in \mathsf{P}(\mathsf{A})$.
Let
\[J_n = \lfloor 2^{n\min_{s \in \overline{\theta}}\chi(p;B_s)-\log
L_{n}-2n\mu} \rfloor\text{ .}\]

Let $p' (x^n):= \begin{cases} \frac{p^{
n}(x^n)}{p^{ n}
(\mathsf{T}^n_{p,\delta})} \text{ ,}& \text{if } x^n \in \mathsf{T}^n_{p,\delta}\text{ ;}\\
0 \text{ ,}& \text{ else .}  \end{cases}$

Let $X^{n} := \{X_{j,l}\}_{j \in \{1, \dots, J_n\}, l \in
\{1, \dots, L_{n}\}}$ be a family of random variables taking value
according to $p'$, i.e., 
with
the uniform distribution over $\mathsf{T}^n_{P,\delta}$.
Here
 $L_{n}$ is a natural
number which will be specified later.\vspace{0.2cm}

We fix a $t^n$ $\in\theta^n$ and
 define a map  $\mathsf{V}:$ $\mathsf{P}(\theta) \times \mathsf{P}(\mathsf{A})$
$\rightarrow$ $\mathcal{S}(H)$ by
\[\mathsf{V}(t,p):= V_{t}(p)\text{ .}\]
For  $t\in \theta$ we define $q(t)$ $:=$ $\frac{N(t\mid t^n)}{n}$.  $t^n$ 
is trivially a typical sequence of $q$.
For  $p\in  \mathsf{P}(\mathsf{A})$,  $\mathsf{V}$
defines a map
$\mathsf{V}(\cdot,p):$
$\mathsf{P}(\theta) $
$\rightarrow$ $\mathcal{S}(H)$.

 Let
\[Q_{t^n}(x^n) := \Pi_{\mathsf{V}(\cdot,p), \alpha }(t^n)\Pi_{\mathsf{V},\alpha}(t^n, x^n)
 \cdot V_{{t^n}}(x^n) \cdot \Pi_{\mathsf{V},\alpha}(t^n,x^n)\Pi_{\mathsf{V}(\cdot,p), \alpha }(t^n)\text{ .}\]\vspace{0.2cm}

\begin{lemma} [Gentle Operator, cf.  \cite{Win} and \cite{Og/Na}] \label{eq_4a}  
Let $\rho$ be a  quantum state and $X$ be a positive operator with $X  \leq
\mathrm{I}$  and $1 - \mathrm{tr}(\rho X)  \leq
\lambda \leq1$. Then
\begin{equation} \| \rho -\sqrt{X}\rho \sqrt{X}\|_1 \leq \sqrt{2\lambda}\text{ .} \label{tenderoper}
\end{equation}\end{lemma}

The Gentle Operator was first introduced in \cite{Win}, where it has been
shown that $\| \rho -\sqrt{X}\rho \sqrt{X}\|_1 \leq
\sqrt{8\lambda}$. In \cite{Og/Na}, the result of \cite{Win}
has been improved, and (\ref{tenderoper}) has been proved.

In view of the fact
that $\Pi_{\mathsf{V}(\cdot,p), \alpha }(t^n)$
 and $\Pi_{\mathsf{V},\alpha}(t^n, x^n)$ are
both projection matrices,
by   (\ref{te1}), (\ref{te7}),  and   Lemma \ref{eq_4a}
for any $t$ and $x^n$, it holds that
\begin{equation} \label{eq_4}\|Q_{{t^n}}(x^n)-V_{{t^n}}(x^n)\|_1 \leq
\sqrt{2^{-n\beta(\alpha)}+2^{-n\beta(\alpha)''} }\text{ .}\end{equation}\vspace{0.2cm}

The following  Lemma was first given in \cite{Ahl/Win}.
Here we cite the lemma as it was formulated in \cite{Wil}.

\begin{lemma} [Covering Lemma]\label{cov}
Let  $\mathcal{V}$ be a finite-dimensional Hilbert space. 
Let $\mathsf{M}$ be a finite set. Suppose we have
an ensemble  $\{\rho_m:m\in\mathsf{M} \} \subset \mathcal{S}(\mathcal{V})$   of quantum states.
Let  $p$ be
 a   probability distribution on $\mathsf{M}$.

Suppose a total subspace projector $\Pi$ and codeword subspace projectors $\{\Pi_{m}: m\in \mathsf{M}\}$  exist which project
onto subspaces of the Hilbert space in which the states exist, and  for all $ m\in \mathsf{M}$ there are positive constants
$\epsilon\in
]0,1[$, $D$, $d$ such that the following conditions hold:
\[\mathrm{tr}(\rho_m\Pi)\geq 1-\epsilon\text{ ,}\] \[\mathrm{tr}(\rho_m\Pi_m)\geq 1-\epsilon\text{ ,}\]
\[\mathrm{tr}(\Pi)\leq D\text{ ,}\] \[\Pi_m\rho_m\Pi_m\leq\frac{1}{d}\Pi_m\text{ .}\]

We denote $\rho:=\sum_{m} p(m) \rho_m$.
We define a  sequence of  i.i.d.~random variables  $X_1,
\dots ,X_L$, taking values  
in $\{\rho_m:m\in\mathsf{M} \}$. If $L\gg \frac{d}{D}$, then

\begin{align}&
  Pr \biggl( \lVert L^{-1}
\sum_{l=1}^{L}\Pi\cdot\Pi_{X_l}\cdot X_l\cdot\Pi_{X_l}\cdot\Pi- \rho\rVert_1 \leq  \epsilon +4 \sqrt{\epsilon} +24 \sqrt[4]{\epsilon}\biggr)  \allowdisplaybreaks\notag\\
& \geq 1- 2D \exp
 \left( -\frac{\epsilon^3Ld}{2\ln 2D} \right) \text{ .}\label{eq_5}\end{align}
\end{lemma}\vspace{0.15cm}

For our result we use an alternative Covering Lemma.

\begin{lemma} \label{cov3}
Let  $\mathcal{V}$ be a finite-dimensional Hilbert space. 
Let $\mathsf{M}$ and  $\mathsf{M}'\subset \mathsf{M}$ be  finite sets. Suppose we have
an ensemble  $\{\rho_m:m\in\mathsf{M} \} \subset \mathcal{S}(\mathcal{V})$   of quantum states.
Let  $p$ be
 a   probability distribution on $\mathsf{M}$.

Suppose a total subspace projector $\Pi$ and codeword subspace projectors $\{\Pi_{m}: m\in \mathsf{M}\}$  exist  which project
onto subspaces of the Hilbert space in which the states exist, and  for all $ m\in \mathsf{M}'$ there are positive constants
$\epsilon\in
]0,1[$, $D$, $d$ such that  the following conditions hold:
\[\mathrm{tr}(\rho_m\Pi)\geq 1-\epsilon\text{ ,}\] \[\mathrm{tr}(\rho_m\Pi_m)\geq 1-\epsilon\text{ ,}\]
\[\mathrm{tr}(\Pi)\leq D\text{ ,}\] 
and 
\[\Pi_m\rho_m\Pi_m\leq\frac{1}{d}\Pi_m\text{ .}\]

We denote $\omega:=\sum_{m\in\mathsf{M}'} p(m) \rho_m$.
Notice that $\omega$ is not a
density operator in general.
We define a  sequence of  i.i.d.~random variables  $X_1,
\dots ,X_L$, taking values  
in $\{\rho_m:m\in\mathsf{M} \}$. If $L\gg \frac{d}{D}$ then

\begin{align}&
Pr \biggl( \lVert L^{-1}
\sum_{i=1}^{L}\Pi\cdot\Pi_{X_i}\cdot X_i \cdot\Pi_{X_i} \cdot\Pi
- \omega \rVert_1\allowdisplaybreaks\notag\\
&\leq    1-p(\mathsf{M}')+  4\sqrt{1-p(\mathsf{M}')}+ 
42 \sqrt[8]{\epsilon} 
 \biggr)\allowdisplaybreaks\notag\\
 & \geq 1- 2D \exp
\left( -p(\mathsf{M}')\frac{\epsilon^3Ld}{2\ln 2D} \right) \text{ .}\label{pbllsilpcpx3}
\end{align}
\end{lemma}

\begin{proof}
We define a function $\mathbbm{1}_{\mathsf{M}'}$ $:\mathsf{M}\rightarrow  \mathsf{M}'\cup \{0^{\mathcal{V}}\}$
by
\[\mathbbm{1}_{\mathsf{M}'}(\rho_m):= \begin{cases} \rho_m \text{ ,}& \text{if } m \in \mathsf{M}'\\
0^{\mathcal{V}} \text{ ,}& \text{if } m \in \mathsf{M}' \text{ ,}\end{cases}\]
where $0^{\mathcal{V}}$ is the zero operator on $\mathcal{V}$, i.e., 
$\langle j|0^{\mathcal{V}}|j\rangle=0$ for all $j\in \mathcal{V}$.
Notice that $0^{\mathcal{V}}$ is not a
density operator.\vspace{0.2cm}

We have
\begin{align}&
\mathrm{tr}\left(\sum_{m\in \mathsf{M}}p(m) \mathbbm{1}_{\mathsf{M}'}(\rho_m)\right)\allowdisplaybreaks\notag\\
&=\mathrm{tr}\left(\sum_{m\in \mathsf{M}'}p(m) \rho_m\right)\allowdisplaybreaks\notag\\
&=\sum_{m\in \mathsf{M}'}p(m) \mathrm{tr}\left(\rho_m\right)\allowdisplaybreaks\notag\\
&= p(\mathsf{M}')\text{ .}\label{mtlsmimmpm}\end{align}

Let $\hat{\Pi}$ be the  projector onto the subspace spanned by the eigenvectors of
$\sum_{m\in \mathsf{M}'}p(m) \Pi\Pi_{m} \rho_m \Pi_{m}\Pi$  whose corresponding
eigenvalues are greater than $p(\mathsf{M}')\frac{\epsilon}{D}$.

The following three
inequalities can be shown
by the same arguments as in the proof of Lemma \ref{cov} in \cite{Wil}:

\begin{equation} \sum_{m\in \mathsf{M}}p(m)  d\cdot \hat{\Pi} \Pi\Pi_{m} \mathbbm{1}_{\mathsf{M}'}(\rho_m) \Pi_{m}\Pi\hat{\Pi} 
\geq p(\mathsf{M}')\frac{d\epsilon}{D}\hat{\Pi}\text{ .}\label{wilderie2}\end{equation}

\begin{align}&
\|\sum_{m\in \mathsf{M}}p(m)\Pi\cdot\Pi_{m}
\cdot \mathbbm{1}_{\mathsf{M}'}(\rho_m) \cdot\Pi_{m}\cdot\Pi - \sum_{m\in \mathsf{M}}p(m)
\cdot \mathbbm{1}_{\mathsf{M}'}(\rho_m)\|_1\allowdisplaybreaks\notag\\
&\leq  \sum_{m\in \mathsf{M}'}p(m) 
\|\Pi\cdot\Pi_{m}
\rho_m \cdot\Pi_{m}\cdot\Pi- \rho_m  \|_1 \allowdisplaybreaks\notag\\
&\leq  \sum_{m\in \mathsf{M}'}p(m) \left(2\sqrt{\epsilon}+2\sqrt{\epsilon+2\sqrt{\epsilon}}\right)\allowdisplaybreaks\notag\\
&= p(\mathsf{M}')\left(2\sqrt{\epsilon}+2\sqrt{\epsilon+2\sqrt{\epsilon}}\right)\allowdisplaybreaks\notag\\
&\leq 2\sqrt{\epsilon}+2\sqrt{\epsilon+2\sqrt{\epsilon}}\allowdisplaybreaks\notag\\
&\leq 6\sqrt[4]{\epsilon}
 \text{ .}\label{tnaresesesl}
\end{align}
The last inequality holds because
$\sqrt{\epsilon+2\sqrt{\epsilon}}$
$\leq  2\sqrt[4]{\epsilon}$ for $0\leq \epsilon \leq 1$.

When $\{\rho_1,\cdots,\rho_L\}$ fulfills 
\begin{align*}&
(1-\epsilon)  \sum_{m\in \mathsf{M}}p(m) \hat{\Pi}\Pi\cdot\Pi_{m}\cdot \mathbbm{1}_{\mathsf{M}'}(\rho_m) \cdot\Pi_{m}\cdot\Pi\hat{\Pi}
\allowdisplaybreaks\notag\\
&\leq  L^{-1}\sum_{i=1}^{L}\hat{\Pi}\Pi\cdot\Pi_{\rho_i}\cdot (\mathbbm{1}_{\mathsf{M}'}(\rho_i))\cdot\Pi_{\rho_i}\cdot\Pi\hat{\Pi}
\allowdisplaybreaks\notag\\
&\leq (1+\epsilon) \sum_{m\in \mathsf{M}}p(m) \hat{\Pi}\Pi\cdot\Pi_{m}\cdot \mathbbm{1}_{\mathsf{M}'}(\rho_m) \cdot\Pi_{m}\cdot\Pi\hat{\Pi}
\text{ ,}
\end{align*}
(i.e.  we assume the event considered in (\ref{useocbound}) below),

then 

\begin{align}& \lVert L^{-1}
\sum_{i=1}^{L}\hat{\Pi}\Pi\cdot\Pi_{\rho_i}\cdot (\mathbbm{1}_{\mathsf{M}'}(\rho_i))\cdot\Pi_{\rho_i}\cdot\Pi\hat{\Pi} \allowdisplaybreaks\notag\\
&- \sum_{m\in \mathsf{M}}p(m) \hat{\Pi}\Pi\cdot\Pi_{m}\cdot \mathbbm{1}_{\mathsf{M}'}(\rho_m) \cdot\Pi_{m}\cdot\Pi\hat{\Pi}\rVert_1\notag\\
& \leq  
\epsilon \text{ .}\label{wilderie3}\end{align} 
\vspace{0.2cm}

\it i)  Application of the Operator Chernoff Bound \rm\vspace{0.2cm}

For all $m\in\mathsf{M}'$ we have
\begin{align}& d\cdot \hat{\Pi} \Pi\Pi_{m} \mathbbm{1}_{\mathsf{M}'}(\rho_m) \Pi_{m}\Pi\hat{\Pi}\allowdisplaybreaks\notag\\
&=  d\cdot \hat{\Pi} \Pi\Pi_{m} \rho_m \Pi_{m}\Pi\hat{\Pi}\allowdisplaybreaks\notag\\
&\leq \hat{\Pi}\label{wilderie1}\end{align}
as a consequence of the inequality 
$A^\dag B A\leq A^\dag A$ which is valid whenever $B\leq id$.

By (\ref{wilderie1}) and the fact that $d\cdot0^{\mathcal{V}}\leq \hat{\Pi}$, we have for all $m\in\mathsf{M}$.
\[0^{\mathcal{V}}\leq d\cdot \hat{\Pi} \Pi\Pi_{m} \mathbbm{1}_{\mathsf{M}'}(\rho_m) \Pi_{m}\Pi\hat{\Pi}\leq \hat{\Pi}\text{ .}\]

Now we apply the Operator Chernoff Bound (cf. \cite{Wil})
on the set of operator $\{d\mathbbm{1}_{\mathsf{M}'}(\rho_m):m\in\mathsf{M}\}\}$ and
the subspace 
spanned by the eigenvectors of $\sum_{m\in \mathsf{M}'}p(m) \Pi\Pi_{m} \rho_m \Pi\Pi_{m}$  whose corresponding
eigenvalues are greater than $p(\mathsf{M}')\frac{\epsilon}{D}$; here $\hat{\Pi}$
acts as the identity on the subspace. 

By (\ref{wilderie2})
we obtain
\begin{align}&
Pr \biggl((1-\epsilon)  \sum_{m\in \mathsf{M}}p(m) \hat{\Pi}\Pi\cdot\Pi_{m}\cdot \mathbbm{1}_{\mathsf{M}'}(\rho_m) \cdot\Pi_{m}\cdot\Pi\hat{\Pi}
\allowdisplaybreaks\notag\\
&~~~~~~\leq L^{-1}\sum_{i=1}^{L}\hat{\Pi}\Pi\cdot\Pi_{X_i}\cdot (\mathbbm{1}_{\mathsf{M}'}(X_i))\cdot\Pi_{X_i}\cdot\Pi\hat{\Pi} 
\allowdisplaybreaks\notag\\
&~~~~~~\leq (1+\epsilon)  L^{-1}\sum_{i=1}^{L}\hat{\Pi}\Pi\cdot\Pi_{X_i}\cdot (\mathbbm{1}_{\mathsf{M}'}(X_i))\cdot\Pi_{X_i}\cdot\Pi\hat{\Pi}
\allowdisplaybreaks\notag\\
  &=Pr \biggl(d(1-\epsilon)  \sum_{m\in \mathsf{M}}p(m) \hat{\Pi}\Pi\cdot\Pi_{m}\cdot \mathbbm{1}_{\mathsf{M}'}(\rho_m) \cdot\Pi_{m}\cdot\Pi\hat{\Pi}
\allowdisplaybreaks\notag\\
&~~~~~~~~~~\leq  dL^{-1}\sum_{i=1}^{L}\hat{\Pi}\Pi\cdot\Pi_{X_i}\cdot (\mathbbm{1}_{\mathsf{M}'}(X_i))\cdot\Pi_{X_i}\cdot\Pi\hat{\Pi} \allowdisplaybreaks\notag\\
&~~~~~~~~~~\leq d(1+\epsilon) \sum_{m\in \mathsf{M}}p(m) \hat{\Pi}\Pi\cdot\Pi_{m}\cdot \mathbbm{1}_{\mathsf{M}'}(\rho_m) \cdot\Pi_{m}\cdot\Pi\hat{\Pi}
\biggr) \allowdisplaybreaks\notag\\
& \geq 1- 2D \exp
\left( -p(\mathsf{M}')\frac{\epsilon^3Ld}{2\ln 2D} \right) \text{ .}\label{useocbound}
\end{align}\vspace{0.2cm}

\it ii)  Upper Bound for  $\lVert\sum_{m\in \mathsf{M}}p(m)\mathbbm{1}_{\mathsf{M}'}(\rho_m)
-  \sum_{m\in \mathsf{M}}p(m) \hat{\Pi}\Pi\Pi_{m} \mathbbm{1}_{\mathsf{M}'}(\rho_m) \Pi_{m}\Pi\hat{\Pi}
\rVert_1$  \rm\vspace{0.2cm}

Let $\sum_{i}\lambda_i | i\rangle\langle i|$ be a spectral decomposition of
$ \sum_{m\in \mathsf{M}'}\frac{p(m)}{p(\mathsf{M}') } \Pi\Pi_{m} \rho_m \Pi_{m}\Pi$.
In view of the fact that
$\hat{\Pi}$ is the  projector onto the subspace spanned by the eigenvectors of the density operator
$ \sum_{m\in \mathsf{M}'}\frac{p(m)}{p(\mathsf{M}') } \Pi\Pi_{m} \rho_m \Pi_{m}\Pi$  whose corresponding
eigenvalues are greater than $\frac{\epsilon}{D}$,
we have
\begin{align*}&\mathrm{tr}\left( \sum_{m\in \mathsf{M}}\frac{p(m)}{p(\mathsf{M}') } \Pi\cdot\Pi_{m}\cdot 
\mathbbm{1}_{\mathsf{M}'}(\rho_m) \cdot\Pi_{m}\cdot\Pi\right)\allowdisplaybreaks\notag\\
 &-\mathrm{tr}\left( \sum_{m\in \mathsf{M}}\frac{p(m)}{p(\mathsf{M}') } \hat{\Pi}\Pi\cdot\Pi_{m}\cdot 
\mathbbm{1}_{\mathsf{M}'}(\rho_m) \cdot\Pi_{m}\cdot\Pi\hat{\Pi}\right)\allowdisplaybreaks\notag\\
&=\sum_{\lambda_i \geq \frac{\epsilon}{D}}\lambda_i \allowdisplaybreaks\notag\\
 &\leq  \epsilon  \text{ .}
\end{align*}

We apply the gentle operator lemma (cf.  \cite{Wil}) to obtain
\begin{align}&\lVert \sum_{m\in \mathsf{M}}p(m) \Pi\cdot\Pi_{m}\cdot 
\mathbbm{1}_{\mathsf{M}'}(\rho_m) \cdot\Pi_{m}\cdot\Pi -\sum_{m\in \mathsf{M}}p(m)\hat{\Pi}\Pi\cdot\Pi_{m}\cdot 
\mathbbm{1}_{\mathsf{M}'}(\rho_m) \cdot\Pi_{m}\cdot\Pi\hat{\Pi}\rVert_1\allowdisplaybreaks\notag\\
&=p(\mathsf{M}') \lVert \sum_{m\in \mathsf{M}}\frac{p(m)}{p(\mathsf{M}') } \Pi\cdot\Pi_{m}\cdot 
\mathbbm{1}_{\mathsf{M}'}(\rho_m) \cdot\Pi_{m}\cdot\Pi -\sum_{m\in \mathsf{M}}\frac{p(m)}{p(\mathsf{M}') } \hat{\Pi}\Pi\cdot\Pi_{m}\cdot 
\mathbbm{1}_{\mathsf{M}'}(\rho_m) \cdot\Pi_{m}\cdot\Pi\hat{\Pi}\rVert_1\allowdisplaybreaks\notag\\
 &\leq  2\sqrt{\epsilon+2\sqrt{\epsilon}} \allowdisplaybreaks\notag\\
&\leq 4\sqrt[4]{\epsilon}\text{ .}\label{wilderie5}
\end{align}\vspace{0.2cm}

When  $\{\rho_1,\cdots,\rho_L\}$ fulfills 
\begin{align*}&\lVert L^{-1}
\sum_{i=1}^{L}\hat{\Pi}\Pi\cdot\Pi_{i}\cdot (\mathbbm{1}_{\mathsf{M}'}(\rho_i))\cdot\Pi_{i}\cdot\Pi\hat{\Pi} \allowdisplaybreaks\notag\\
&- \sum_{m\in \mathsf{M}}p(m) \hat{\Pi}\Pi\cdot\Pi_{m}\cdot 
\mathbbm{1}_{\mathsf{M}'}(\rho_m) \cdot\Pi_{m}\cdot\Pi\hat{\Pi}\rVert_1\allowdisplaybreaks\notag\\
 &\leq  \epsilon  
\end{align*}(i.e.  we assume the event considered in (\ref{useocbound}) occurs and thus (\ref{wilderie3}) holds), then
by (\ref{tnaresesesl}) and  (\ref{wilderie5})
it holds that
\begin{align}&\lVert L^{-1}
\sum_{i=1}^{L}\hat{\Pi}\Pi\cdot\Pi_{i}\cdot (\mathbbm{1}_{\mathsf{M}'}(\rho_i))\cdot\Pi_{i}\cdot\Pi\hat{\Pi}  - \sum_{m\in \mathsf{M}}p(m) 
\mathbbm{1}_{\mathsf{M}'}(\rho_m) \rVert_1\allowdisplaybreaks\notag\\
& \leq  \epsilon + 10 \sqrt[4]{\epsilon}\notag\\
&\leq 11 \sqrt[4]{\epsilon}\text{ .}\label{wilderie6}
\end{align}\vspace{0.2cm}

\it iii)  Upper Bound for  $\lVert L^{-1}
\sum_{i=1}^{L}\Pi\Pi_{i}(\mathbbm{1}_{\mathsf{M}'}(\rho_i))\Pi_{i}\Pi
- L^{-1}
\sum_{i=1}^{L}\hat{\Pi}\Pi\Pi_{i}(\mathbbm{1}_{\mathsf{M}'}(\rho_i))\Pi_{i}\Pi\hat{\Pi} 
\rVert_1$  \rm\vspace{0.2cm}

When  the event considered in (\ref{useocbound}) is true,
i.e., when (\ref{wilderie6})
holds, then
by (\ref{mtlsmimmpm})
\begin{align*}&
\mathrm{tr}\left(L^{-1}
\sum_{i=1}^{L}\hat{\Pi}\Pi\cdot\Pi_{i}\cdot (\mathbbm{1}_{\mathsf{M}'}(\rho_i))
\cdot\Pi_{i}\cdot\Pi\hat{\Pi} \right)\allowdisplaybreaks\notag\\
&\geq p(\mathsf{M}')- 11 \sqrt[4]{\epsilon}
\text{ .}
\end{align*}

We apply the gentle operator lemma (cf.  \cite{Wil}) to obtain
\begin{align}&
\lVert L^{-1}
\sum_{i=1}^{L} \hat{\Pi} \Pi\cdot\Pi_{i}\cdot (\mathbbm{1}_{\mathsf{M}'}(\rho_i))
\cdot\Pi_{i}\cdot\Pi\hat{\Pi} - L^{-1}
\sum_{i=1}^{L}\Pi\cdot\Pi_{i}\cdot (\mathbbm{1}_{\mathsf{M}'}(\rho_i))
\cdot\Pi_{i}\cdot\Pi \rVert_1\allowdisplaybreaks\notag\\
&\leq 2\sqrt{1-p(\mathsf{M}')+ 11 \sqrt[4]{\epsilon}}\notag\\
&\leq 2\sqrt{1-p(\mathsf{M}')}+ 22 \sqrt[8]{\epsilon}
\text{ .}
\end{align} The last inequality holds because
$\sqrt{a+b}\leq \sqrt{a}+\sqrt{b}$ for positive $a$ and $b$.
\vspace{0.2cm}

\it iv) Upper Bound for $\lVert L^{-1}
\sum_{i=1}^{L}\Pi\Pi_{i}\rho_i\Pi_{i}\Pi
- L^{-1}
\sum_{i=1}^{L}\Pi\Pi_{i}(\mathbbm{1}_{\mathsf{M}'}(\rho_i))\Pi_{i}\Pi
\rVert_1$  \rm\vspace{0.2cm}

In view of the fact that $\Pi$ and
$\Pi_{i}$ are projection matrices
for every $\rho_i$ $\in\{\rho_1,\cdots,\rho_L\}$
it holds that
\[\mathrm{tr}(\Pi_{i}\rho_l\Pi_{i})\leq \mathrm{tr}(\rho_l)=1\] and
\begin{align*}&\mathrm{tr}(L^{-1}\sum_{i=1}^{L} \Pi\Pi_{i}\rho_l\Pi_{i}\Pi)\\
&\leq \mathrm{tr}(L^{-1}\sum_{i=1}^{L} \Pi_{i}\rho_l\Pi_{i})\\
&\leq 1\text{ .}\end{align*}

When  $\{\rho_1,\cdots,\rho_L\}$ fulfills 
\begin{align*}&\lVert L^{-1}
\sum_{i=1}^{L}\Pi\cdot\Pi_{i}\cdot (\mathbbm{1}_{\mathsf{M}'}(\rho_i))\cdot\Pi_{i}\cdot\Pi - \sum_{m\in \mathsf{M}}p(m) 
\mathbbm{1}_{\mathsf{M}'}(\rho_m) \rVert_1\allowdisplaybreaks\notag\\
 &\leq   2\sqrt{1-p(\mathsf{M}')}+ 
20 \sqrt[8]{\epsilon} \text{ ,}
\end{align*} i.e., we assume that the event considered in (\ref{useocbound}) is true, and
then
by (\ref{mtlsmimmpm}) and the triangle inequality, 
we have
\begin{align}&\mathrm{tr}\left(\Pi\cdot\Pi_{i}\cdot (\mathbbm{1}_{\mathsf{M}'}(\rho_i))\cdot\Pi_{i}\cdot\Pi\right)\allowdisplaybreaks\notag\\
&\geq p(\mathsf{M}')-
 2\sqrt{1-p(\mathsf{M}')}- 
20 \sqrt[8]{\epsilon}
\text{ .}\end{align}
Since 
\begin{align}& L^{-1}
\sum_{i=1}^{L} \Pi\cdot\Pi_{i}\cdot \rho_i\cdot\Pi_{i}\cdot\Pi\allowdisplaybreaks\notag\\
&= L^{-1}
\sum_{i=1}^{L} \Pi\cdot\Pi_{i}\cdot \mathbbm{1}_{\mathsf{M}'}(\rho_i)\cdot\Pi_{i}\cdot\Pi\allowdisplaybreaks\notag\\
 &+ L^{-1}\sum_{i\notin \mathsf{M}'}
\Pi\cdot\Pi_{i}\cdot \rho_i\cdot\Pi_{i}\cdot\Pi
\text{ ,}\label{mtlsmimmpm1}\end{align} 
we have 
\begin{align}&\lVert L^{-1}\sum_{i\notin \mathsf{M}'}
\Pi\cdot\Pi_{i}\cdot \rho_i\cdot\Pi_{i}\cdot\Pi\rVert_1\allowdisplaybreaks\notag\\
&=\mathrm{tr}\left( L^{-1}\sum_{i\notin \mathsf{M}'}
\Pi\cdot\Pi_{i}\cdot \rho_i\cdot\Pi_{i}\cdot\Pi\right)\allowdisplaybreaks\notag\\
&\leq 1-p(\mathsf{M}')+
 2\sqrt{1-p(\mathsf{M}')}+ 
20 \sqrt[8]{\epsilon}
\text{ ,}\label{mtlsmimmpm2}\end{align}

which implies
\begin{align}&\lVert L^{-1}
\sum_{i=1}^{L}\Pi\cdot\Pi_{i}\cdot \rho_i\cdot\Pi_{i}\cdot\Pi  - \sum_{m\in \mathsf{M}}p(m) 
\mathbbm{1}_{\mathsf{M}'}(\rho_m) \rVert_1\allowdisplaybreaks\notag\\
 &\leq 1-p(\mathsf{M}')+ 4\sqrt{1-p(\mathsf{M}')}+ 
42 \sqrt[8]{\epsilon} \text{ .}\label{pbllsilpcpx5}
\end{align}

By (\ref{pbllsilpcpx5}), we have
\begin{align}&
Pr \biggl( \lVert L^{-1}
\sum_{i=1}^{L}\Pi\cdot\Pi_{X_i}\cdot X_i \cdot\Pi_{X_i} \cdot\Pi
- \sum_{m\in \mathsf{M}}p(m) \cdot \mathbbm{1}_{\mathsf{M}'}(\rho_m) \rVert_1\allowdisplaybreaks\notag\\
&~~~~~\leq   1-p(\mathsf{M}')+  4\sqrt{1-p(\mathsf{M}')}+ 
42 \sqrt[8]{\epsilon} 
 \biggr)\allowdisplaybreaks\notag\\
 & \geq 1- 2D \exp
\left( -p(\mathsf{M}')\frac{\epsilon^3Ld}{2\ln 2D} \right) \text{ .}\label{pbllsilpcpx2}
\end{align} \qed\end{proof}\vspace{0.2cm}

By
(\ref{te2}),  we have
\begin{align}&\mathrm{tr}(\Pi_{\mathsf{V}(\cdot,p), \alpha }(t^n))\allowdisplaybreaks\notag\\
&\leq 2^{n (S(\mathsf{V}(\cdot,p)\mid q) +\delta(\alpha))}\allowdisplaybreaks\notag\\
&= 2^{n (\sum_{t}q(t)\mathsf{V}(t,p) +\delta(\alpha))}\allowdisplaybreaks\notag\\
&= 2^{n (\sum_{t}q(t)S(V_{t}(p)) +\delta(\alpha))}\text{ .}\label{eq_5a2}\end{align}

Furthermore, for all $x^n$ it holds that

\begin{align}&
\Pi_{\mathsf{V},\alpha}(t^n, x^n)
 V_{{t^n}}(x^n)  \Pi_{\mathsf{V},\alpha}(t^n,x^n)\allowdisplaybreaks\notag\\
&\leq 2^{-n(S(\mathsf{V}|r) + \delta(\alpha)')}\Pi_{\mathsf{V},\alpha}(t^n,x^n)\allowdisplaybreaks\notag\\
&= 2^{-n(\sum_{t,x}r(t,x)S(\mathsf{V}(t,x)) + \delta(\alpha)')}\Pi_{\mathsf{V},\alpha}(t^n,x^n)
\text{ .}\label{eq_5a}
\end{align} 
\vspace{0.15cm}

We define 
\[\theta':=\left\{t\in\theta:nq(t)\geq \sqrt{n}\right\}\text{ .}\]
By properties of classical typical set (cf. \cite{Win}). there is a positive $\hat{\beta}(\alpha)$
such that
\begin{equation} \underset{p'}{Pr} \biggl(x^n\in \biggl\{x^n\in{\mathsf{A}}^n:
(x_{\mathtt{I}_t})\in \mathsf{T}^{nq(t)}_{p,\delta} ~ \forall 
t\in\theta'\biggr\}\biggr)\geq \left(1-2^{-\sqrt{n}\hat{\beta}(\alpha)}\right)^{|\theta|}\geq 1-2^{-\sqrt{n}
\frac{1}{2}\hat{\beta}(\alpha)}\text{ ,}\label{uppbxm}\end{equation} 
where $\mathtt{I}_t$ $:=$ $\{i\in\{1,\cdots,n\}: t_i = t\}$ 
is an indicator set that selects the indices $i$ in the sequence $t^n$
$=$ $(t_1,\cdots,t_n)$.

We denote the set $\{x^n: (x_{\mathtt{I}_t})\in \mathsf{T}^{nq(t)}_{p,\delta} ~ \forall 
t\in\theta'\}$ $\subset {\mathsf{A}}^n$ by $\mathsf{M}_{t^n}$.
For all $x^n\in\mathsf{M}_{t^n}$, if $n$ is sufficiently
large,
we have
\begin{align}&\left\vert\sum_{t,x}r(t,x)S(\mathsf{V}(t,x))-
\sum_{t}q(t)S(V_{t}|p)\right\vert\allowdisplaybreaks\notag\\
&\leq\left\vert\sum_{t\in\theta',x}r(t,x)S(\mathsf{V}(t,x))-
\sum_{t\in\theta'}q(t)S(V_{t}|p)\right\vert\allowdisplaybreaks\notag\\ 
&+ \left\vert\sum_{t\notin\theta',x}r(t,x)S(\mathsf{V}(t,x))-
\sum_{t\notin\theta'}q(t)S(V_{t}|p)\right\vert\allowdisplaybreaks\notag\\
&\leq\sum_{t\in\theta'}\left\vert\sum_{x}r(t,x)S(\mathsf{V}(t,x))-
q(t)S(V_{t}|p)\right\vert+2|\theta|\frac{1}{\sqrt{n}}C \notag\\ 
&\leq 2|\theta|\frac{\delta}{n}C  +2|\theta|\frac{1}{\sqrt{n}}C
\text{ ,}\label{uppbxmn2} \end{align} where $C:=\max_{t\in\theta}\max_{x\in \mathsf{A}}(S(\mathsf{V}(t,x))+
S(V_{t}|p))$.




We set $\Theta_{t^n}:=  \sum_{x^n \in \mathsf{M}_{t^n}}
{p}(x^n) Q_{{t^n}}(x^n)$. For given $z^n\in\mathsf{M}_{t^n}$ and ${t^n}\in\theta^n$, $\langle
z^n|\Theta_{t^n}|z^n\rangle$  is the expected  value of $\langle z^n|
Q_{{t^n}}(x^n) |z^n\rangle$ under the condition $x^n \in
\mathsf{M}_{t^n}$.\vspace{0.15cm}

We choose a positive $\bar{\beta}(\alpha)$
such that  $\bar{\beta}(\alpha)\leq \min(2^{-n\beta(\alpha)},
2^{-n\beta(\alpha)'})$, and 
set
$\epsilon := 2^{-n\bar{\beta}(\alpha)}$.
In view of (\ref{eq_5a}), 
we now apply Lemma \ref{cov3}, where
we consider the set $\mathsf{M}_{t^n}$ $\subset {\mathsf{A}}^n$:
If  $n$ is sufficiently large, for all $j$ 
 we have

\begin{align}&Pr \left( \lVert \sum_{l=1}^{L_{n}} \frac{1}{L_{n}} Q_{t^n}(X_{j,l}) -
\Theta_{t^n} \rVert_1 >  2^{-\sqrt{n}
\frac{1}{8}\hat{\beta}(\alpha)}  +40 \sqrt[8]{\epsilon} 
\right)\allowdisplaybreaks\notag\\
&\leq 2^{n (\sum_{t,x}r(t,x)S(\mathsf{V}(t,x)) +\delta(\alpha))} \notag\\
&\cdot \exp \left( -L_{n}\frac{\epsilon^3}{2\ln  2}  (1-2^{-\sqrt{n}\frac{1}{2}\hat{\beta}(\alpha)})
 \cdot 2^{n(\sum_{t}q(t)S(V_{t}(p))-\sum_{t}q(t)S(V_{t}|p)) +\delta(\alpha)+ \delta(\alpha)' + 2|\theta|\frac{\delta}{n}C
+ 2|\theta|\frac{1}{\sqrt{n}}C}
 \right)\allowdisplaybreaks\notag\\
&= 2^{n (\sum_{t,x}r(t,x)S(\mathsf{V}(t,x)) +\delta(\alpha)} \notag\\
&\cdot \exp \left( -L_{n}\frac{\epsilon^3}{2\ln  2}  \cdot (1-2^{-\sqrt{n}\frac{1}{2}\hat{\beta}(\alpha)})
2^{n (-\sum_{t}q(t)\chi(p;Z_{t}) + \delta(\alpha)+\delta(\alpha)'+ 2|\theta|\frac{\delta}{n}C
+ 2|\theta|\frac{1}{\sqrt{n}}C )} \right)\text{
.}\label{eq_5b+}\end{align}  The  equality  holds since
$S(V_{t}(p))-S(V_{t}|p)  =\chi(p;Z_{t})$. 

Furthermore,
 \begin{align}&Pr \left( \lVert \sum_{l=1}^{L_{n}} \frac{1}{L_{n}} Q_{t^n}(X_{j,l}) -
\Theta_{t^n} \rVert_1 >   2^{-\sqrt{n}
\frac{1}{8}\hat{\beta}(\alpha)}  +40 \sqrt[8]{\epsilon} ~ \forall t^n ~ \forall j \right)\allowdisplaybreaks\notag\\
&\leq J_n {|\theta|}^n 2^{n (\sum_{t,x}r(t,x)S(\mathsf{V}(t,x)) +\delta(\alpha)} \notag\\
&\cdot \exp \left( -L_{n}\frac{\epsilon^3}{2\ln  2} (1-2^{-\sqrt{n}\frac{1}{2}\hat{\beta}(\alpha)})
2^{n (-\sum_{t}q(t)\chi(p;Z_{t}) + \delta(\alpha)+\delta(\alpha)'+ 2|\theta|\frac{\delta}{n}C
+ 2|\theta|\frac{1}{\sqrt{n}}C )} \right)\text{
.}\label{eq_5b}\end{align} 
\vspace{0.2cm}

Let $\phi_{t}^{j}$ be the quantum state at the output of wiretapper's channel when
the channel state is $t$ and $j$ has
been sent.
We have
\begin{align*}&\sum_{t\in \theta}q(t)\chi\left(p;Z_{t}\right)- \chi\left(p;\sum_{t}q(t)Z_{t}\right)\\
&= \sum_{t\in \theta}q(t)S\left(\sum_{j=1}^{J_n}\frac{1}{J_n} \phi_{t}^{j}\right)
-\sum_{t\in \theta}\sum_{j=1}^{J_n}q(t)\frac{1}{J_n}S\left( \phi_{t}^{j}\right)\\
&-S\left(\frac{1}{J_n}\sum_{t\in \theta}\sum_{j=1}^{J_n} q(t)\phi_{t}^{j}\right)
+\sum_{j=1}^{J_n}\frac{1}{J_n}S\left(\sum_{t\in \theta}q(t) \phi_{t}^{j}\right)\text{ .}
\end{align*}
 Let $H^{\mathfrak{T}}$ be a $\left|\theta\right|$-dimensional Hilbert space
spanned by an orthonormal basis $\{|t\rangle : t = 1, \cdots, \left|\theta\right|\}$. 
Let $H^{\mathfrak{J}}$ be a $J_n$ dimensional Hilbert space spanned by an orthonormal basis 
$\{|j\rangle : j = 1, \cdots, J_n\}$. 
We define
\[\varphi^{\mathfrak{J}\mathfrak{T}H^{n}}:=\frac{1}{J_n}\sum_{j=1}^{J_n}\sum_{t\in \theta}q(t)
|j\rangle\langle j|\otimes|t\rangle\langle t|\otimes
 \phi_{t}^{j}\text{ .}\]
We have
\[\varphi^{\mathfrak{J}H^{n}}=\mathrm{tr}_{\mathfrak{T}}\left(\varphi^{\mathfrak{J}\mathfrak{T}H^{n}}\right)=
\frac{1}{J_n}\sum_{j=1}^{J_n}\sum_{t\in \theta}q(t)
|j\rangle\langle j|\otimes
 \phi_{t}^{j}\text{ ;}\]
\[\varphi^{\mathfrak{T}H^{n}}=\mathrm{tr}_{\mathfrak{J}}\left(\varphi^{\mathfrak{J}\mathfrak{T}H^{n}}\right)=
\frac{1}{J_n}\sum_{j=1}^{J_n}\sum_{t\in \theta}q(t)|t\rangle\langle t|\otimes
 \phi_{t}^{j}\text{ ;}\]
\[\varphi^{H^{n}}=\mathrm{tr}_\mathfrak{J}{\mathfrak{T}}\left(\varphi^{\mathfrak{J}\mathfrak{T}H^{n}}\right)=
\frac{1}{J_n}\sum_{j=1}^{J_n}\sum_{t\in \theta}q(t)
 \phi_{t}^{j}\text{ .}\]
Thus,
\[S(\varphi^{\mathfrak{J}H^{n}}) = H(R_{uni})+ \frac{1}{J_n}\sum_{j=1}^{J_n}S\left(\sum_{t\in \theta}q(t)
\phi_{t}^{j}\right)\text{ ;}\]
\[S(\varphi^{\mathfrak{T}H^{n}})=H(Y_q)+  \sum_{t\in \theta}q(t)S\left(\frac{1}{J_n}\sum_{j=1}^{J_n}
 \phi_{t}^{j}\right)\text{ ;}\]
\[S(\varphi^{\mathfrak{J}\mathfrak{T}H^{n}})= H(R_{uni})+H(Y_q)+ \frac{1}{J_n}\sum_{j=1}^{J_n}\sum_{t\in \theta}q(t)S\left(
 \phi_{t}^{j}\right)\text{ ,}\]
where $Y_q$ is a random variable on $\theta$ with distribution $q(t)$.

 By strong subadditivity of von Neumann entropy, it holds that $S(\varphi^{\mathfrak{J}H^{n}}) + S(\varphi^{\mathfrak{T}H^{n}})$
$\geq$ $S(\varphi^{H^{n}})+S(\varphi^{\mathfrak{j}\mathfrak{T}H^{n}})$, therefore
\begin{equation}\sum_{t}q(t)\chi\left(p;Z_{t}\right)- \chi\left(p;\sum_{t}q(t)Z_{t}\right)
\geq 0\text{ .}\label{stqtclpztr}\end{equation}\vspace{0.2cm}

For an arbitrary $\zeta$,
we define $L_{n} = \lceil 2^{\max_{t^n}\chi(p;Z_{t^n})+n\zeta} \rceil$,
and choose a suitable $\alpha$, $\bar{\beta}(\alpha)$, and 
sufficiently large $n$ such that
$6\bar{\beta}(\alpha)$ + $2\delta(\alpha)$ $+2\delta(\alpha)'$ $+ 2|\theta|\frac{\delta}{n}C$
$+ 2|\theta|\frac{1}{\sqrt{n}}C$ $\leq\zeta$.
By (\ref{stqtclpztr}), if $n$ is  sufficiently large, we have
$L_{n} \geq \lceil 2^{n(\sum_{t}q(t)\chi(p;Z_{t})+\zeta)} \rceil$
and
\begin{align*}&
L_{n}\frac{\epsilon^3}{2\ln  2} (1-2^{-\sqrt{n}\frac{1}{2}\hat{\beta}(\alpha)})
2^{n (-\sum_{t}q(t)\chi(p;Z_{t}) + \delta(\alpha)+\delta(\alpha)'+ 2|\theta|\frac{\delta}{n}C
+ 2|\theta|\frac{1}{\sqrt{n}}C )} > 2^{\frac{1}{2}n\zeta}\text{ .}
\end{align*} 

When $n$ is sufficiently large 
for any positive $\vartheta$ it holds that
\begin{align*}&J_n{|\theta|}^n 2^{n (\sum_{t,x}r(t,x)S(\mathsf{V}(t,x)) +\delta(\alpha)} \exp(- 2^{\frac{1}{4}n\zeta})\\
&\leq 2^{-n\vartheta}
\end{align*}
and 
\[ 2^{-\sqrt{n}
\frac{1}{8}\hat{\beta}(\alpha)}  +40 \sqrt[8]{\epsilon}  \leq 2^{-\sqrt{n}
\frac{1}{16}\hat{\beta}(\alpha)} \text{ .}\]

Thus for sufficiently large $n$ 
we have
\begin{align} &Pr \biggl( \lVert \sum_{l=1}^{L_{n}} \frac{1}{L_{n}} Q_{t^n}(X_{j,l}) -\Theta_{t^n}
\rVert_1 \leq  2^{-\sqrt{n}
\frac{1}{16}\hat{\beta}(\alpha)} \text{ }\forall t^n  ~ \forall j 
 \biggr)\allowdisplaybreaks\notag\\
&\geq 1- 2^{n\vartheta}
\label{eq_6blslln}
\end{align} for any positive $\varrho $.
\vspace{0.3cm}

Now we have $J_n\cdot L_{n} < 2^{n(\min_{s}\chi(p;B_{s})-\mu)}$.

In 
\cite{Bj/Bo} and \cite{Bj/Bo/No}, the following was
 shown (using results of \cite{Ha/Na}).  Let $\{{\dot{X}}_{j,l}\}_{j \in \{1, \dots, J_n\}, l \in
\{1, \dots, L_{n}\}}$ be a family of random variables taking value
according to $\dot{p} \in \mathsf{P}(\mathsf{A}^n)$.
If $n$ is sufficiently large,
and if $J_n\cdot L_{n} \leq 2^{\min_{s}n(\chi(\dot{p};B_{s})-\mu)}$
for an  arbitrary positive $\mu$
there exists a 
projection $q_{x^n}$ on $H$ for every $x^n\in \mathsf{A}^n$
and positive constants
 $\beta$ and $\gamma$, such that  for any $(s,j,l)\in \theta \times
\{1,\dots,J_n\} \times \{1,\dots,L_n\}$, it holds that
\begin{equation}\underset{\dot{p}}{Pr} \left[\mathrm{tr}\left(\overline{W}_s^{\otimes n}({\dot{X}}_{j,l})D_{{\dot{X}}_{j,l}}\right) \geq 1-|\overline{\theta}|2^{-{n}^{1/16}\beta}\right] > 1-2^{-n\gamma } \text{
,}\label{qnocsig4'}\end{equation}
where for $ j \in  \{1,\dots,J_n\}, l \in \{1,\dots,L_n\}$,
we have
\[D_{{\dot{X}}_{j,l}}:= \left(\sum_{j',l'} q_{{\dot{X}}_{j',l'}}\right)^{-\frac{1}{2}} q_{{\dot{X}}_{j,l}}  \left(\sum_{j',l'} q_{{\dot{X}}_{j',l'}}\right)^{-\frac{1}{2}}\text{ .}\]
Notice that by this definition, for any realization $\{{\dot{x}}_{j,l}:j,l\}$ of $\{{\dot{X}}_{j,l}:j,l\}$ it holds that
$\sum_{j=1}^{J_n}\sum_{l=1}^{L_n}  D_{{\dot{x}}_{j,l}} \leq \mathrm{id}_{H^{\otimes n}}$.

(Actually in \cite{tobepublished},
it was shown that
there exists
 a collection of positive semidefinite operators
$\{D_{s,{\dot{X}}_{j,l}}: s \in \overline{\theta}, j \in  \{1,\dots,J_n\}, l \in \{1,\dots,L_n\}    \}$ such that  for any $s$, $j$, and $l$ it holds that
\[Pr \left[\mathrm{tr}\left(\overline{W}_s^{\otimes n}({\dot{X}}_{j,l})D_{s,{\dot{X}}_{j,l}}\right) \geq 1-2^{|\overline{\theta}|}2^{-n\beta}\right] > 1-2^{-n\gamma } \text{
,}\]
and for any realization $\{{\dot{x}}_{j,l}:j,l\}$ of $\{{\dot{X}}_{j,l}:j,l\}$ it holds that
$\sum_{s\in \overline{\theta}}\sum_{j=1}^{J_n}\sum_{l=1}^{L_n}  D_{s,{\dot{x}}_{j,l}} \leq \mathrm{id}_{H^{\otimes n}}$.)\vspace{0.2cm}

For any given $s\in \theta$, it holds  that
\begin{align*} & \overline{W}_{s}^{\otimes n}({p}^n)- \overline{W}_{s}^{\otimes n}({p'}^n)\\
& =\left(1-\frac{1}{P
(\mathsf{T}^n_{p,\delta})}\right)\sum_{a^n\in \mathsf{T}^n_{p,\delta}} p^n(a^n)\overline{W}_{s}^{\otimes n}(a^n)
+\sum_{a^n\notin \mathsf{T}^n_{p,\delta}} p^n(a^n)\overline{W}_s^{\otimes n}(a^n)\text{ .}\end{align*}
Thus, we have $\left\vert\mathrm{tr}\left(\overline{W}_{s}^{\otimes n}(p^n)
-\overline{W}_{s}^{\otimes n}({p'}^n)\right)\right\vert$ $\leq$
$2P(\mathsf{T}^n_{p,\delta})$ $\leq$ $2^{-n\eta(\delta)}$ for a positive $\eta(\delta)$.

\begin{lemma}[Fannes-Audenaert  Ineq.,
 cf. \cite{Fa}, \cite{Au}]\label{eq_9}  
Let $\Phi$ and $\Psi$ be two  quantum states in a
$d$-dimensional complex Hilbert space and
$\|\Phi-\Psi\|_1 \leq \mu < \frac{1}{e}$, then
\begin{equation} |S(\Phi)-S(\Psi)| \leq \mu \log (d-1)
+ h(\mu) \text{
,}\label{faaudin}\end{equation}
where $h(\nu) := -\nu \log \nu - (1- \nu) \log (1-\nu)$
for $\nu\in [0,1]$.
\end{lemma}\vspace{0.15cm}

 The Fannes  Inequality was first introduced in \cite{Fa} where it has been
shown that $|S(\Phi)-S(\Psi)| \leq \mu \log d - \mu
\log \mu $. In \cite{Au}, the result of \cite{Fa} has been
improved, and (\ref{faaudin}) has been proved.\vspace{0.15cm}

By Lemma \ref{eq_9} for any positive $\omega$, if $n$ is 
sufficiently large, we have 

\begin{align*} & \left\vert S\left(\overline{W}_{s}^{\otimes n}({p}^n)
\right) -S\left(\overline{W}_{s}^{\otimes n}({p'}^n)\right)\right\vert\\
& \leq 2^{-n\eta(\delta)}\log (d^n-1)
+ h(2^{-n\eta(\delta)})\\
&\leq \omega\text{ .}\end{align*}\vspace{0.15cm}

Furthermore, we have
\begin{align*} & \Bigl\vert \sum_{a^n\in \mathsf{T}^n_{p,\delta}} {p'}^n(a^n)S\left(\overline{W}_{s}^{\otimes n}(a^n)\right)
-\sum_{a^n\in \mathsf{T}^n_{p,\delta}} {p'}^n(a^n)S\left(\overline{W}_{s}^{\otimes n}(a^n)\right) \Bigr\vert\\
&=\Bigl\vert\left(1-\frac{1}{P
(\mathsf{T}^n_{p,\delta})}\right)\sum_{a^n\in \mathsf{T}^n_{p,\delta}} p^n(a^n)S\left(\overline{W}_{s}^{\otimes n}(a^n)\right)
+\sum_{a^n\notin \mathsf{T}^n_{p,\delta}} p^n(a^n)S\left(\overline{W}_s^{\otimes n}(a^n)\right)\Bigr\vert\\
&\leq 2P(\mathsf{T}^n_{p,\delta})\max_{a^n\in \mathsf{A}^n}S\left(\overline{W}_s^{\otimes n}(a^n)\right)\\
&\leq \omega\end{align*} for any positive $\omega$, if $n$ is 
sufficiently large.\vspace{0.15cm}

We now have \begin{align*} &\left\vert \chi(p;B_{s}^{\otimes n}) - \chi(p';B_{s}^{\otimes n}) \right\vert\\
&\leq \left\vert S\left(\overline{W}_{s}^{\otimes n}({p}^n)
\right) -S\left(\overline{W}_{s}^{\otimes n}({p'}^n)\right)\right\vert\\
&+ \Bigl\vert \sum_{a^n\in \mathsf{T}^n_{p,\delta}} {p'}^n(a^n)S\left(\overline{W}_{s}^{\otimes n}(a^n)\right)
-\sum_{a^n\in \mathsf{T}^n_{p,\delta}} {p'}^n(a^n)S\left(\overline{W}_{s}^{\otimes n}(a^n)\right) \Bigr\vert\\
&\leq 2\omega\end{align*} for any positive $\omega$, if $n$ is 
sufficiently large.

Thus, when $J_n\cdot L_{n} < 2^{\min_{s}n\chi(p;B_{s})-\mu}$ holds, we also have
\begin{equation} J_n\cdot L_{n} < 2^{\min_{s}n\chi(p';B_{s})-\mu}\label{twjncln}\end{equation}
if $n$ is 
sufficiently large. \vspace{0.2cm}


By (\ref{twjncln}), we can apply 
(\ref{qnocsig4'}) to  $X_{j,l}$. 
We have: If
$n$ is  sufficiently large,  the event
\begin{align*}&\left(\bigcap_{s}  \left\{ \max_{j\in \{ 1,\dots ,J_n\}} \max_{l\in \{ 1,\dots ,L_n\}}
\mathrm{tr}\left(\overline{W}_s^{\otimes n}(X_{j,l})D_{X_{j,l}}\right) \geq
1-|\overline{\theta}|2^{-{n}^{1/16}\beta}\right\}\right)\allowdisplaybreaks\\
&\cap \biggl( \lVert \sum_{l=1}^{L_{n}} \frac{1}{L_{n}} Q_{t^n}(X_{j,l}) -\Theta_{t^n}
\rVert_1 \leq  2^{-\sqrt{n}
\frac{1}{16}\hat{\beta}(\alpha)} \text{ }\forall t^n ~ \forall j  \biggr)
\end{align*} has a positive
probability with respect to $p'$. 

 This means that for any $\epsilon >0$, if
$n$ is  sufficiently large we can find a realization $x_{j,l}$ of
$X_{j,l}$
 with a positive probability such that for all $s\in \overline{\theta}$, $t^n \in \theta^n$, $\pi\in\Pi_n$,
 and $j \in \{1,\dots,J_n\}$,   we have
\begin{equation}
 \sum_{l=1}^{L_{n}} \mathrm{tr}\left(\overline{W}_s^{\otimes n}(x_{j,l})D_{x_{j,l}}\right) \geq
1-\epsilon\text{ ,}\label{wappstfasi}\end{equation}  and
\begin{equation} \lVert \sum_{l=1}^{L_{n}} \frac{1}{L_{n}} Q_{t^n}(x_{j,l}) -\Theta_{t^n}
\rVert_1  \leq 2^{-\sqrt{n}
\frac{1}{16}\hat{\beta}(\alpha)} \text{ .}\label{absc1}\end{equation}\vspace{0.2cm}

We define
for $\pi\in\Pi_n$
 its permutation matrix on ${H}^{\otimes n}$ by $P_{\pi}$.
We have $V_{t^n}(\pi(x^n)) $ $=$ $P_{\pi}V_{\pi^{-1}(t^n)}(x^n)P_{\pi}^{\dagger}$.
For $\pi\in\Pi_n$, we define
 $\Theta_{t^n,\pi}:=  \sum_{x^n \in \mathsf{T}_{p,\delta}}
{p'}(x^n) Q_{{t^n}}(\pi(x^n))$.
We have 
$\Theta_{t^n,\pi}$ $=$ $  P_{\pi} \left(\sum_{x^n \in \mathsf{T}_{p,\delta}}
{p'}(x^n)  Q_{\pi^{-1}(t^n)}(x^n)\right)P_{\pi}^{\dagger}$ $=$ $P_{\pi} \Theta_{\pi(t^n)}P_{\pi}^{\dagger}$.


We choose a suitable positive  $\alpha$.
For any given $j' \in \{1, \dots, J_n\}$,
 we have
\begin{align}
& \left\|\sum_{l=1}^{L_n}\frac{1}{L_n}V_{t^n}(\pi(x_{j',l}))-\Theta_{t^n,\pi}\right\|_1\notag\\
 &\leq\|\sum_{l=1}^{L_n}\frac{1}{L_n}V_{t^n}(\pi(x_{j',l}))-\sum_{l=1}^{L_n}\frac{1}{L_n}Q_{t^n}(\pi(x_{j',l}))\|_1\notag\\
 &\qquad+\|\sum_{l=1}^{L_n}\frac{1}{L_n}Q_{t^n}(\pi(x_{j',l}))-\Theta_{t^n,\pi}\|_1\notag\\
 &\leq \sum_{l=1}^{L_n}2^{-\sqrt{n}
\frac{1}{16}\hat{\beta}(\alpha)}+\|P_\pi Q_{\pi^{-1}(t^n)}(x_{j',l})P_{\pi}^{\dagger}-P_{\pi} \Theta_{\pi(t^n)}P_{\pi}^{\dagger}\|_1\notag\\
 &=2^{-\sqrt{n}
\frac{1}{16}\hat{\beta}(\alpha)}+\|\sum_{l=1}^{L_n}\frac{1}{L_n} Q_{\pi^{-1}(t^n)}(x_{j',l})-\Theta_{\pi^{-1}(t^n)}\|_1\notag\\
 &\leq 2^{-\sqrt{n}
\frac{1}{16}\hat{\beta}(\alpha)}+\sqrt{2^{-\frac{1}{2}n\beta(\alpha)}+2^{-\frac{1}{2}n\beta(\alpha)''}}\notag\\
 &\leq 2^{-\sqrt{n}
\frac{1}{32}\hat{\beta}(\alpha)}\text{ ,}\label{eq_8}
 \end{align}
where the first inequality is an application of the triangle inequality and the second is again the triangle inequality combined
 with (\ref{eq_4}). The following equality follows because $\|U\cdot A\cdot U^\dag\|_1=\|A\|_1$ for all $A\in\mathcal B(H^{\otimes n})$ and unitary matrices $U\in\mathcal B(H^{\otimes n})$.
At last, we use  (\ref{absc1}).

By (\ref{eq_8}), we have
\begin{align*}
&\|\frac{1}{J_n\cdot L_n}\sum_{j=1}^{J_n}\sum_{l=1}^{L_n}V_{t^n}(\pi(x_{j,l}))-\Theta_{t^n,\pi}\|_1\\
&\leq  2^{-\sqrt{n}
\frac{1}{32}\hat{\beta}(\alpha)}
\text{ .}\end{align*}

By Lemma \ref{eq_9}  and  the inequality (\ref{eq_8}),  for a uniformly distributed
 random variable $R_{uni}$ with values  in
$\{1,\dots,J_n\}$ and all $\pi\in\Pi_{n}$ and $t^n\in\theta^n$, we have
\begin{align}& \chi(R_{uni};Z_{t^n,\pi}) \allowdisplaybreaks\notag\\
&=S\left( \sum_{j=1}^{J_n} \frac{1}{J_{n}} \sum_{l=1}^{L_{n}}
\frac{1}{L_{n}} V_{t^n}(\pi(x_{j,l}))\right) \allowdisplaybreaks\notag\\
&- \sum_{j=1}^{J_n}
\frac{1}{J_{n}}S\left(\sum_{l=1}^{L_{n}}
 \frac{1}{L_{n}}V_{t^n}(\pi(x_{j,l}))\right)\allowdisplaybreaks\notag\\
  &\leq \left\vert  S\left( \sum_{j=1}^{J_n} \frac{1}{J_{n}} \sum_{l=1}^{L_{n}}
\frac{1}{L_{n}} V_{t^n}(\pi(x_{j,l}))\right)-S\left( \Theta_{t^n,\pi} \right) \right\vert \allowdisplaybreaks\notag\\
&+\left\vert  S(\Theta_{t^n,\pi} )- \sum_{j=1}^{J_n} \frac{1}{J_{n}}S\left( \sum_{l=1}^{L_{n}}
\frac{1}{L_{n}} V_{t^n}(\pi(x_{j,l}))\right)\right\vert \allowdisplaybreaks\notag\\
&\leq 2\cdot  2^{-\sqrt{n}
\frac{1}{32}\hat{\beta}(\alpha)} \log (nd-1) + 
2h( 2^{-\sqrt{n}
\frac{1}{32}\hat{\beta}(\alpha)}) 
\text{ .}\label{wehaveb}\end{align}

By (\ref{wehaveb}), for any positive $\lambda$ if $n$ is sufficiently
large, we have
\begin{equation}\label{e10a}
\max_{t^n \in \theta^n} \chi(R_{uni};Z_{t^n,\pi}) \leq \lambda\text{
.}\end{equation}

For an arbitrary positive $\delta$, let
\[J_n:=2^{n\min_{s \in \overline{\theta}}\chi(p;B_s)-\max_{t^n\in \theta^n}\chi(p;Z_{t^n})-n\delta}\text{ .}\]
Now we define a code $(E, \{D_{j}: j =1,\dots,J_n\})$,
by $E(x^n\mid j)=\frac{1}{L_{n}}$ if $x^n\in \{x_{j,l} :l\in\{1,\dots,L_{n}\}$,
and $E(x^n\mid j)=0$ if $x\not\in \{x_{j,l} :l\in\{1,\dots,L_{n}\}$,
and $D_{j}:=\frac{1}{L_{n}}\sum_{l=1}^{L_{n}}D_{x_{j,l}}$.
For any positive $\lambda$ and $\epsilon$ if $n$ is sufficiently large,
by (\ref{wappstfasi}) and (\ref{e10a}),
it holds that
\[\max_{s \in \overline{\theta}}\frac{1}{J_n}  
\sum_{j=1}^{J_n}  \sum_{a^n \in \mathsf{A}^n}
E^n(a^n|j)
\mathrm{tr}\left(\overline{W}_s^{\otimes n}(a^n)D_{j}\right) \geq
1-\epsilon \text{ ,}\]
\[\max_{t^n\in\theta^n}\max_{\pi\in\Pi_n}
\chi\left(R_{uni};Z_{t^n,\pi}\right) \leq \epsilon\text{ .}\]

We obtain

	\begin{equation}
	\hat{C}_s(\{(\overline{W}_s,{V}_t) :s \in \overline{\theta}, t \in \theta\})
	\geq \min_{s \in \overline{\theta}}\chi(p;B_s)-\lim_{n\rightarrow \infty} \frac{1}{n}\max_{t^n\in \theta^n}\chi(p;Z_{t^n})
	\text{ .}\label{hatceque}\end{equation}

The achievability of
$\lim_{n\rightarrow \infty} \frac{1}{n}\Bigl(\min_{s \in \overline{\theta}}\chi(p_U;B_s^{\otimes n})$ $-$ $ \max_{t^n\in \theta^n}\chi(p_U;Z_{t^n})\Bigr)$
is then shown via standard arguments  (cf. \cite{De}).\vspace{0.3cm}

Now we are going to prove the converse.\vspace{0.15cm}

Let  $(\mathcal{C}_n)=(E^{(n)},\{D_{j}^{(n)}:j\})$ be a sequence of $(n,J_n)$ code such
that
\[
\max_{s \in \overline{\theta}} P_e(\mathcal{C}_n,s,n)
\leq\lambda_n\text{ ,}
\]
\[
\max_{t^n\in\theta^n}\max_{\pi\in\Pi_n}
\chi\left(R_{uni};Z_{t^n,\pi}\right) \leq \epsilon_{n}\text{ ,}
\] where $ \lim_{n\to\infty}\lambda_n=0$ and
$\lim_{n\to\infty}\epsilon_{n}=0$, where $R_{uni}$ is the uniform
distribution on $\{1,\cdots J_n\}$.

It is known (cf. (\cite{Win})) that the  capacity of a 
classical-quantum channel $\overline{W}$
cannot exceed $I(R_{uni};B)$. Since  the  capacity of a 
compound
classical-quantum channel $(\overline{W}_s)_{s \in \overline{\theta}}$ 
cannot exceed the worst channel in $\{\overline{W}_s: s \in \overline{\theta}\}$,
its capacity is
bounded by $\frac{1}{n}(\min_{s \in \overline{\theta}}\chi(p_U;B_s)$.
The 
 enhanced achievable secrecy
rate for the compound-arbitrarily varying wiretap classical-quantum channel cannot exceed the  capacity 
without a wiretapper; thus for any $\xi > 0$ let us choose
$\epsilon_{n}=\frac{1}{2}\xi$, if $n$ is sufficiently large, the  secrecy rate
of  $(\mathcal{C}_n)$ cannot be greater than

\begin{align}&\min_{s \in \overline{\theta}}I(R_{uni};B_s)
-\xi\allowdisplaybreaks\notag\\
&\leq\min_{s \in \overline{\theta}}I(R_{uni};B_s)-\frac{1}{n}\max_{t^n\in \theta^n}\chi(R_{uni};Z_{t^n}) -\xi
+\frac{1}{n}\epsilon_{n}\allowdisplaybreaks\notag\\
&\leq  \min_{s \in \overline{\theta}}I(R_{uni};B_s)-\frac{1}{n}\max_{t^n\in \theta^n}\chi(R_{uni};Z_{t^n})-\frac{1}{2}\xi\allowdisplaybreaks\notag\\
&=\min_{s \in \overline{\theta}} H(R_{uni}) + H(B_s)- H(R_{uni}B_s)- \frac{1}{n}\max_{t^n\in \theta^n}\chi(R_{uni};Z_{t^n})-\frac{1}{2}\xi\allowdisplaybreaks\notag\\
&\leq  \min_{s \in \overline{\theta}}\chi(R_{uni};B_s)- \frac{1}{n}\max_{t^n\in \theta^n}\chi(R_{uni};Z_{t^n})-\frac{1}{2}\xi\allowdisplaybreaks\notag\\
&\leq\frac{1}{n}\max_{\Lambda_n}( \min_{s \in \overline{\theta}}\chi(p_U;B_s^n)
-\max_{t^n\in \theta^n}\chi(p_U;Z_{t^n})-\frac{1}{2}\xi\text{  .}
\end{align}The  third inequality holds because $R_{uni}\rightarrow A \rightarrow \{B_s^{\otimes n},Z_{t^n}:s,t_n\}$ is
always a Markov chain.

This and (\ref{hatceque}) prove Theorem \ref{loatbfis}. \qed\end{proof}
\vspace{0.15cm}

\begin{corollary}
Let  $\theta$ $:=$ $\{1,\cdots,T\}$ be a finite index set.
Let $\overline{\theta}$ $:=$ $\{1,2\cdots\}$ be an infinite index set.
	Let $\{(\overline{W}_s,{V}_t) :s \in \overline{\theta}, t \in \theta\}$ 
	be a compound-arbitrarily varying wiretap classical-quantum channel.
	We have
\[
	\hat{C}_s(\{(\overline{W}_s,{V}_t) :s \in \overline{\theta}, t \in \theta\})
	=\lim_{n\rightarrow \infty} \frac{1}{n} \max_{\Lambda_n}
	\Bigl(\inf_{s \in \overline{\theta}}\chi(p_U;B_s^{\otimes n})-\max_{t^n\in \theta^n}\chi(p_U;Z_{t^n})\Bigr)
	\text{ .}\]\label{infcompunt}
\end{corollary}

\begin{proof}

For a linear map $W: \mathcal{S}(H') \rightarrow  \mathcal{S}(H'')$
let
\begin{equation}\|W \|_{\lozenge}:=\sup_{n\in \mathbb{N}}\max_{a\in S(\mathbb{C}^n
\otimes H'), \|a\|_1=1}\| (\mathrm{I}_n \otimes W)(a)\|_1\text{ .}\end{equation}

It is  known \cite{Pa} that this norm is multiplicative, i.e.,
$\|W  \otimes W' \|_{\lozenge} = \|W\|_{\lozenge}\cdot \|W'
\|_{\lozenge}$.

A $\tau$-net in the space of the completely positive trace-preserving maps  $\mathcal{S}(H') \rightarrow  \mathcal{S}(H'')$ is
a finite  set $\left({W^{(k)}}\right)_{k=1}^{K}$ of completely positive trace-preserving maps  $\mathcal{S}(H') \rightarrow  \mathcal{S}(H'')$
 with the property that for each completely positive trace-preserving map $W: \mathcal{S}(H') \rightarrow  \mathcal{S}(H'')$,
there is at least one $k \in \{1, \dots , K\}$
 with $ \| W-W^{(k)} \|_{\lozenge} < \tau$.
\begin{lemma}[$\tau-$net \cite{Mi/Sch}]\label{e2} Let $H'$ and $H''$ be  finite-dimensional
complex Hilbert spaces.
For any $\tau \in (0,1]$, there is a $\tau$-net of quantum channels
$\left(W^{(k)}\right)_{k=1}^{K}$ in  the space of the completely
positive trace-preserving maps $\mathcal{S}(H') \rightarrow
\mathcal{S}(H'')$  with $K \leq (\frac{3}{\tau})^{2{d'}^4}$, where
$d' = \dim H'$.
\end{lemma}\vspace{0.15cm}


We now consider a $\overline{\theta}$ such that $|\overline{\theta}|$ is not finite.
For $n\in\mathbb{N}$  we define $\tau_n :=n^2$. $\{\tau_n: n\in\mathbb{N}\}$ 
 is a series of positive constants such that
$ (\frac{3}{\tau_n})^{2{d'}^4} < 2^{\frac{1}{2}n^{1/16}\beta}$ and $\lim_{n\rightarrow\infty} n\tau_n =0$. 
By Lemma \ref{e2},  there exists a finite set
$\overline{\theta_{\tau_n}}'$ with $|\overline{\theta_{\tau_n}}'|\leq (\frac{3}{\tau_n})^{2{d'}^4}$ and
$\tau_n$-nets $\left(\overline{W}_{s'}\right)_{s' \in \overline{\theta_{\tau_n}}'}$,
$\left(V_{s'}\right)_{s' \in \overline{\theta_{\tau_n}}'}$ such that for every $
t\in\overline{\theta} $ we can find a $s' \in \overline{\theta_{\tau_n}}'$  with $\left\| \overline{W}_s -
\overline{W}_{s'}\right\|_{\lozenge}\leq \tau_n$.

We assume that the sender's encoding is restricted to transmitting an  indexed finite set of  quantum  states
$\{\rho_{x}: x\in \mathsf{A}\}\subset \mathcal{S}({H'}^{\otimes n})$.

By Theorem \ref{loatbfis}, the
legitimate transmitters are able to build a code $C_2
=\{E,\{D_{j}:j\}\}$
 such that for all $s''\in \overline{\theta_{\tau_n}}'$,
$t\in\theta$, and $\pi\in\Pi_n$, it holds that
\begin{equation}
 \frac{1}{J_n} \sum_{j=1}^{J_n} \sum_{x^n\in \mathsf{A}^n} E(x^n\mid j) \mathrm{tr}\left(
\overline{W}_{s''}^{\otimes n}\left( \rho_{x^n} \right)D_{j}^n\right)\geq 1 -  (\frac{3}{\tau_n})^{2{d'}^4} 2^{-n^{1/16}\beta}
 \geq 1 - 2^{-\frac{1}{2}n^{1/16}\beta} \text{ ,}\label{con1}\end{equation}
\begin{equation} \chi(R_{uni};Z_{t,\pi}^{ n})\leq
2^{-n\upsilon} \label{con2}\text { .}\end{equation}

Let $|\psi_{x^n}\rangle
 \langle \psi_{x^n} | \in \mathcal{S}({H'}^{\otimes n}\otimes{H'}^{\otimes n})$
 be an arbitrary purification of the  quantum state
$\rho_{x^n}$, then $\mathrm{tr}\left[ \left(\overline{W}_s^{\otimes n} -
\overline{W}_{s'}^{\otimes n}\right)(\rho_{x^n})\right] =\mathrm{tr}
\left(\mathrm{tr}_{{H'}^{\otimes n}} \left[\mathrm{I}_{H'}^{\otimes n}
\otimes (\overline{W}_s^{\otimes n}-\overline{W}_{s'}^{\otimes n}) \left(
|\psi_{x^n}\rangle \langle \psi_{x^n} |\right)\right]\right)$. We
have
\begin{align*}& \mathrm{tr}\left|\sum_{x^n\in \mathsf{A}^n} E(x^n\mid j) \left(\overline{W}_s^{\otimes n} - \overline{W}_{s'}^{\otimes n}\right)(\rho_{x^n})\right|\\
&= \mathrm{tr} \left(\sum_{x^n\in \mathsf{A}^n} E(x^n\mid j)\mathrm{tr}_{{H'}^{\otimes n}} \left|\mathrm{I}_{H'}^{\otimes
N} \otimes (\overline{W}_s^{\otimes n}-\overline{W}_{s'}^{\otimes n})
\left( |\psi_{x^n}\rangle \langle \psi_{x^n} |\right)\right|\right)\\
& =  \mathrm{tr} \left|\sum_{x^n\in \mathsf{A}^n} E(x^n\mid j)\mathrm{I}_{H'}^{\otimes n} \otimes (\overline{W}_s^{\otimes
n}-\overline{W}_{s'}^{\otimes n})
\left( |\psi_{x^n}\rangle \langle \psi_{x^n} |\right)\right|\\
&=\sum_{x^n\in \mathsf{A}^n} E(x^n\mid j)  \left\|\mathrm{I}_{H'}^{\otimes n} \otimes (\overline{W}_s^{\otimes n}-\overline{W}_{s'}^{\otimes
N})\left( |\psi_{x^n}\rangle \langle \psi_{x^n} |\right)\right\|_1\\
& \leq\sum_{x^n\in \mathsf{A}^n} E(x^n\mid j)   \| \overline{W}_s^{\otimes n}-\overline{W}_{s'}^{\otimes n}\|_{\lozenge}
\cdot\left\|\left( |\psi_{x^n}\rangle \langle \psi_{x^n} |\right)\right\|_1\\
&\leq N\tau_n\text{ .}
\end{align*}
 The
second equality follows from the definition of trace. The
third inequality follows by the definition of
${\|\cdot\|_{\lozenge}}$. The second inequality follows from the facts
that $\|\left( |\psi_{x^n}\rangle \langle \psi_{x^n}
|\right)\|_1=1$ and $\left\| \overline{W}_s ^{\otimes n}-\overline{W}_{s'}^{\otimes
N}\right\|_{\lozenge}= \left\| \left(\overline{W}_s-\overline{W}_{s'}\right)^{\otimes
N}\right\|_{\lozenge} = N\cdot\left\| \overline{W}_s -
\overline{W}_{s'}\right\|_{\lozenge}$, since $\|\cdot\|_{\lozenge}$ is
multiplicative.\vspace{0.15cm}

It follows that \begin{align}  &\biggl|\frac{1}{J_n} \sum_{j=1}^{J_n}\sum_{x^n\in \mathsf{A}^n} E(x^n\mid j)
\mathrm{tr}\left( \overline{W}_s^{\otimes n}\left( \rho_{x^n}
\right)D_{j}^n\right)\allowdisplaybreaks\notag\\
&-\frac{1}{J_n} \sum_{j=1}^{J_n}\sum_{x^n\in \mathsf{A}^n} E(x^n\mid j)
\mathrm{tr}\left(
\overline{W}_{s'}^{\otimes n}\left( \rho_{x^n} \right)D_{j}^n\right) \biggr| \allowdisplaybreaks\notag\\
&\leq\frac{1}{J_n} \sum_{j=1}^{J_n} \sum_{x^n\in \mathsf{A}^n} E(x^n\mid j)\left|\mathrm{tr}\left[ \left(
\overline{W}_s^{\otimes n}
-\overline{W}_{s'}^{\otimes n}\right)\left( \rho_{x^n} \right)D_{j}^n \right]\right| \allowdisplaybreaks\notag\\
&\leq\frac{1}{J_n} \sum_{j=1}^{J_n}\sum_{x^n\in \mathsf{A}^n} E(x^n\mid j) \mathrm{tr}\left[ \left(
\overline{W}_s^{\otimes n}
-\overline{W}_{s'}^{\otimes n}\right)\left( \rho_{x^n} \right)D_{j}^n\right]  \allowdisplaybreaks\notag\\
&\leq \frac{1}{J_n} \sum_{j=1}^{J_n} \sum_{x^n\in \mathsf{A}^n} E(x^n\mid j)\mathrm{tr}\left[ \left(
\overline{W}_s^{\otimes n}
-\overline{W}_{s'}^{\otimes n}\right)\left( \rho_{x^n} \right)\right]\allowdisplaybreaks\notag\\
&\leq \frac{1}{J_n}J_n  n\tau_n \allowdisplaybreaks\notag\\
&=  n\tau_n \text{ .} \label{iea}\end{align} 

By (\ref{iea}),  we have
\[\sup_{s\in\overline{\theta}}
 \frac{1}{J_n} \sum_{j=1}^{J_n} \sum_{x^n\in \mathsf{A}^n} E(x^n\mid j)\mathrm{tr}\left(
\overline{W}_s^{\otimes n}\left( \rho_{x^n} \right)D_{j}^n\right)\geq 1
-\lambda_{\tau_n}-n\tau_n\text{ .}\]
Thus,
	\begin{equation}
	\hat{C}_s(\{(\overline{W}_s,{V}_t) :s \in \overline{\theta}, t \in \theta\})
	\geq\lim_{n\rightarrow \infty} \frac{1}{n}(\inf_{s \in \overline{\theta}}\chi(p;B_s^{\otimes n})-\max_{t^n\in \theta^n}\chi(p;Z_{t^n}))
	\text{ .}\label{hatceque2}\end{equation}

The achievability of
$\lim_{n\rightarrow \infty} \frac{1}{n}\Bigl(\min_{s \in \overline{\theta}}\chi(p_U;B_s)$ $-$ $ \max_{t^n\in \theta^n}\chi(p_U;Z_{t^n})\Bigr)$
is then shown via standard arguments.

The proof of the converse is similar to those given in the proof of Theorem \ref{loatbfis}. \qed
\end{proof}

\begin{corollary}
Let  $\overline{\theta}$  and
  $\theta$  be finite index sets.
	Let $\{(\overline{W}_s,{V}_t) :s \in \overline{\theta}, t \in \theta\}$ 
	be a compound-arbitrarily varying wiretap classical-quantum channel.
The secrecy capacity of $\{(\overline{W}_s,{V}_t) :s \in \overline{\theta}, t \in \theta\}$ is
equal to 
	\[ \lim_{n\rightarrow \infty} \frac{1}{n}\max_{\Lambda_n}
	 \Bigl(\min_{s \in \overline{\theta}}\chi(p_U;B_s^{\otimes n})-\max_{t^n\in \theta^n}\chi(p_U;Z_{t^n})\Bigr)
	\text{ .}\]
\end{corollary}

\begin{proof} The corollary follows  immediately from the fact that
the enhanced secrecy capacity of a compound-arbitrarily varying wiretap classical-quantum channel
is less or equal to its secrecy capacity. \qed
\end{proof}

\section{Secrecy
Capacity of Arbitrarily
Varying Classical-Quantum Wiretap Channel}\label{SCoavcqwc}

In this section, we use the
  results of Section \ref{CAVWCQC} to prove our main result: the formula for
the secrecy capacities under common
randomness assisted coding of arbitrarily varying classical-quantum
wiretap  channels. 

\begin{theorem}
Let   $\theta$ $:=$ $\{1,\cdots,T\}$ be a finite index set.
	Let $(W_t,{V}_t)_{ t \in \theta}$ 
	be an arbitrarily varying classical-quantum wiretap channel.
	We have
\begin{align}&
	C_s(\{(W_t,{V}_t): t \in \theta\};cr)\notag\\
	&= \lim_{n\rightarrow \infty} \frac{1}{n}\max_{\Lambda_n}\Bigl(\inf_{B_q \in Conv((B_t)_{t\in \theta})}
	\chi(p_U;B_q^{\otimes n})-\max_{t^n\in \theta^n}\chi(p_U;Z_{t^n})\Bigr)
	\text{ .}\end{align}
	Here $Conv((B_t)_{t\in \theta})$ is  the  convex hull of $\{B_t :t\in \theta\}$.\label{commperm}
\end{theorem}

\begin{proof}

\it i) Achievement \rm \vspace{0.2cm}

Our idea is similar to the results for  classical arbitrarily varying wiretap  channel
in \cite{Wi/No/Bo}: Applying Ahlswede's robustification technique (cf. \cite{Bj/Bo/Ja/No}),
we use the results of Section \ref{CAVWCQC} to show the existence of a 
common randomness assisted quantum code. 
Additionally, we have to consider the security.\vspace{0.2cm}

We denote the set of distribution function  on $\theta$ by $ \mathsf{P}(\theta)$.
For every $q\in \mathsf{P}(\theta)$, we define a classical-quantum channel
 $\overline{W}_q:= \sum_{s\in \theta} q(s)W_s$.
We now define a
compound-arbitrarily varying wiretap classical-quantum channel
by
\[\{(\overline{W}_q,{V}_t); q\in \mathsf{P}(\theta), t \in \theta\}\text{ .}\]

We fix a probability distribution
$p\in \mathsf{A}$.
 We choose arbitrarily  $\epsilon>0$, $\delta>0$, and
$\zeta>0$.
Let 
\[J_n =  \lfloor  2^{n\inf_{B_q \in Conv((B_s)_{s\in \theta})}\chi(p;B_q)-\max_{t^n\in \theta^n}\chi(p;Z_{t^n})-n\delta} \rfloor  \text{ .}\]
By Corollary \ref{infcompunt},  if   $n$ is sufficiently large, 
there exists an  $(n, J_n)$
code $C = \bigl(E^n, \{D_j^n : j = 1,\cdots J_n\}\bigr)$  such that 
\[ \max_{q\in \mathsf{P}(\theta)}  1- \frac{1}{J_n} \sum_{j=1}^{J_n}  
\mathrm{tr}(\overline{W}_{q}(E^n(~|j))D_j^n)
< \epsilon\text{ ,}\]
\[\max_{t^n\in\theta^n}\max_{\pi\in\Pi_n}
\chi\left(R_{uni};Z_{t^n,\pi}\right) < \zeta\text{ .}\]

Similar to the proofs in \cite{Bj/Bo/Ja/No},
 we now apply Ahlswede's robustification technique.

\begin{lemma}[cf. \cite{Ahl2}, \cite{Ahl3}, and \cite{Ahl/Bj/Bo/No}]
Let $S$ be a finite set and $n\in \mathbb{N}$. If a function $f$ $:$ 
$S^n \rightarrow [0, 1]$ satisfies
\[\sum_{s^n\in S^n} f(s^n)q(s_1)q(s_2)\cdots q(s_n)\geq 1-\epsilon \text{ ,}\]
for all $q\in \mathsf{P}(\theta)$ and a positive $\epsilon\in [0,1]$,
then
\begin{equation}\frac{1}{n!}\sum_{\pi\in\Pi_n} f(\pi(s^n))\geq 1-3(n+1)^{\vert S\vert}\epsilon
\text{ .}\end{equation}\label{roboustpermu}
\end{lemma}

We define a function $f$ $:$ $\theta^n\rightarrow [0,1]$ by
\[f(t^n) := \frac{1}{J_n} \sum_{j=1}^{J_n} 
\mathrm{tr}(W_{t^n}(E^n(~|j))D_j^n)\text{ .}\]
For every $q\in \mathsf{P}(\theta)$ we have
\begin{align*}&\sum_{t^n\in \theta^n} f(t^n)q(t_1)\cdots q(t_n)\\
&=\sum_{t^n\in \theta^n}\frac{1}{J_n} \sum_{j=1}^{J_n}  
\mathrm{tr}(W_{t^n}(E^n(~|j))D_j^n)q(t_1)\cdots q(t_n)\\
&=\frac{1}{J_n} \sum_{j=1}^{J_n}  
\mathrm{tr}\left(\sum_{t^n\in \theta^n}q(t_1)\cdots q(t_n)W_{t^n}(E^n(~|j))D_j^n\right)\\
&=\frac{1}{J_n} \sum_{j=1}^{J_n}  
\mathrm{tr}(\overline{W}_{q}(E^n(~|j))D_j^n)\\
&> 1- 2^{-n\beta/2}
\text{ .}
\end{align*}

Applying Lemma \ref{roboustpermu}, we have
\begin{align}&1-3(n+1)^{\vert \theta\vert}2^{-n\beta/2}\notag\\
&\leq \frac{1}{n!}\sum_{\pi\in\Pi_n} f(\pi(t^n))\allowdisplaybreaks\notag\\
&=\frac{1}{n!}\sum_{\pi\in\Pi_n}\frac{1}{J_n} \sum_{j=1}^{J_n}  
\mathrm{tr}(W_{\pi(t^n)}(E^n(~|j))D_j^n)\allowdisplaybreaks\notag\\
&=\frac{1}{n!}\sum_{\pi\in\Pi_n}\frac{1}{J_n} \sum_{j=1}^{J_n}  \sum_{a^n \in \mathsf{A}^n}
E^n(a^n|j)\mathrm{tr}(W_{\pi(t^n)}(a^n)D_j^n)\allowdisplaybreaks\notag\\
&=\frac{1}{n!}\sum_{\pi\in\Pi_n}\frac{1}{J_n} \sum_{j=1}^{J_n}  \sum_{a^n \in \mathsf{A}^n}
E^n(a^n|j)\mathrm{tr}(W_{t^n}(\pi^{-1}(a^n))P_{\pi}^{\dagger}D_j^n P_{\pi})
\text{ ,}\label{inqupermu}
\end{align}
where for $\pi\in\Pi_n$, $P_{\pi}$
is its permutation matrix on ${H}^{\otimes n}$.

We now define our
common randomness assisted quantum code
by
\[\left\{\left(\pi \circ E^n, \{P_{\pi} D_j^n P_{\pi}^{\dagger},j\in\{1,\cdots,J_n\}\}\right): \pi\in\Pi_n\right\}\text{ .}\]


$P_{\pi} D_j^n P_{\pi}^{\dagger}$ is Hermitian and
positive semidefinite. Furthermore, it holds that
$\sum_{j=1}^{J_n} P_{\pi}D_j^n P_{\pi}^{\dagger} $ $=$ $\sum_{j=1}^{J_n} P_{\pi} id_{{H}^{\otimes n}} P_{\pi}^{\dagger}$
$=$ $id_{{H}^{\otimes n}}$.\vspace{0.2cm}

By (\ref{inqupermu}), and by the fact that

\begin{align*}&\frac{1}{n!}\sum_{\pi\in\Pi_n}\max_{t^n\in\theta^n}
\chi\left(R_{uni};Z_{t^n,\pi}\right)\\
&\leq \max_{t^n\in\theta^n}\max_{\pi\in\Pi_n}
\chi\left(R_{uni};Z_{t^n,\pi}\right)\\
&< \zeta
\text{ ,}\end{align*}
for any positive $\varepsilon$ when $n$ is sufficiently large, it holds that:
\begin{equation}
	C_s(\{(W_t,{V}_t): t \in \theta\};cr)
	\geq \inf_{B_q \in Conv((B_s)_{s\in \theta})}\chi(p;B_{q})-\lim_{n\rightarrow \infty} \frac{1}{n}\max_{t^n\in \theta^n}\chi(p;Z_{t^n})-\varepsilon
	\text{ .}\label{hatceque3}\end{equation}

The achievability of
$\lim_{n\rightarrow \infty} \frac{1}{n}\Bigl(\min_{B_q \in Conv((B_s)_{s\in \theta})}\chi(p_U;B_q^{\otimes n})$ $-$ $ \max_{t^n\in \theta^n}\chi(p_U;Z_{t^n})\Bigr)$
is then shown via standard arguments  (cf. \cite{De}).\vspace{0.2cm}

\it ii) Converse \rm \vspace{0.2cm}

Now we are going to prove the converse.
Similar to the results for  classical arbitrarily varying wiretap  channel
in \cite{Wi/No/Bo}, we limit the amount of common randomness.
\vspace{0.2cm}

Let   $(\{\mathcal{C}^{\gamma}_n:\gamma\in \Gamma\})$ be a sequence of $(n, J_n)$  common randomness assisted
 codes such
that
\begin{equation}\max_{s \in \theta}\frac{1}{\left|\Gamma\right|} \sum_{\gamma=1}^{\left|\Gamma\right|} P_e(\mathcal{C}_n^{\gamma},t^n)
\leq\lambda_n\text{ ,}\label{tnaqntaclp}\end{equation}
\begin{equation}\max_{t^n\in\theta^n}\frac{1}{\left|\Gamma\right|} \sum_{\gamma=1}^{\left|\Gamma\right|}
\chi\left(R_{uni};Z_{\mathcal{C}^{\gamma},t^n}\right) \leq \epsilon_{n}\text{ ,}\label{tnaqntaclp2}\end{equation} where $ \lim_{n\to\infty}\lambda_n=0$ and
$\lim_{n\to\infty}\epsilon_{n}=0$. 

We consider a $\left|\Gamma\right|$-long  sequence of outputs $(1,\cdots,\left|\Gamma\right|)$ has been given
by the common randomness and a $n\left|\Gamma\right|$-long block has been sent. 
The  legitimate receiver obtains the quantum states $\{B_{q}^{\gamma}:\gamma\in \Gamma\}$.
By (\ref{tnaqntaclp}), 
he is able to decode $2^{n\left|\Gamma\right|\log J_n}$ messages. By \cite{Bj/Bo/Ja/No},
for every $B_q \in Conv((B_s)_{s\in \theta})$ we have
\[\log J_n\leq\frac{1}{\left|\Gamma\right|}\frac{1}{n} \sum_{\gamma=1}^{\left|\Gamma\right|}\chi(R_{uni};B_{q}^{\gamma\otimes n})\text{  ,}\]

%

and by (\ref{tnaqntaclp2}),
for  and every $t^n\in \theta^n$, we have
\[\frac{1}{n}\log J_n\leq\frac{1}{\left|\Gamma\right|} \frac{1}{n}\sum_{\gamma=1}^{\left|\Gamma\right|}
 (\chi(R_{uni};B_{q}^{\gamma\otimes n})-\chi(R_{uni};Z_{t^n}^{\gamma})) +\epsilon_{n}\text{  .}\]

\begin{lemma} Let $c > 0$.
 For every $q \in \mathsf{P}(\theta)$
 and $s^n\in \theta^n$, let a function $I_{q,s^n} : \Gamma \rightarrow
[0,c]$ be given. We assume that
these functions satisfy the following:
for every $\gamma\in\Gamma$
 and  $s^n\in \theta^n$
\[|I_{q,s^n} (\gamma) - I_{q',s^n}(\gamma)| \leq f(\delta)\text{  ,}\]
if $q,q' \in \mathsf{P}(\theta)$ satisfy $\|q - q'\|_1 \leq \delta$ for some $f(\delta)$
 which tends to $0$ as $\delta$ tends to $0$. We write $\mu(I_{q,s^n} ):=\sum_{\gamma\in\Gamma}\mu(\gamma)I_{q,s^n} (\gamma)$,
where $\mu(\gamma)$ is the probability of $\gamma$.
 Then, for every $\varepsilon > 0$ and sufficiently large $n$, there are $L = n^2$
realizations $\gamma_1,\cdots, \gamma_L$ such that
\[\frac{1}{L}\sum_{l=1}^{L}I_{q,s^n} (\gamma_l)\geq(1-\varepsilon)\mu(I_{q,s^n} )-\varepsilon\]
for every $q \in \mathsf{P}(\theta)$
 and $s^n\in \theta^n$.
\label{lcfeqicsit}\end{lemma}

\begin{proof}
Let $0 < \delta < \frac{1}{2}$ and
$K$ be a positive integer. 
We denote the set
of possible types of sequences of length $K$  by
$P_0^{K}(\theta)$.
As in the approximation argument in \cite{Bl/Br/Th2}, 
one can show that every $q \in \mathsf{P}(\theta)$ is at most a
distance $\delta$ away from some $q' \in P_0^{K}(\theta)$
if $K \geq 2\frac{|\theta| - 1}{\delta}$.

Let $K:= \lceil 2\frac{|\theta|-1}{\delta}\rceil$. Then, 
$| P_0^{K}(\theta)| \leq  \left(2\frac{|\theta|}{\delta}\right)^{|\theta|}$. This approximating set is used to handle
the infinite set $\mathsf{P}(\theta)$.

Now let $G_1, \cdots , G_L$ be i.i.d. random variables with values in $\Gamma$ and distributed according to $\mu$. 
Set $\mu_{*}:= \min_{q\in \mathsf{P}(\theta)} \min_{ s^n\in \theta^n} \mu(I_{q,s^n})$.
 Using the union bound and the Chernoff bound (cf. \cite{Du/Pa}), we obtain
\begin{align*}&Pr\left\{ \frac{1}{L}\sum_{l=1}^{L}I_{q,s^n} (G_l) < \mu(I_{q,s^n}) ~
\forall q \in P_0^{K}(\theta) ~ \forall  s^n\in \theta^n\right\}\\
&\leq \exp\left(|\theta| \log \left(\frac{2|\theta|}{\delta}\right)+
n\log|\theta|-\frac{L\epsilon^2\mu_{*}}{3c}\right)\text{ .}\end{align*}
This, probability is smaller than $1$ if $L$ 
tends to infinity faster than $n$, e.g., if $L=n^2$.

Thus we have proved the existence of $\gamma_1, \cdots , \gamma_L$ which satisfies
\[ \frac{1}{L}\sum_{l=1}^{L}I_{q,s^n} (\gamma_l) \geq (1-\epsilon) \mu(I_{q,s^n})\]
for every $q\in  P_0^{K}(\theta)$ and  $s^n\in \theta^n$. 
Now let $q\in \mathsf{P}(\theta)$ be arbitrary and let
$ q' \in P_0^{K}(\theta)$
satisfy $\|q - q'\|_1 \leq \delta$. Then

\begin{align*}&\frac{1}{L}\sum_{l=1}^{L}I_{q,s^n} (\gamma_l) \\
&\geq \frac{1}{L}\sum_{l=1}^{L}I_{{q'},s^n} (\gamma_l) 
- f(\delta)\\
&\geq (1-\epsilon) \mu(I_{{q'},s^n})- f(\delta)\\
&\geq (1-\epsilon) \mu(I_{q,s^n}) -(2-\epsilon) f(\delta)
\text{  .}\end{align*}
Choosing $\delta$ sufficiently small proves the claim of the lemma. \qed
\end{proof}

For $q \in Conv(\{s: s\in \theta\})$,
we define
\[I_{q,s^n}(\gamma) := \frac{1}{n}\left(\chi(R_{uni};B_{q}^{\gamma\otimes n})-\chi(R_{uni};Z_{t^n}^{\gamma})\right)\text{  .}\]

In \cite{Bo/No2}, the continuity of $q\rightarrow \frac{1}{n}\chi(R_{uni};B_{q}^{\gamma\otimes n})$
 has been shown; thus, there is a 
$f(\delta)$ such that $|I_{q,s^n} (\gamma) - I_{q',s^n}(\gamma)|$
$\frac{1}{n}\frac{1}{\left|\Gamma\right|} \sum_{\gamma=1}^{\left|\Gamma\right|}(\chi(R_{uni};B_{q}^{\gamma\otimes n})$ $-$
$\frac{1}{n}\frac{1}{\left|\Gamma\right|} \sum_{\gamma=1}^{\left|\Gamma\right|}(\chi(R_{uni};B_{q'}^{\gamma\otimes n})$
$\leq f(\delta)$ for a  $f(\delta)$ that fulfills $f(\delta) \rightarrow 0$ when $\|q - q'\|_1 =\delta\rightarrow 0$. By Lemma \ref{lcfeqicsit},
there is a set $\Gamma'\subset\Gamma$ such that $\left|\Gamma'\right| = n^2$ and
\begin{align*}&
\frac{1}{\left|\Gamma'\right|}\frac{1}{n}\sum_{\gamma'\in \Gamma'}
\left(\chi(R_{uni};B_{q}^{\gamma'\otimes n})-\chi(R_{uni};Z_{t^n}^{\gamma'})\right)\\
&\geq (1-\varepsilon)\frac{1}{n}
\frac{1}{\left|\Gamma\right|}\sum_{\gamma\in \Gamma}
\left(\chi(R_{uni};B_{q}^{\gamma\otimes n})-\chi(R_{uni};Z_{t^n}^{\gamma})\right)  \text{  ,}\end{align*}
where $B_{q}^{\gamma'}$ and $Z_{t^n}^{\gamma}$ are the quantum states at
 the output of legitimate receiver channel and the wiretapper's channel, respectively,
when the output of the common randomness
is $\gamma'$.

Thus,
\begin{equation} \frac{1}{n}\log J_n\leq \frac{1}{1-\varepsilon}
\frac{1}{n}\frac{1}{\left|\Gamma'\right|}\sum_{\gamma\in \Gamma'}
\left(\chi(R_{uni};B_{q}^{\gamma\otimes n})-\chi(R_{uni};Z_{t^n}^{\gamma})+\epsilon_{n}\right)\text{  .}\end{equation}\vspace{0.2cm}

To prove the converse, we now consider 
\begin{align*}\allowdisplaybreaks[2]& \frac{1}{\left|\Gamma'\right|}\sum_{\gamma\in \Gamma'}
\frac{1}{n}\left(\chi(R_{uni};B_{q}^{\gamma\otimes n})-\chi(R_{uni};Z_{t^n}^{\gamma})\right)
-\frac{1}{n}\left(\chi(R_{uni};B_{q}^{\otimes n})-\chi(R_{uni};Z_{t^n})\right)\\
&=\frac{1}{\left|\Gamma'\right|}\sum_{\gamma\in \Gamma'}\frac{1}{n}
\left(\chi(R_{uni};B_{q}^{\gamma\otimes n})-\chi(R_{uni};Z_{t^n}^{\gamma})\right)\\
&-\frac{1}{n}\left(\chi(R_{uni};\frac{1}{\left|\Gamma'\right|}\sum_{\gamma\in \Gamma'}B_{q}^{\gamma\otimes n})
-\chi(R_{uni};\frac{1}{\left|\Gamma'\right|}\sum_{\gamma\in \Gamma'}Z_{t^n}^{\gamma})\right)\\
&=  \frac{1}{n}\frac{1}{\left|\Gamma'\right|}\sum_{\gamma\in \Gamma'}\left(\chi(R_{uni};B_{q}^{\gamma\otimes n})
-\chi(R_{uni};\frac{1}{\left|\Gamma'\right|}\sum_{\gamma\in \Gamma'}B_{q}^{\gamma\otimes n})\right)\\
&-\frac{1}{n}\frac{1}{\left|\Gamma'\right|}\sum_{\gamma\in \Gamma'}\left(\chi(R_{uni};Z_{t^n}^{\gamma})
+ \frac{1}{n}\chi(R_{uni};\frac{1}{\left|\Gamma'\right|}\sum_{\gamma\in \Gamma'}Z_{t^n}^{\gamma})\right)
\text{  .}\end{align*}\vspace{0.2cm}

Let $G_{uni}$ be the uniformly distributed random variable with value in
$\Gamma'$.
We have
\begin{align}\allowdisplaybreaks[2]&
 \frac{1}{n}\frac{1}{\left|\Gamma'\right|}\sum_{\gamma\in\Gamma'}\chi(R_{uni};B_{q}^{\gamma\otimes n})\allowdisplaybreaks\notag\\
&=\frac{1}{n}\frac{1}{\left|\Gamma'\right|}\sum_{\gamma\in\Gamma'}I(R_{uni};B_{q}^{\gamma\otimes n})\allowdisplaybreaks\notag\\
 &=\frac{1}{n}\frac{1}{\left|\Gamma'\right|}\sum_{\gamma\in\Gamma'}(H(R_{uni})-H(R_{uni}|B_{q}^{\gamma\otimes n}))\allowdisplaybreaks\notag\\
 &=\frac{1}{n}H(R_{uni})-\frac{1}{n}H(R_{uni}|B_{q}^{\gamma\otimes n},\Gamma')\allowdisplaybreaks\notag\\
 &\leq \frac{1}{n}H(R_{uni})-\frac{1}{n}H(R_{uni}|B_{q})+H(G_{uni})\allowdisplaybreaks\notag\\
 &=\frac{1}{n}I(R_{uni};B_{q}^{\otimes n})+H(G_{uni})\allowdisplaybreaks\notag\\
 &=\frac{1}{n}\chi\left(R_{uni};\frac{1}{\left|\Gamma'\right|}\sum_{\gamma\in \Gamma'}B_{q}^{\gamma\otimes n}\right)+H(G_{uni})\allowdisplaybreaks\notag\\
&=\frac{1}{n}\chi\left(R_{uni};\frac{1}{\left|\Gamma'\right|}\sum_{\gamma\in \Gamma'}B_{q}^{\gamma\otimes n}\right)+2 \log n \text{ .}\label{slgutbqnrn}\end{align} \vspace{0.2cm}

Let $\phi_{t^n}^{j,\gamma}$ be
the quantum state at the output of the wiretapper's channel
when the channel state is $t^n$, the output of the common randomness
is $\gamma$,  and $j$ has been sent.

We have
\begin{align}\allowdisplaybreaks[2]&
 \frac{1}{\left|\Gamma'\right|}\sum_{\gamma\in \Gamma'}\chi\left(R_{uni};Z_{t^n}^{\gamma}\right)
-\chi\left(R_{uni};\frac{1}{\left|\Gamma'\right|}\sum_{\gamma\in \Gamma'}Z_{t^n}^{\gamma}\right)\notag\\
&= \frac{1}{\left|\Gamma'\right|}\sum_{\gamma\in \Gamma'}S\left(\frac{1}{J_n}\sum_{j=1}^{J_n}\phi_{t^n}^{j,\gamma}\right)
-\frac{1}{\left|\Gamma'\right|}\frac{1}{J_n}\sum_{\gamma\in \Gamma'}\sum_{j=1}^{J_n}S\left(\phi_{t^n}^{j,\gamma}\right)\notag\\
&-S\left(\frac{1}{\left|\Gamma'\right|}\frac{1}{J_n}\sum_{\gamma\in \Gamma'}\sum_{j=1}^{J_n}\phi_{t^n}^{j,\gamma}\right)
+ \frac{1}{J_n}\sum_{j=1}^{J_n}S\left(\frac{1}{\left|\Gamma'\right|}\sum_{\gamma\in \Gamma'}\phi_{t^n}^{j,\gamma}\right)\text{ .}\label{slgutbqnrn2+}\end{align}

 Let $H^{\mathfrak{G}}$ be a $\left|\Gamma'\right|$-dimensional Hilbert space,
spanned by an orthonormal basis $\{|i\rangle : i = 1, \cdots, \left|\Gamma'\right|\}$. 
Let $H^{\mathfrak{J}}$ be a $J_n$-dimensional Hilbert space, spanned by an orthonormal basis 
$\{|j\rangle : j = 1, \cdots, J_n\}$. 
Similar to (\ref{stqtclpztr}), we define
\[\varphi^{\mathfrak{J}\mathfrak{G}H^{n}}:=\frac{1}{J_n}\frac{1}{\left|\Gamma'\right|}\sum_{j=1}^{J_n}\sum_{\gamma\in \Gamma'}
|j\rangle\langle j|\otimes|i\rangle\langle i|\otimes
\phi_{t^n}^{j,\gamma}\text{ .}\]

 By strong subadditivity of von Neumann entropy, it holds that $S(\varphi^{\mathfrak{J}H^{n}}) + S(\varphi^{\mathfrak{G}H^{n}})$
$\geq$ $S(\varphi^{H^{n}})+S(\varphi^{\mathfrak{J}\mathfrak{G}H^{n}})$, therefore
 \begin{equation}\frac{1}{\left|\Gamma'\right|}\sum_{\gamma\in \Gamma'}\chi\left(R_{uni};Z_{t^n}^{\gamma}\right) 
- \chi\left(R_{uni};\frac{1}{\left|\Gamma'\right|}\sum_{\gamma\in \Gamma'}Z_{t^n}^{\gamma}\right)
\geq 0\text{ .}\label{slgutbqnrn2}\end{equation} \vspace{0.2cm}

By (\ref{slgutbqnrn}) and (\ref{slgutbqnrn2}), we have
\[\chi(R_{uni};B_{q})
-\frac{1}{n}\chi(R_{uni};Z_{t^n}) + 2 \log n \geq
\frac{1}{\left|\Gamma'\right|}\sum_{\gamma\in \Gamma'}\frac{1}{n}
\left(\chi(R_{uni};B_{q}^{\gamma\otimes n})-\chi(R_{uni};Z_{t^n}^{\gamma})\right) \text{ .}\]

Thus for every $B_q \in Conv((B_s)_{s\in \theta})$ and every $t^n\in \theta^n$ we have
\begin{equation} \frac{1}{n}\log J_n\leq \frac{1}{1-\varepsilon}\frac{1}{n}
\left(\chi(R_{uni};B_{q}^{\otimes n})-\chi(R_{uni};Z_{t^n})+\epsilon_{n}+ 2 \frac{1}{n}\log n\right)
\text{  .}\label{tbf1nljnlfin}\end{equation}

Similar to the proof of Theorem \ref{loatbfis}, we have 
$\frac{1}{n}\Bigl(\inf_{B_q \in Conv((B_t)_{t\in \theta})}\chi(R_{uni};B_{q}^{\otimes n})$
$-\max_{t^n\in \theta^n}\chi(R_{uni};Z_{t^n})\Bigr)$ $\leq$
$ \frac{1}{n}	\max_{\Lambda_n}$
$\Bigl(\inf_{B_q \in Conv((B_t)_{t\in \theta})}$ $\chi(p_U;B_q^{\otimes n})$ $-\max_{t^n\in \theta^n}\chi(p_U;Z_{t^n})\Bigr)$.
The converse has been shown.
 (\ref{hatceque3}) and (\ref{tbf1nljnlfin}) prove Theorem \ref{commperm}.

 \qed
\end{proof}

\begin{corollary} Let $\{(W_t,{V}_t): t \in \theta\}$
 be an arbitrarily
varying classical-quantum wiretap channel.

1) Let $\mathsf{X}$ and $\mathsf{Y}$ be finite sets.
If $I(X,Y)>0$ holds for  a
random variable $(X,Y)$ which is distributed to a joint probability distribution
$p\in P(\mathsf{X},\mathsf{Y})$,
then the
 $(X,Y)$
  correlation assisted  secrecy capacity of $\{(W_t,{V}_t): t \in \theta\}$
is equal to \[\lim_{n\rightarrow \infty} \frac{1}{n}	\max_{\Lambda_n}
\Bigl(\inf_{B_q \in Conv((B_t)_{t\in \theta})}\chi(p_U;B_q^{\otimes n})-\max_{t^n\in \theta^n}\chi(p_U;Z_{t^n})\Bigr)\text{ .}\]

2)If the
arbitrarily varying  classical-quantum channel $\{W_t : t \in \theta\}$
is not symmetrizable, then  the deterministic secrecy capacity of $\{(W_t,{V}_t): t \in \theta\}$
is equal to \[ \lim_{n\rightarrow \infty} \frac{1}{n}\max_{\Lambda_n}
	\Bigl(\inf_{B_q \in Conv((B_t)_{t\in \theta})}\chi(p_U;B_q^{\otimes n})-\max_{t^n\in \theta^n}\chi(p_U;Z_{t^n})\Bigr)\text{ .}\]
\label{cocacommrdet}\end{corollary}

\begin{proof}

1) follows  immediately from Theorem \ref{commperm}
and the results of \cite{Bo/Ca/De}.\vspace{0.2cm}

To show 2)
 we use a technique similar to the proof of Theorem 3.1 in 
\cite{Bo/Ca/De}: We build a two-part code word
which consists of a non-secure code word and a common randomness assisted
secure code word. The first part is used to create the common randomness
for the sender and the legitimate receiver. The second part is a common randomness assisted secure
code word  transmitting the message to the legitimate receiver.

We consider the Markov chain
 $U\rightarrow A \rightarrow \{B_q^{\otimes n},Z_{t^n}:q,t_n\}$, where
we define the classical channel $U\rightarrow A$ by $T_U$.
Let
\[J_n =  \lfloor  2^{n\inf_{B_q \in Conv((B_s)_{s\in \theta})}\chi(p_U;B_q)
-\max_{t^n\in \theta^n}\chi(p_U;Z_{t^n})-n\delta} \rfloor  \text{ .}\]
By  Theorem \ref{commperm}, for any positive $\epsilon$
  if   $n$ is sufficiently large, 
there is  an   $(n, J_n)$ 
 code $\bigl(E^n, \{D_j^n : j = 1,\cdots J_n\}\bigr)$ 
for the  arbitrarily
varying classical-quantum wiretap channel $\{(W_t\circ T_U,{V}_t\circ T_U): t \in \theta\}$
 such that 
\[\frac{1}{n!}\sum_{\pi\in\Pi_n}\frac{1}{J_n} \sum_{j=1}^{J_n}  \sum_{a^n \in \mathsf{A}^n}
E^n(a^n|j)\mathrm{tr}(W_{t^n}(\pi^{-1}(a^n))P_{\pi}^{\dagger}D_j^n P_{\pi})\geq 1-\epsilon \]
and 
\[\frac{1}{n!}\sum_{\pi\in\Pi_n}\max_{t^n\in\theta^n}
\chi\left(R_{uni};Z_{t^n,\pi}\right)\leq \epsilon\text{ .}\]
By  Theorem 3.1.2 in
\cite{Bo/Ca/De},  for any positive $\lambda$  if   $n$ is sufficiently large, 
there is  an   $(n, J_n)$ common randomness  assisted
 code $\left\{\mathcal{C}_1,\mathcal{C}_2,\cdots,\mathcal{C}_{n^3}\right\}$
for the  arbitrarily
varying classical-quantum wiretap channel $\{(W_t\circ T_U,{V}_t\circ T_U): t \in \theta\}$
 such that
\[ \max_{t^n\in\theta^n}\frac{1}{n^3}\sum_{i=1}^{n^3}P_{e}(\mathcal{C}_{i},t^n)<
\lambda\text{ ,}\] and
\[\max_{t^n\in\theta^n} \frac{1}{n^3}\sum_{i=1}^{n^3}
\chi\left(R_{uni},Z_{\mathcal{C}_{i},t^n}\right) < \lambda\text{
.}\] 
Similar to the proof of  Theorem 3.1.1 in 
\cite{Bo/Ca/De},
for any positive $\vartheta$
if  $\{W_t : t \in \theta\}$
is not symmetrizable and
$n$ is sufficiently large, there is a code $\biggl(\Bigl(c^{\mu(n)}_i\Bigr)_{i\in\{1,\cdots,n^3\}},\{D_i^{\mu(n)}:
i\in\{1,\cdots,n^3\}\}\biggr)$ with deterministic encoder of length $\mu(n)$, where $2^{\mu(n)}=o(n)$
for the  arbitrarily
varying classical-quantum wiretap channel $\{(W_t,{V}_t): t \in \theta\}$
such that\[1- \frac{1}{n^3} \sum_{i=1}^{n^3}
\mathrm{tr}(W_{t^n}(c^{\mu(n)}_i)D_i^{\mu(n)})\leq \vartheta\text{ .}\]
We now can construct a code
$\mathcal{C}^{det} $ $=$ $\biggl(E^{\mu(n)+n},\Bigl\{D_{j}^{\mu(n)+n} :
j=1,\cdots,J_n\Bigr\}\biggr)$, where for $a^{\mu(n)
+n} = (a^{\mu(n)},a^n)\in{\mathsf{A}}^{\mu(n)+n}$ \[E^{\mu(n)+n}(a^{\mu(n)
+n}|j)=\begin{cases}
  \frac{1}{n^3}E^{n}_{i}(a^{n}|j) \text{ if } a^{\mu(n)} = c^{\mu(n)}_i\\
    0 \text{ else}             \end{cases}\text{ ,}
\] and \[D_{j}^{\mu(n)+n} := \sum_{i=1}^{n^3} D_i^{\mu(n)}\otimes D_{i,j}^{n} \text{ .}
\]
Similar to the proof  of Theorem 3.1.1 in 
\cite{Bo/Ca/De}, for any positive $\lambda$  if   $n$ is sufficiently large,  we have
\[ \max_{t^{\mu(n) +n}\in\theta^{\mu(n) +n}}P_{e}(\mathcal{C}^{det},t^{\mu(n) +n})< \lambda\text{ ,}\]
\[\max_{t^{\mu(n) +n}\in\theta^{\mu(n) +n}}
\chi\left(R_{uni},Z_{\mathcal{C}^{det},t^{\mu(n) +n}}\right) < \lambda\text{
.}\] \qed
\end{proof}

\begin{remark} For the proof of Corollary \ref{cocacommrdet}, 2), it is important to assume
that
$\biggl(\Bigl(c^{\mu(n)}_i\Bigr)_{i\in\{1,\cdots,n^3\}},\{D_i^{\mu(n)}:
i\in\{1,\cdots,n^3\}\}\biggr)$
is a code for the channel $\{(W_t,{V}_t): t \in \theta\}$ and not for $\{(W_t\circ T_U,{V}_t\circ T_U): t \in \theta\}$, since it may
happen that $\{W_t\circ T_U: t \in \theta\}$ is symmetrizable although $\{W_t: t \in \theta\}$ is not
symmetrizable, as the following example shows:

We assume that  $\{{W}_t: t \in \theta\}$ :$P(\mathsf{A})\rightarrow \mathcal{S}(H)$
is not symmetrizable, but there is a subset $\mathsf{A}'\subset \mathsf{A}$ such that
 $\{{W}_t: t \in \theta\}$ limited on $\mathsf{A}'$ is symmetrizable.  We choose
a $T_U$ such that  for every  $u\in \mathsf{U}$ there is $a\in \mathsf{A}'$
such that $T_U(a\mid u) = 1$, and 
$T_U(a\mid u) = 0$ for all $a\in \mathsf{A}\setminus \mathsf{A}'$ and $u\in \mathsf{U}$.
It is clear that $\{{W}_t\circ T_U: t \in \theta\}$ is symmetrizable
(cf. also \cite{Wi/No/Bo2} for an example for classical channels).
\end{remark}


\section{Investigation of Secrecy Capacity's Continuity}\label{iosccits}

In this section we show that
the secrecy capacity of
an arbitrarily
varying classical-quantum wiretap channel under
common randomness assisted quantum coding
is  continuous 
in the following sense:

\begin{corollary}
For an arbitrarily
varying classical-quantum wiretap channel 
$\{(W_t,{V}_t): t \in \theta\}$, where
$W_{t}$  $:$
$\mathsf{P}(\mathsf{A}) \rightarrow \mathcal{S}(H)$ and ${V}_t$ $:$ $\mathsf{P}(\mathsf{A})
\rightarrow \mathcal{S}(H')$
 and
a positive $\delta$, let
$\mathsf{C}_{\delta}$ be the set of all
arbitrarily
varying classical-quantum wiretap channels
 $\{({W'}_t,{V'}_t): t \in \theta\}$,
where
${W'}_{t}$  $:$
$\mathsf{P}(\mathsf{A}) \rightarrow \mathcal{S}(H)$ and ${V'}_t$ $:$ $\mathsf{P}(\mathsf{A})
\rightarrow \mathcal{S}(H')$,
 such
that
\[\max_{ a\in \mathsf{A}} \|W_t(a)- {W'}_t(a)\|_{1} <  \delta\]
and
\[\max_{ a\in \mathsf{A}} \|V_t(a)- {V'}_t(a)\|_{1} <  \delta\]
for all $t \in \theta$.

For any positive $\epsilon$ there is a positive $\delta$
such that for all  $\{({W'}_t,{V'}_t): t \in \theta\}$
$\in$ $\mathsf{C}_{\delta}$ we have
	\begin{equation} |C_s(\{(W_t,{V}_t): t \in \theta\};cr)
	-C_s(\{(({W'}_t,{V'}_t): t \in \theta\};cr)| \leq \epsilon
	\text{ .}\end{equation}\label{eetelctit}
\end{corollary}

\begin{proof}
By Corollary
\ref{cocacommrdet},
the secrecy capacity of
$\{(W_t,{V}_t): t \in \theta\}$ is
	\[\lim_{n\rightarrow \infty} \frac{1}{n}
\max_{\Lambda_n}\Bigl(\inf_{B_q \in Conv((B_t)_{t\in \theta})}\chi(p_U;B_q^{\otimes n})-
 \max_{t^n\in \theta^n}\chi(p_U;Z_{t^n})\Bigr)
	\text{ ,}\] 
and for every $\{({W'}_t,{V'}_t): t \in \theta\}$
$\in$ $\mathsf{C}_{\delta}$ the secrecy capacity of
 $\{({W'}_t,{V'}_t): t \in \theta\}$ is
	\[ \lim_{n\rightarrow \infty} \frac{1}{n}
\max_{\Lambda_n}\Bigl(\inf_{{B'}_q \in Conv(({B'}_t)_{t\in \theta})}\chi(p_U;{B'}_q^{\otimes n})
-\max_{t^n\in \theta^n}\chi(p_U;{Z'}_{t^n})\Bigr)
	\text{ ,}\] 
where ${B'}_t$ is the resulting  quantum state at the output of
${W'}_t$ and ${Z'}_t$ is the resulting  quantum state  at
the output of ${V'}_t$.\vspace{0.2cm}

To analyze $\vert\chi(p;Z_{t^n})-\chi(p;{Z'}_{t^n})\vert$,
we use the technique introduced in
\cite{Le/Sm} and apply the following lemma given in \cite{Al/Fa}. 

\begin{lemma}[Alicki-Fannes Inequality] Suppose we have
a composite system $\mathfrak{PQ}$ with components
 $\mathfrak{P}$ and $\mathfrak{Q}$. Let $G^\mathfrak{P}$ and $G^\mathfrak{Q}$ be
 Hilbert space of $\mathfrak{P}$ and $\mathfrak{Q}$, respectively.  
Suppose we have two
bipartite
quantum states $\phi^\mathfrak{PQ}$ and
 $\sigma^\mathfrak{PQ}$ in $\mathcal{S}(G^\mathfrak{PQ})$ such that
$\|\phi^\mathfrak{PQ}-\sigma^\mathfrak{PQ}\|_{1} = \epsilon <1$,  it holds that
\begin{equation}S(\mathfrak{P}\mid\mathfrak{Q})_{\rho}- S(\mathfrak{P}\mid\mathfrak{Q})_{\sigma}
\leq 4 \epsilon \log(d-1) - 2h(\epsilon)\text{ ,}\end{equation}
where $d$ is the dimension of $G^\mathfrak{P}$ and $h(\epsilon)$ is defined as in Lemma \ref{eq_9}.\label{AFLswhacs}
\end{lemma}

 In contrast to \cite{Al/Fa},  we consider here
classical-quantum channels instead of quantum-quantum channels.

We fix an $n\in\mathbb{N}$ and a $t^n$ $= $ $(t_1,\cdots t_n)$ $\in\theta^n$.
For any $a^n\in \mathsf{A}^n$ we have
\begin{align*}&\left|S\left({V}_{t^n}(a^n)\right)-S\left({V'}_{t^n}(a^n)\right)\right|\\
&= \biggl|\sum_{k=1}^{n} S\left({V}_{(t_1,\cdots t_{k-1})}\otimes{V'}_{(t_{k},\cdots t_n)}(a^n )\right)
-S\left({V}_{(t_1,\cdots t_k)}\otimes{V'}_{(t_{k+1},\cdots t_n)}(a^n)\right)\biggr|\\
&\leq \sum_{k=1}^{n} \biggl|S\left({V}_{(t_1,\cdots t_{k-1})}\otimes {V'}_{(t_{k},\cdots t_n)}(a^n)\right)
-S\left({V}_{(t_1,\cdots t_k)}\otimes{V'}_{(t_{k+1},\cdots t_n)}(a^n)\right)\biggr|\text{ .}
\end{align*}\vspace{0.2cm}

For a $k\in\{1,\cdots, n\}$ and $a^n$ $= $ $(a_1,\cdots a_n)$ $\in \mathsf{A}^n$ by Lemma \ref{AFLswhacs}
we have
\begin{align*}&\biggl|S\left({V}_{(t_1,\cdots t_{k+1})}\otimes {V'}_{(t_{k},\cdots t_n)}(a^n)\right)
-S\left({V}_{(t_1,\cdots t_{k+1})}\otimes{V'}_{(t_{k+1},\cdots t_n)}(a^n)\right)\biggr|\\
&=\biggl|S\left({V}_{(t_1,\cdots t_k)}\otimes {V'}_{(t_{k},\cdots t_n)}(a^n)\right)
-S\left({V}_{(t_1,\cdots t_{k-1})}\otimes {V'}_{(t_{k+1},\cdots t_n)}((a_1,\cdots a_{k-1}, a_{k+1},\cdots a_n))\right)\\
&-S\left({V}_{(t_1,\cdots t_k)}\otimes{V'}_{(t_{k+1},\cdots t_n)}(a^n)\right)
+ S\left({V}_{(t_1,\cdots t_{k-1})}\otimes {V'}_{(t_{k+1},\cdots t_n)}((a_1,\cdots a_{k-1}, a_{k+1},\cdots a_n))\right)\biggr|\\
&= \biggl| S\left( {V'}_{t_{k}}(a_k)\mid  {V}_{(t_1,\cdots t_{k-1})}\otimes {V'}_{(t_{k+1},\cdots t_n)}((a_1,\cdots a_{k-1}, a_{k+1},\cdots a_n))  \right)\\
&-S\left( {V}_{t_{k}}(a_k)\mid  {V}_{(t_1,\cdots t_{k-1})}\otimes {V'}_{(t_{k+1},\cdots t_n)}((a_1,\cdots a_{k-1}, a_{k+1},\cdots a_n))  \right) \biggr|\\
&\leq  4 \delta \log(d_E-1) - 2\cdot h(\delta) \text{ ,}
\end{align*}
where $d_E$ is the dimension of $H^\mathfrak{E}$.\vspace{0.2cm}

Thus, 
\begin{equation}\left|S\left({V}_{t^n}(a^n)\right)-
S\left({V'}_{t^n}(a^n)\right)\right|\leq  4 n\delta \log(d_E-1) - 2n\cdot h(\delta) \text{ .}
\label{lswtnanr}\end{equation}

For any probability distribution $p\in \mathsf{P}(\mathsf{A})$, $n\in\mathbb{N}$, and $t^n\in\theta^n$, we have
\begin{align}&
\vert \chi(p;Z_{t^n}) - \chi(p;{Z'}_{t^n}) \vert \notag\\
&=\Bigl\vert S(\sum_{a}p(a)V_{t^n}(a))- \sum_{a}p(a)S({V}_{t^n}(a))  \notag\\
&- S(\sum_{a}p(a){V'}_{t^n}(a)) + S(\sum_{a}p(a){V'}_{t^n}(a)) \Bigr\vert \notag\\
&\leq\Bigl\vert S(\sum_{a}p(a){V}_{t^n}(a))- S(\sum_{a}p(a){V'}_{t^n}(a)) \Bigr\vert \notag\\
&+ \Bigl\vert \sum_{a}p(a)S({V'}_{t^n}(a))  - \sum_{a}p(a)S({V'}_{t^n}(a)) \Bigr\vert \notag\\
&\leq 8 n\delta \log(d_E-1) - 4n\cdot h(\delta)\text{ .}\label{lbndln14n}
\end{align}

\vspace{0.2cm}

We fix a  probability distribution $q$ on $\theta$, a probability distribution $p\in \mathsf{P}(\mathsf{A})$,
  and an $n\in\mathbb{N}$. By Lemma \ref{eq_9} 
we have

\begin{align}&
\vert \chi(p;B_q) - \chi(p;{B'}_q) \vert \notag\\
&=\Bigl\vert \sum_{t}q(t)S(\sum_{a}p(a)W_t(a))- \sum_{t}\sum_{a}q(t)p(a)S(W_t(a))  \notag\\
&- \sum_{t}q(t)S(\sum_{a}p(a){W'}_t(a)) + S(\sum_{t}\sum_{a}q(t)p(a){W'}_t(a)) \Bigr\vert \notag\\
&\leq\Bigl\vert \sum_{t}q(t)S(\sum_{a}p(a)W_t(a))- \sum_{t}q(t)S(\sum_{a}p(a){W'}_t(a)) \Bigr\vert \notag\\
&+ \Bigl\vert \sum_{t}\sum_{a}q(t)p(a)S(W_t(a))  - S(\sum_{t}\sum_{a}q(t)p(a){W'}_t(a)) \Bigr\vert \notag\\
&\leq 8 \delta \log(d_B-1) - 4\cdot h(\delta)\text{ ,}
\end{align}\vspace{0.2cm}
where $d_B$ is the dimension of $H^\mathfrak{B}$.\vspace{0.2cm}

Thus, for any probability distribution $q$ on $\theta$, 
$n\in\mathbb{N}$,  $p\in \mathsf{P}(\mathsf{A})$, and $t^n\in\theta^n$, we have for
all $\{({W'}_t,{V'}_t): t \in \theta\}$
$\in$ $\mathsf{C}_{\delta}$ 
\begin{align}&
\Bigl\vert  (\chi(p;B_q)-\frac{1}{n}\chi(p;Z_{t^n}))
-  (\chi(p;{B'}_q)-\frac{1}{n}\chi(p;{Z'}_{t^n}))\Bigr\vert\notag\\
&\leq 
8 \delta \log(d_B-1) + 8 \delta \log(d_E-1)-
 8\cdot h(\delta)
	\text{ .}\end{align}
	
For any positive $\epsilon$ we can find a positive $\delta$
such that $8 \delta \log(d_B-1)$ $ + $ $8 \delta \log(d_E-1)$
 $-$ $8 \cdot h(\delta)$ 
$\leq$ $\epsilon$.

Thus for all $n\in\mathbb{N}$ and any positive $\epsilon$ we can find a positive $\delta$
such that for
all $\{({W'}_t,{V'}_t): t \in \theta\}$
$\in$ $\mathsf{C}_{\delta}$ 
	\begin{align}&\Bigl\vert
(\max_{p}\inf_{B_q \in Conv((B_t)_{t\in \theta})}\chi(p;B_q)- \max_{t^n\in \theta^n}\chi(p;Z_{t^n}))\notag\\
&	- (\max_{p}\inf_{{B'}_q \in Conv(({B'}_t)_{t\in \theta})}\chi(p;{B'}_q)
-\frac{1}{n}\max_{t^n\in \theta^n}\chi(p;{Z'}_{t^n}))\Bigr\vert\notag\\
&\leq 	\epsilon\text{ .}\label{bv1nmbibq1}\end{align}
	
	(\ref{bv1nmbibq1}) shows Corollary \ref{eetelctit}. \qed
\end{proof}

 \begin{corollary}
The deterministic secrecy capacity of
an arbitrarily
varying classical-quantum wiretap channel
is in general not continuous.
\label{tscaav}\end{corollary}

\begin{proof}

We show Corollary \ref{tscaav} by giving an example.\vspace{0.15cm}

Let $\theta:=\{1,2\}$. Let $\mathsf{A}$ $=$ $\{0,1\}$.
Let ${H}^{\mathfrak{B}}$ $=$ $\mathbb{C}^{5}$.
Let $\{|0\rangle^{\mathfrak{B}}, |1\rangle^{\mathfrak{B}}, 
|2\rangle^{\mathfrak{B}}, |3\rangle^{\mathfrak{B}}, |4\rangle^{\mathfrak{B}}\}$ be a set of orthonormal vectors
on ${H}^{\mathfrak{B}}$. Let $\lambda$ be $\in [0,1]$.

For $r\in[0,1]$, let $P_{r}$ be the probability distribution on $\mathsf{A}$
such that $P_{r}(0)=r$ and $P_{r}(1)=1-r$.
 We define a  channel $W_{1}^{\lambda}$ $:\mathsf{P}(\mathsf{A})$ $\rightarrow$ $\mathcal{S}({H}^{\mathfrak{B}})$
by 
\[W_{1}^{\lambda} (P_{r})= (1-\lambda) r |0\rangle\langle 0|^{\mathfrak{B}} 
+ (1-\lambda) (1-r)|1\rangle\langle 1|^{\mathfrak{B}}
+\lambda |3\rangle\langle 3|^{\mathfrak{B}}\text{ ,}\]
and a  channel $W_{2}^{\lambda}$ $:\mathsf{P}(\mathsf{A})$ $\rightarrow$ $\mathcal{S}({H}^{\mathfrak{B}})$
by 
\[W_{2}^{\lambda} (P_{r})= (1-\lambda) r |1\rangle\langle 1|^{\mathfrak{B}} 
+ (1-\lambda) (1-r)|2\rangle\langle 2|^{\mathfrak{B}}
+\lambda |4\rangle\langle 4|^{\mathfrak{B}}\text{ .}\]

In other words:
\[W_{1}^{\lambda} (0)= (1-\lambda)  |0\rangle\langle 0|^{\mathfrak{B}} 
+\lambda |3\rangle\langle 3|^{\mathfrak{B}}\text{ ,}\]
\[W_{1}^{\lambda} (1)=  (1-\lambda) |1\rangle\langle 1|^{\mathfrak{B}}
+\lambda |3\rangle\langle 3|^{\mathfrak{B}}\text{ ,}\]

\[W_{2}^{\lambda} (0)= (1-\lambda)  |1\rangle\langle 1|^{\mathfrak{B}} 
+\lambda |4\rangle\langle 4|^{\mathfrak{B}}\text{ ,}\]
\[W_{2}^{\lambda} (1)=  (1-\lambda) |2\rangle\langle 2|^{\mathfrak{B}}
+\lambda |4\rangle\langle 4|^{\mathfrak{B}}\text{ .}\]\vspace{0.2cm}

Let ${H}^{\mathfrak{E}}$ $=$ $\mathbb{C}^{5}$.
Let $\{|0\rangle^{\mathfrak{E}}, |1\rangle^{\mathfrak{E}}, 
|2\rangle^{\mathfrak{E}}, |3\rangle^{\mathfrak{E}}, |4\rangle^{\mathfrak{E}}$ be a set of orthonormal vectors
on ${H}^{\mathfrak{E}}$.

 We define a  channel $V_{1}^{\lambda}$ $:\mathsf{P}(\mathsf{A})$ $\rightarrow$ $\mathcal{S}({H}^{\mathfrak{E}})$
by 
\[V_{1}^{\lambda} (P_{r})= \lambda r |0\rangle\langle 0|^{\mathfrak{E}} 
+ \lambda (1-r)|1\rangle\langle 1|^{\mathfrak{E}}
+ (1-\lambda)|3\rangle\langle 3|^{\mathfrak{E}}\text{ ,}\]
and a  channel $V_{2}^{\lambda}$ $:\mathsf{P}(\mathsf{A})$ $\rightarrow$ $\mathcal{S}({H}^{\mathfrak{E}})$
by 
\[V_{2}^{\lambda} (P_{r})=  \lambda r |1\rangle\langle 1|^{\mathfrak{E}} 
+  \lambda (1-r)|2\rangle\langle 2|^{\mathfrak{E}}
+(1-\lambda) |4\rangle\langle 4|^{\mathfrak{E}}\text{ .}\]

In other words:
\[V_{1}^{\lambda} (0)= \lambda  |0\rangle\langle 0|^{\mathfrak{E}} 
+ (1-\lambda)|3\rangle\langle 3|^{\mathfrak{E}}\text{ ,}\]
\[V_{1}^{\lambda} (1)= \lambda |1\rangle\langle 1|^{\mathfrak{E}}
+ (1-\lambda)|3\rangle\langle 3|^{\mathfrak{E}}\text{ ,}\]

\[V_{2}^{\lambda} (0)= \lambda  |1\rangle\langle 1|^{\mathfrak{E}} 
+ (1-\lambda)|4\rangle\langle 4|^{\mathfrak{E}}\text{ ,}\]
\[V_{2}^{\lambda} (1)= \lambda |2\rangle\langle 2|^{\mathfrak{E}}
+ (1-\lambda)|4\rangle\langle 4|^{\mathfrak{E}}\text{ .}\]\vspace{0.2cm}

For every $a\in \mathsf{A}$ and $t\in\theta$ we have 
\begin{align*}&\|W_{t}^{0}(a)- W_{t}^{\lambda}(a)\|_{1}\\
&=\|\lambda |t+a-1\rangle\langle t+a-1|^{\mathfrak{B}} - \lambda|t+2\rangle\langle t+2|^{\mathfrak{B}}\|_{1}\\
& = 2\lambda \end{align*} and
\begin{align*}&\|V_{t}^{0}(a)- V_{t}^{\lambda}(a)\|_{1}\\
&=\|-\lambda |t+a-1\rangle\langle t+a-1|^{\mathfrak{E}} + \lambda|t+2\rangle\langle t+2|^{\mathfrak{E}}\|_{1}\\
& = 2\lambda \text{ .}\end{align*}

$\{(W_t^{\lambda},{V}_t^{\lambda}): t \in \theta\}$
defines an arbitrarily
varying classical-quantum wiretap channel
for every $\lambda\in[0,1]$.\vspace{0.2cm} 

At first, we  consider $\{(W_t^{0},{V}_t^{0}): t \in \theta\}$.\vspace{0.2cm}

\it i) The deterministic
secrecy capacity of $\{(W_t^{0},{V}_t^{0}): t \in \theta\}$
 is equal to zero.
  \rm\vspace{0.2cm}

We set
\begin{align*}&\tau(1\mid 0) = 0\text{ ; } ~~\tau(2\mid 0) = 1\text{ ;}\\
&\tau(1\mid 1) = 1\text{ ; } ~~\tau(2\mid 1) = 0\text{ .}\end{align*}

It holds that

\[\sum_{t\in\theta}\tau(t\mid 0)W_{t}^{0}({1})=|1\rangle\langle 1|^{\mathfrak{E}} 
=\sum_{t\in\theta}\tau(t\mid {1})W_{t}^{0}(0)\text{ ,}\]
and of course for every  $a\in \mathsf{A}$ 
\[\sum_{t\in\theta}\tau(t\mid a)W_{t}^{0}({a})=\sum_{t\in\theta}\tau(t\mid {a})W_{t}^{0}(a)\text{ .}\]
$\{(W_t^{0}): t \in \theta\}$ is therefore symmetrizable.
By \cite{Bo/Ca/De}, 
we have

\begin{equation} C_s(\{(W_t^{0},{V}_t^{0}): t \in \theta\})=0\text{ .}\label{foracswtvtcr1}\end{equation}\vspace{0.2cm}

\it ii) The
secrecy capacity of $\{(W_t^{0},{V}_t^{0}): t \in \theta\}$
under common randomness assisted quantum
coding
is positive.  \rm\vspace{0.2cm}

We denote by $p'\in \mathsf{P}(\mathsf{A})$ the distribution on $\mathsf{A}$ such that $p'(1)=p'(2)=\frac{1}{2}$.
Let $q\in[0,1]$.  We define $Q(1)=q$,   $Q(2)=1-q$. We have
\begin{align*}
&\chi\left(p',\{W_Q^{0}(a): a\in {\mathsf{A}}\}\right)\\
&=-\frac{1}{2}q\log \frac{1}{2}q + \frac{1}{2}(1-q)\log\frac{1}{2}(1-q)
-\frac{1}{2}\log\frac{1}{2}\\
&+q\log q + (1-q)\log(1-q)\text{ .}\end{align*}
When we differentiate  this term    by $q$,
we obtain
\begin{align*}&\frac{1}{\log e}\biggl(-\frac{1}{2}\log \frac{1}{2}q - \frac{1}{2}+ 
\frac{1}{2}\log\frac{1}{2}(1-q)+\frac{1}{2}
+\log q +1- \log(1-q)-1\biggr)\\
&=\frac{1}{2\log e}\left(\log q - \log(1-q)\right)\text{ .}\end{align*}
$\log q - \log(1-q)$ is equal to zero if and only if $q=\frac{1}{2}$. By further calculation,
one can show
that  $\chi\left(p',\{W_Q^{0}(a): a\in {\mathsf{A}}\}\right)$ achieves its minimum when  $q=\frac{1}{2}$.  This
minimum is equal to $-\frac{1}{2}\log\frac{1}{4}+\frac{1}{2}\log\frac{1}{2}$ $=$ $\frac{1}{2}$ $>0$.
Thus, \[\max_p\min_{q} \chi\left(p,B_q^{0}\right)\geq\frac{1}{2}\text{ .}\]

For  all $t\in \theta$ it holds that
${V}_{t}^{0} (0) = {V}_{t}^{0} (1)$; 
therefore for  all $t^n\in \theta^n$ and any $p^n\in \mathsf{P}(\mathsf{A}^n)$  we have
\begin{align*}&\chi(p;Z_{t^n}^{0})\\
&= S({V}_{t^n}^{0} (p^n) ) - \sum_{a^n\in \mathsf{A}^n} p^n(a^n)S({V}_{t^n}^{0} (a^n) ) \\
&= S({V}_{t^n}^{0} (0^n) ) - \sum_{a^n\in \mathsf{A}^n} p^n(a^n)S({V}_{t^n}^{0} (0^n) ) \\
&=0\text{ .}\end{align*} Thus, 

\begin{equation} C_s(\{(W_t^{0},{V}_t^{0}): t \in \theta\},cr) \geq  \frac{1}{2}-0 >0\text{ .}\label{foracswtvtcr}\end{equation}
 \vspace{0.2cm}

Now we consider $\{(W_t^{\lambda},{V}_t^{\lambda}): t \in \theta\}$
when $\lambda\not= 0$. \vspace{0.2cm}

\it iii) When $\lambda\not= 0$,
the deterministic secrecy capacity of
$\{(W_t^{\lambda},{V}_t^{\lambda}): t \in \theta\}$ is equal to its
secrecy capacity of 
under common randomness assisted quantum
coding.\rm \vspace{0.2cm}

We suppose that for any $a,a'\in \mathsf{A}$ there are two  distributions $\tau(\cdot\mid a)$ and $\tau(\cdot\mid a')$
 on $\theta$ such that
\begin{align}&\sum_{t\in\theta} \tau(t\mid a')\cdot W_{t}^{\lambda} (a) 
= \sum_{t\in\theta} \tau(t\mid a) \cdot W_{t}^{\lambda} (a')\notag\\
& \Rightarrow (1-\lambda) \sum_{t\in\theta} \tau(t\mid a') |t+a-1\rangle\langle t+a-1|^{\mathfrak{B}} +
\lambda  \tau(1\mid a')|3\rangle\langle 3|^{\mathfrak{E}} +   \lambda  \tau(2\mid a')|4\rangle\langle 4|^{\mathfrak{E}}\notag\\
&=  (1-\lambda) \sum_{t\in\theta} \tau(t\mid a)|t+a'-1\rangle\langle t+a'-1|^{\mathfrak{B}} +
\lambda  \tau(1\mid a)|3\rangle\langle 3|^{\mathfrak{E}} +   \lambda  \tau(2\mid a)|4\rangle\langle 4|^{\mathfrak{E}}
 \text{ .}\label{em4lr42ap}\end{align}
Since $|t+a-1\rangle\langle t+a-1|^{\mathfrak{B}} $ $\in$
$\Bigl\{|0\rangle\langle 0|^{\mathfrak{E}},$ $|1\rangle\langle 1|^{\mathfrak{E}},$ 
$|2\rangle\langle 2|^{\mathfrak{E}}\Bigr\}$ for all $t$ and $a$, if $\lambda\not= 0$,  (\ref{em4lr42ap}) 
 implies that 
\[ \tau(t\mid a')= \tau(t\mid a)\]
for all $t\in\theta$.  This means we have a
 distribution $\acute{p}$ on $\theta$ such that 
$\acute{p}(t) = \tau(t\mid a)$ for all $a \in \mathsf{A}$.

But there is 
clearly no such distribution $\acute{p}$ such that
$\sum_{t\in\theta}\acute{p}(t) W_{t}^{\lambda}(0)$ $=$ $\sum_{t\in\theta}\acute{p}(t) W_{t}^{\lambda}(1)$,
because then we would have
\begin{align*}&\acute{p}(1) |0\rangle\langle 0|^{\mathfrak{B}} + \acute{p}(2) |1\rangle\langle 1|^{\mathfrak{B}} \\
&= \acute{p}(1) |1\rangle\langle 1|^{\mathfrak{B}} + \acute{p}(2) |2\rangle\langle 2|^{\mathfrak{B}} \text{ .}\end{align*}
This would mean $\acute{p}(1)=\acute{p}(2)=0$, which obviously cannot be true.
Thus, $(W_{t}^{\lambda})_{t\in\theta}$ is not symmetric.

By \cite{Bo/Ca/De}, if $\lambda\not= 0$

\begin{equation} C_s(\{(W_t^{\lambda},{V}_t^{\lambda}): t \in \theta\}) 
= C_s(\{(W_t^{\lambda},{V}_t^{\lambda}): t \in \theta\},cr)\text{ .}\label{foracswtvtcr3}\end{equation}
\vspace{0.2cm}

When $\lambda \searrow 0$ for every $a\in \mathsf{A}$ and $t\in\theta$ we have
$\|W_{t}^{0}(a)- W_{t}^{\lambda}(a)\|_{1}$ $=$ 
$\|V_{t}^{0}(a)- V_{t}^{\lambda}(a)\|_{1} = 2\lambda$
$\searrow 0$.

By Corollary \ref{eetelctit},
the secrecy capacity of $\{(W_t^{\lambda},{V}_t^{\lambda}): t \in \theta\}$
under common randomness assisted quantum
coding is continues. Thus
for any positive $\varepsilon$ there is a  $\delta$, such that for all $\lambda\in ]0,\delta[$, we have

\begin{equation}  C_s(\{(W_t^{\lambda},{V}_t^{\lambda}): t \in \theta\}) \geq
C_s(\{(W_t^{0},{V}_t^{0}): t \in \theta\},cr)- \varepsilon \geq \frac{1}{2}-\varepsilon\text{ .}\label{foracswtvtcr4}\end{equation}

In other words,
when $\lambda\not= 0$ tends to
zero, the deterministic secrecy capacity of
$\{(W_t^{\lambda},{V}_t^{\lambda}): t \in \theta\}$
tends to the  secrecy capacity of $\{(W_t^{0},{V}_t^{0}): t \in \theta\}$
under common randomness assisted quantum
coding, which is positive, but
the deterministic secrecy capacity of
$\{(W_t^{0},{V}_t^{0}): t \in \theta\}$ is equal to zero.
Hence, the deterministic secrecy capacity of
$\{(W_t^{\lambda},{V}_t^{\lambda}): t \in \theta\}$
is  not continues at zero. \qed
\end{proof}

Corollary \ref{tscaav} shows that small
errors in the description of an arbitrarily varying classical-quantum
 wiretap channel may have severe consequences on the secrecy
capacity. Corollary \ref{eetelctit}  shows that 
resources are very helpful to protect these consequences.

\section{Conclusion}\label{conclav}
In this paper, we
 deliver the  formula for
the secrecy capacities under common
randomness assisted coding of arbitrarily varying classical-quantum
wiretap  channels. In  our
previous paper 
\cite{Bo/Ca/De}, we established the Ahlswede Dichotomy for arbitrarily varying classical-quantum
wiretap  channels:  Either the deterministic secrecy capacity of an
arbitrarily varying classical-quantum wiretap channel is zero or it
equals its randomness assisted secrecy
 capacity, depending on the status whether the legitimate
receiver's channel is symmetrizable or not.
  When we combine the results of these two works
we can now completely characterize  the secrecy capacity formulas for
arbitrarily varying classical-quantum
wiretap  channels (cf. Corollary  \ref{cocacommrdet}).

 As an application of these results, we
turn to the general question: When is
secure message transmission through
arbitrarily varying classical-quantum
wiretap  channels continuous?
Our results show the discontinuity in general
 and
demonstrate the importance of shared randomness: it
stabilizes the secure message transmission through
arbitrarily varying classical-quantum
wiretap  channels.

\section*{Acknowledgment}
Support by the Bundesministerium f\"ur Bildung und Forschung (BMBF)
via Grant 16KIS0118K and  16KIS0117K is gratefully acknowledged.


\begin{thebibliography}{xxx}

\bibitem{Ahl0} R. Ahlswede, A note on the existence of the weak capacity for channels with arbitrarily
varying channel probability functions and its relation to Shannon's zero error
capacity, The Annals of Mathematical Statistics, Vol. 41, No. 3, 1970.
\bibitem{Ahl1} R. Ahlswede,  Elimination of correlation in random codes for arbitrarily varying channels,
Z. Wahrscheinlichkeitstheorie verw. Gebiete
Vol. 44, 159-175, 1978.
\bibitem{Ahl2} R. Ahlswede,  Coloring hypergraphs: a new approach to multi-user
source coding-II, Journal of Combinatorics, Information \& System
Sciences
Vol. 5, No. 3, 220-268, 1980.
\bibitem{Ahl3} R. Ahlswede, Arbitrarily varying channels with states
sequence
known to the sender,
IEEE Trans. Inf. Th., Vol. 32, 621-629, 1986.
\bibitem{Ahl/Bj/Bo/No}R. Ahlswede,  I. Bjelakovi\'{c}, H. Boche,  and J. N\"otzel,
Quantum capacity under adversarial quantum noise:  arbitrarily
varying quantum channels, Comm. Math. Phys. A, Vol. 317, No. 1, 103-156,
2013.
\bibitem{Ahl/Bli} R. Ahlswede and  V. Blinovsky,  Classical capacity of classical-quantum arbitrarily
varying channels, IEEE Trans. Inform. Theory, Vol. 53, No. 2,
526-533, 2007.
\bibitem{Ahl/Win} R. Ahlswede and A. Winter, Strong converse for identification
via quantum channels, IEEE Trans. Inform. Theory, Vol. 48, No. 3,
569-579, 2002. Addendum: IEEE Trans. Inform. Theory, Vol. 49, No. 1,
346, 2003.
\bibitem{Al/Fa} R. Alicki and M. Fannes, Continuity of quantum conditional information,
J. Phys. A: Math. Gen., Vol. 37, L55-L57, 2004.
\bibitem{Au} K. M. R. Audenaert, A sharp continuity estimate for the von Neumann entropy, J. Phys. A: Math. Theor., Vol. 40,
 8127-8136,  2007.
\bibitem{Be} C. H. Bennett, Quantum cryptography using any two non-orthogonal states, Physical
Review Letters, Vol. 68, 3121-3124,  1992.
\bibitem{Be/Br} C. H. Bennett and G. Brassard, Quantum cryptography: public key distribution and coin tossing,
 Proceedings of the IEEE International Conference on Computers, Systems, and Signal Processing, Bangalore,  175, 1984.
\bibitem{Bj/Bo} I. Bjelakovi\'{c} and H. Boche, Classical capacities of
averaged and compound quantum channels. IEEE Trans. Inform. Theory,
Vol. 57, No. 7, 3360-3374, 2009.
\bibitem{Bj/Bo/Ja/No} I. Bjelakovi\'{c}, H. Boche, G. Jan\ss en, and J. N\"otzel,
Arbitrarily varying and compound classical-quantum channels and a
note on quantum zero-error capacities,  Information Theory, Combinatorics, and Search Theory, 
in Memory of
Rudolf Ahlswede, H. Aydinian, F. Cicalese, and C. Deppe eds., LNCS
Vol.7777,  247-283, \tt
arXiv:1209.6325\rm, 2012.
\bibitem{Bj/Bo/No} I. Bjelakovi\'{c}, H. Boche,  and J. N\"otzel,
Entanglement transmission and generation under channel uncertainty: universal quantum channel coding,
Communications in Mathematical Physics,
Vol. 292, No. 1, 55-97,  2009.
\bibitem{Bj/Bo/So} I. Bjelakovi\'{c}, H. Boche, and J. Sommerfeld,
Secrecy  results for compound wiretap channels, Problems of Information Transmission,
Vol. 59, No. 3, 1405-1416, 2013.
\bibitem{Bj/Bo/So2}I. Bjelakovi\'{c}, H. Boche,
and J. Sommerfeld, Capacity results for arbitrarily varying wiretap
channels, 
Information Theory, Combinatorics, and Search Theory, in Memory of
Rudolf Ahlswede, H. Aydinian, F. Cicalese, and C. Deppe eds., LNCS
Vol.7777, 114-129, \tt arXiv:1209.5213\rm, 2012.
\bibitem{Bl/Br/Th2} D. Blackwell, L. Breiman, and A. J.
Thomasian, The capacities of a certain channel classes under random
coding, Ann. Math. Statist. Vol. 31, No. 3, 558-567, 1960.
\bibitem{Bl/Ca}V. Blinovsky and M. Cai, arbitrarily classical-quantum varying wiretap
channel, Information Theory, Combinatorics, and Search Theory, in Memory of
Rudolf Ahlswede, H. Aydinian, F. Cicalese, and C. Deppe eds., LNCS
Vol.7777,   
234-246, 3013.
\bibitem{Bl/La}  M. Bloch and J. N. Laneman, On the secrecy capacity of arbitrary wiretap channels, Communication, Control, and Computing, Forty-Sixth Annual Allerton Conference
Allerton House, UIUC,  USA, 818-825, 2008.
\bibitem{tobepublished}H. Boche,   M. Cai and N. Cai, Channel state detecting code for compound quantum  channel, preprint.
\bibitem{Bo/Ca/Ca/De} H. Boche,  M. Cai, N. Cai,  and C. Deppe,  
 Secrecy capacities of compound quantum wiretap channels and applications, 
Phys. Rev. A, Vol. 89, No. 5, 052320,
\tt arXiv:1302.3412\rm, 2014.
\bibitem{Bo/Ca/De} H. Boche,  M. Cai,  C. Deppe, and J. N\"otzel,   
Classical-quantum arbitrarily varying wiretap channel - Ahlswede Dichotomy - positivity
 - resources - super activation, accepted for publication in Quantum Information Processing,
\tt arXiv:1307.8007\rm, 2014.
\bibitem{Bo/No} H. Boche and J. N\"otzel, Arbitrarily small amounts of correlation
for arbitrarily varying quantum channel,  J. Math. Phys., Vol. 54, Issue 11,  \tt arXiv 1301.6063\rm,
2013.
\bibitem{Bo/No2} H. Boche and J. N\"otzel, Positivity, discontinuity, finite resources, 
and nonzero error for arbitrarily varying quantum channels, J. Math. Phys. Vol. 55, 122201, 2014.
\bibitem{Ca/Wi/Ye} N. Cai, A. Winter, and R. W. Yeung, Quantum privacy and
quantum wiretap channels, Problems of Information Transmission, Vol.
40, No. 4,  318-336, 2004.
\bibitem{Cs/Na} I.  Csisz\'ar and  P. Narayan, The capacity of the arbitrarily varying channel revisited: positivity, constraints,
IEEE Trans. Inform. Theory, Vol. 34, No. 2,  181-193, 1988.
\bibitem{De} I. Devetak, The private classical information capacity and
quantum information capacity of a quantum channel, IEEE Trans. Inform. Theory, Vol. 51, No. 1, 44-55,  2005.
\bibitem{Du/Pa} D. D. Dubhashi and A. Panconesi, Concentration of Measure for the Analysis of Randomized Algorithms, Cambridge University Press,
2012.
\bibitem{Rei}T. Ericson, Exponential error bounds for random codes in the arbitrarily varying channel,
IEEE Trans. Inform. Theory, Vol. 31, No. 1, 42-48, 1985.
\bibitem{Fa} M. Fannes, A continuity property of the entropy density for spin lattice systems,
 Comm. Math. Phys., Vol. 31. 291-294, 1973.
\bibitem{Ha/Na} M. Hayashi, H. Nagaoka, General formulas for capacity of classical-
quantum channels, IEEE Trans. Inform. Theory, Vol. 49. No. 7, 1753-1768,
2003.
\bibitem{Ho} A. S. Holevo,  The capacity of quantum channel with general signal states,
 IEEE Trans. Inform. Theory,  Vol. 44, 269-273, 1998.
\bibitem{Le/Sm} D. Leung and G. Smith, Continuity of quantum channel capacities,
Commun. Math. Phys, Vol. 292, No. 1,  201-215, 2009.
\bibitem{Mi/Sch}  V. D. Milman and G. Schechtman, Asymptotic Theory of Finite Dimensional Normed Spaces.
Lecture Notes in Mathematics 1200, Springer-Verlag, corrected second
printing, Berlin, UK, 2001.
\bibitem{Wi/No/Bo2}J. N\"otzel,  M. Wiese,  and H. Boche, 
The Arbitrarily Varying Wiretap Channel --- Secret Randomness, 
Stability and Super-Activation, \tt arXiv:1501.07439\rm, 2015.
\bibitem{Og/Na} T. Ogawa and H. Nagaoka, Making good codes for
classical-quantum channel coding via quantum hypothesis testing,
IEEE Trans. Inform. Theory, Vol. 53, No. 6, 2261-2266, 2007.
\bibitem{Pa} V. Paulsen, Completely Bounded Maps and Operator Algebras,
Cambridge Studies in Advanced Mathematics 78, Cambridge University
Press, Cambridge, UK, 2002.
\bibitem{Sch/Wes} B. Schumacher  and M. D. Westmoreland, Sending classical information via noisy quantum channels,
Phys. Rev., Vol. 56, 131-138, 1997.
\bibitem{Wi/No/Bo}  M. Wiese, J. N\"otzel, and H. Boche, A channel under simultaneous jamming and eavesdropping attack -
 correlated random coding capacities under strong secrecy criteria, accepted for publication in
 IEEE Trans. Inform. Theory,\tt arXiv:1410.8078\rm, 2014.
\bibitem{Wil}M. Wilde,  Quantum Information Theory,
 Cambridge University Press,  2013.
\bibitem{Win} A. Winter, Coding theorem and strong converse for quantum
channels, IEEE Trans. Inform. Theory, Vol. 45, No. 7,  2481-2485,
1999.
\bibitem{Wyn} A. D. Wyner, The wire-tap channel, Bell System Technical
Journal, Vol. 54, No. 8, 1355-1387, 1975.
\bibitem{WynSite} Quantum information problem page of the ITP Hannover, \tt http://qig.itp.uni-hannover.de/qiproblems/11 \rm
\end{thebibliography}
\end{document}